\pdfoutput=1
\documentclass{article}

\usepackage{arxiv}
\usepackage{amsmath,amssymb,amsfonts,amsthm,mathtools}
\usepackage[utf8]{inputenc}
\usepackage[T1]{fontenc}
\usepackage{cite}
\usepackage{url}
\usepackage{booktabs}
\usepackage{array}
\usepackage{microtype}
\usepackage{xcolor}
\usepackage{graphicx}
\usepackage{caption}
\usepackage{algorithm}
\usepackage{algorithmic}
\usepackage{enumitem}
\usepackage[section]{placeins}
\usepackage{hyperref}
\hypersetup{colorlinks=true,
  linkcolor={black!45!blue}, citecolor={black!45!blue}, urlcolor={black!45!blue},
  pdftitle={Rarity Determines Minimax Rates of Semantic Entropy From Weak Pairwise Judgments With a Known Channel},
  pdfauthor={Xiaoli Yu; Jiamiao Liu},
  pdfkeywords={entropy estimation, large language models, minimax rates, stochastic block model}}
\usepackage[capitalize,noabbrev]{cleveref}

\theoremstyle{plain}
\newtheorem{theorem}{Theorem}
\newtheorem{corollary}{Corollary}
\newtheorem{proposition}{Proposition}
\newtheorem{lemma}{Lemma}
\theoremstyle{definition}
\newtheorem{definition}{Definition}
\theoremstyle{remark}
\newtheorem{remark}{Remark}

\newcommand{\Real}{\mathbb{R}}
\newcommand{\Cpx}{\mathbb{C}}
\newcommand{\Exp}{\mathbb{E}}
\newcommand{\Prob}{\mathbb{P}}
\newcommand{\Var}{\mathrm{Var}}
\newcommand{\Cov}{\mathrm{Cov}}
\newcommand{\Bern}{\mathrm{Bern}}
\newcommand{\TV}{\mathrm{TV}}
\newcommand{\Ind}[1]{\mathbf{1}\{#1\}}
\newcommand{\Ent}{H}
\newcommand{\hH}{\widehat{H}}
\newcommand{\Simplex}{\Theta_K}
\newcommand{\Rare}[2]{\Theta_{K,#1}(#2)}
\newcommand{\Mom}{\mathcal{M}_K}
\newcommand{\hT}{\widehat{T}}
\newcommand{\hTs}{\widehat{T}^{\star}}
\newcommand{\tP}{\widetilde{P}}
\newcommand{\htheta}{\widehat{\theta}}
\newcommand{\Trees}{\mathcal{T}}
\newcommand{\tr}{\mathsf{t}}
\newcommand{\SG}{S_{\mathrm{G}}}
\newcommand{\Sw}[1]{\widetilde{S}_{#1}}
\newcommand{\Sobs}{S_{\mathrm{obs}}}
\newcommand{\Pu}{\mathsf{P}}
\newcommand{\DS}{\mathsf{D}}
\newcommand{\coef}[1]{[J^{#1}]}
\newcommand{\eps}{\varepsilon}
\newcommand{\pbar}{\bar p}
\newcommand{\Lam}{\Lambda}
\newcommand{\Poly}{\mathrm{Poly}}
\newcommand{\Rem}{\mathrm{Rem}}
\newcommand{\Rloc}{R_{\mathrm{loc}}}
\newcommand{\Chc}{C_{\mathrm{hc}}}

\graphicspath{{figures/}}

\renewcommand{\shorttitle}{Rarity Determines Minimax Rates of Semantic Entropy}

\title{Rarity Determines Minimax Rates of Semantic Entropy From Weak Pairwise Judgments With a Known Channel%
\thanks{This work has been submitted to the IEEE for possible publication. Copyright may be transferred without notice,
after which this version may no longer be accessible.}}

\author{%
  Xiaoli Yu \\
  School of Cyber Security and Information Law \\
  Chongqing University of Posts and Telecommunications \\
  Chongqing 400065, China
  \And
  Jiamiao Liu\thanks{Corresponding author.} \\
  Department of Information, Xinqiao Hospital \\
  Army Medical University (Third Military Medical University) \\
  Chongqing 400037, China \\
  \texttt{ljm6046@163.com}
}

\date{}

\begin{document}

\maketitle

\begin{abstract}
The semantic entropy of a language model is the entropy of the meanings of its sampled answers. The meanings are
never observed, and the entropy is estimated from the verdicts of a pairwise equivalence judge. Given the meanings of
$n$ answers, a known channel declares each pair equivalent independently, with probability $q+\Delta$ within a meaning
and $q$ across meanings. At most $K$ meanings occur, and $q$ is bounded away from $0$ and $1$. For a weak judge, with
$n\Delta^2\to0$, the global minimax rate is set by distributions near a vertex of the simplex and overstates the
difficulty elsewhere. We show that the difficulty is set by how many meanings can be rare. Over distributions with at
most $s$ meanings below a fixed probability $\tau<1/(K-s)$, the minimax squared risk is of order $N^{-1/s}$ for $s\ge2$,
$N^{-1}\log^2(e+N)$ for $s=1$ and $N^{-1}$ for $s=0$, capped at one, where $N=n^K\Delta^{2K-2}$. The matching lower
bounds hold on every fixed polynomial window $n\Delta^2\le n^{-\kappa}$. One estimator attains these rates without
knowing $s$ or $\tau$: it averages products of centered verdicts along all spanning trees of every $k$-set of answers,
projects the resulting power sums onto the valid moment set and inverts them. The case $s=K-1$ gives the global rate
$N^{-1/(K-1)}$ and the consistency threshold: with $\Delta=n^{-\gamma}$ the entropy is uniformly consistently estimable
exactly when $\gamma<K/(2K-2)$. On polynomial windows, the same estimator attains the sharp local minimax constant on
shrinking neighborhoods of interior points with distinct probabilities. Restricting the average to stars
multiplies the constant by $K^{K-3}$.

\end{abstract}


\section{Introduction}
\label{sec:intro}

Semantic entropy scores a language model's uncertainty about a prompt by the entropy of the
distribution of \emph{meanings} among its sampled answers~\cite{kuhn2023semantic,farquhar2024detecting}. The meanings
are never observed: an equivalence judge compares the answers in pairs, and the entropy is estimated from its
verdicts. This paper determines what makes the entropy hard to estimate when the judge is weak.

We model the judge as a known binary channel. Each of $n$ answers carries a meaning, and the meanings are independent
draws from an unknown distribution $\theta$ on at most $K$ meanings. Given the meanings, the judge declares two
answers equivalent with probability $q+\Delta$ if they share a meaning and with probability $q$ otherwise, independently
across pairs. The verdicts thus form a stochastic block model on the complete graph, with $K$
communities of proportions $\theta$, edge probability $q+\Delta$ within communities and $q$ between them. The target is the
entropy $\Ent(\theta)$ of the proportions, not the partition.

We study the weak regime $n\Delta^2\to0$, in which the
$n-1$ verdicts on one answer have a signal-to-noise ratio of order $n\Delta^2/(q(1-q))$ about its meaning. For at most three
meanings, $q$ bounded away from $0$ and $1$, a fixed distribution with no meaning of probability one, and $\Delta=n^{-\gamma}$
with $\gamma>1/2$, no estimator of an answer's meaning has asymptotically smaller error probability than guessing the
most likely meaning, even given the meanings of all other answers~\cite{anonymous2026three}. The entropy can still be
estimated consistently in part of this regime. Concurrent analyses of at most three~\cite{anonymous2026three} and at most four
meanings~\cite{anonymous2026four} give global minimax rates, set by a least-favorable family that approaches a vertex
of the simplex. In that family one meaning, called macroscopic, has probability bounded away from zero, and all others
are rare. These rates do not distinguish distributions by how many meanings are rare, and for $K\ge5$ the rate was open.

We show that the difficulty is set by \emph{how many meanings can be rare}. For $0\le s\le K-1$ and a fixed
$\tau\in(0,1/(K-s))$, let $\Rare s\tau$ be the distributions on $K$ meanings with at most $s$ probabilities below $\tau$. With
$N:=n^K\Delta^{2K-2}$, the minimax squared risk over $\Rare s\tau$ is of order $\psi_s(N)$, that is, $\min\{1,N^{-1/s}\}$ for $s\ge2$,
$\min\{1,N^{-1}\log^2(e+N)\}$ for $s=1$ and $\min\{1,N^{-1}\}$ for $s=0$, on every fixed polynomial window
$n\Delta^2\le n^{-\kappa}$. One estimator attains all of these rates without knowing $s$ or $\tau$. These
windows cover $\Delta=n^{-\gamma}$ for every fixed $\gamma>1/2$, which is the whole weak regime in this parametrization. For
$s=K-1$ and $\tau\le1/K$ the class is the whole simplex, which gives the global rate $\min\{1,N^{-1/(K-1)}\}$ for every $K$ and
the consistency threshold $\gamma_c(K)=K/(2K-2)$ conjectured in~\cite{anonymous2026four}. The classes
share this threshold and differ in the rate.

\begin{table}[t]
  \centering
  \caption{Minimax Squared Risk of the Semantic Entropy From a Known Pairwise Judge}
  \label{tab:comparison}
  \footnotesize
  \begin{tabular}{@{}>{\raggedright\arraybackslash}p{0.34\textwidth}>{\raggedright\arraybackslash}p{0.25\textwidth}>{\raggedright\arraybackslash}p{0.31\textwidth}@{}}
\toprule
setting & minimax squared risk & matching lower bound proved on \\
\midrule
at most $3$ meanings, all distributions \cite{anonymous2026three} & $\min\{1,(n^3\Delta^4)^{-1/2}\}$ & the whole window $n\Delta^2\le c_1$ \\
at most $4$ meanings, all distributions \cite{anonymous2026four} & $\min\{1,(n^4\Delta^6)^{-1/3}\}$ & every fixed window $n\Delta^2\le n^{-\kappa}$ \\
$K$ meanings, at most $s\ge2$ rare & $\min\{1,N^{-1/s}\}$ & every fixed window $n\Delta^2\le n^{-\kappa}$ \\
$K$ meanings, at most one rare & $\min\{1,N^{-1}\log^2(e+N)\}$ & every fixed window $n\Delta^2\le n^{-\kappa}$ \\
$K$ meanings, none rare & $\min\{1,N^{-1}\}$ & every fixed window $n\Delta^2\le n^{-\kappa}$ \\
shrinking $\ell_2$ ball of radius $h_n$ about an interior $\theta_0$ with distinct coordinates & $(1+o(1))\,\mathcal V(\theta_0)/N$ & $n\Delta^2\le n^{-\kappa}$, $\kappa<\frac1{K-1}$, $N\to\infty$, $h_n\to0$, $h_n\sqrt N\to\infty$ \\
\bottomrule
\end{tabular}

  \par\smallskip
  \begin{minipage}{0.92\textwidth}\footnotesize
    The channel parameter $q\in[q_0,1-q_0]$, the rarity threshold $\tau$ and the window exponent $\kappa>0$ are fixed, and
    $N=n^K\Delta^{2K-2}$. The first five rows hold up to multiplicative constants. The first two rows report concurrent
    work, and $c_1>0$ in the first row is a constant. The last row, for fixed $q\in(0,1)$, gives the local asymptotic
    minimax constant $\mathcal V(\theta_0)=(K!/K^K)\Lam(\theta_0)^2(q(1-q))^{K-1}$, where
    $\Lam(\theta_0)=\int_0^\infty\mathrm dt/\prod_c(t+\theta_{0,c})$.
  \end{minipage}
\end{table}

\begin{figure}[t]
  \centering
  \includegraphics{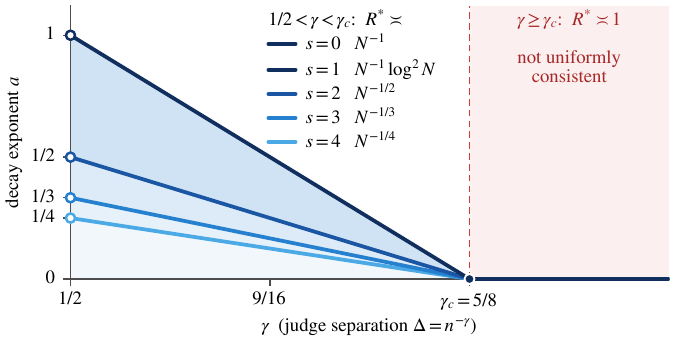}
  \caption{Polynomial decay exponent $a$ of the minimax squared risk for the semantic entropy,
  $R^*=n^{-a+o(1)}$, when
  $\Delta=n^{-\gamma}$ and $q\in[q_0,1-q_0]$ is fixed, for $K=5$ possible meanings of which at most $s$
  have probability below a fixed $\tau\in(0,1/5)$. Here $N=n^5\Delta^8$. The classes
  $s=0$ and $s=1$ have the same exponent and share the darkest curve, and for $s=1$ the risk carries an extra
  factor $\log^2n$ wherever $a>0$. The rates are matched for every
  fixed $\gamma>1/2$, and open circles mark the excluded endpoint $\gamma=1/2$, where $n\Delta^2=1$. Every class loses
  uniform consistency at the same $\gamma_c=K/(2K-2)=5/8$.}
  \label{fig:hero}
\end{figure}

\subsection{Main Contributions}

The contributions of this paper are as follows.
\begin{enumerate}[label=\arabic*),leftmargin=*,itemsep=1pt,topsep=2pt]
  \item \emph{Minimax rates graded by rarity.} Matching upper and lower bounds of order $\psi_s(N)$ hold on every fixed
    polynomial window, for every $K\ge3$ and every class, and one all-tree moment estimator attains these rates adaptively.
    Over the whole simplex this gives the rate $\min\{1,N^{-1/(K-1)}\}$ for every $K$. With no window condition, the
    entropy is uniformly consistently estimable over every class exactly when $N\to\infty$, that is, when
    $\gamma<K/(2K-2)$ for $\Delta=n^{-\gamma}$. Replacing the exact projection by a grid search keeps the rates and
    makes the estimator polynomial-time for every fixed $K$. For $K\le5$ it needs one $n\times n$ matrix product and
    $O(n^2)$ further operations.
  \item \emph{An information inequality that reaches the mixed faces.} For finite-order sweeps of the likelihood score in
    weak pairwise experiments we prove an exact identity for small motifs, a uniform coefficient bound for all motifs, a
    pairing bound for overlapping motifs and a bound on the Taylor remainder. No step needs a single
    macroscopic meaning, so the inequality reaches the mixed faces of the simplex, which have two or more positive
    coordinates, and with them the classes $s<K-1$.
  \item \emph{A one-sided local transfer.} For $\theta\in\Rare s\tau$ and every $\theta'$ in the simplex,
    $|\Ent(\theta)-\Ent(\theta')|$ is at most a constant times $D_e^{1/s}$ for $s\ge2$, $D_e\log(e/D_e)$ for $s=1$ and $D_e$ for
    $s=0$, where $D_e$ is the $\ell_1$ distance between the coefficients of $\prod_c(z-\theta_c)$ and $\prod_c(z-\theta'_c)$.
    As $\theta'$ is unrestricted, the bound applies to the projected estimate even when it leaves
    the class, and this is what lets one estimator adapt to $s$ and $\tau$. For $s\ge2$ the exponent $1/s$ cannot be
    improved.
  \item \emph{The sharp local constant.} On every fixed polynomial window with $N\to\infty$ and at interior points with distinct
    probabilities, $N$ times the local minimax risk converges to an explicit constant, which the same estimator attains.
    Along the interior fiber an information equality holds. Averaging over stars instead of all trees multiplies the
    constant by $K^{K-3}$.
\end{enumerate}

\subsection{Proof Techniques}

The estimator averages, over all spanning trees of every set of $k$ answers, the product of the centered verdicts
$(A_{ij}-q)/\Delta$ along the edges of the tree, where $A_{ij}\in\{0,1\}$ is the verdict on answers $i$ and $j$. These averages are unbiased for the power sums $P_k=\sum_c\theta_c^k$, and the
estimator projects them onto the valid moment set and solves for the distribution. Two facts explain the hierarchy. First, the last power sum is the
statistical bottleneck: with $\delta:=n\Delta^2$, $P_k$ is estimated at the scale $N_k:=n\delta^{k-1}$, and $N_k/N_K\to\infty$ for $k<K$. Second, the
entropy, as a function of the coefficients of $\prod_c(z-\theta_c)$, has a local modulus $x^{1/s}$ for $s\ge2$, set by the
number of roots that can approach zero, whatever the other roots do. The upper bound rests on a one-sided form of this
modulus: the projected estimate may leave the class, and a Rouch\'e argument returns it to $\Rare s{\tau/2}$.

The lower bounds use fibers that move only the constant term of $\prod_c(z-\theta_c)$: they keep $P_2,\dots,P_{K-1}$ fixed, and
for $s\ge2$ they move $P_K$ by order $\eps^s$ while the entropy moves by order $\eps$. We bound the Fisher information of the
verdict graph along such a fiber through a finite-order sweep of the likelihood score. For four meanings the
concurrent analysis~\cite{anonymous2026four} uses majority-vertex counting, which gains a factor of the probability
scale of the rare meanings for every vertex labeled with the single macroscopic meaning, and so needs such a meaning.
Our information inequality needs none: Cauchy estimates on a complex polydisc, used in~\cite{anonymous2026four} to
prove this factor, here bound every coefficient of the sweep uniformly.

At an interior point $\theta_0$ with distinct probabilities, $N$ times the local minimax risk on shrinking neighborhoods tends to
$(K!/K^K)\Lam(\theta_0)^2(q(1-q))^{K-1}$, where $\Lam$ is the derivative of the entropy in $e_K:=\prod_c\theta_c$ with the other
coefficients of $\prod_c(z-\theta_c)$ fixed. Along the fiber that moves only $e_K$, the sweep of order $K$ bounds the Fisher
information, up to a factor $1+o(1)$, by the contribution of the spanning trees of $K$ answers, each paired with itself.
The van Trees inequality then gives the lower bound. A delta method controlled by fourth moments gives the matching
upper bound. A star average, which
uses $K$ of the $K^{K-2}$ spanning trees of each $K$-set, loses the factor $K^{K-3}$.

\subsection{Related Work}
\label{sec:related}

\subsubsection{Concurrent Analyses}
For at most three meanings and $q$ bounded away from $0$ and $1$, the entropy is uniformly consistently estimable for
$\Delta=n^{-\gamma}$ with $\gamma<3/4$, and the minimax squared risk is of order $\min\{1,(n^3\Delta^4)^{-1/2}\}$ on the whole window
$n\Delta^2\le c_1$~\cite{anonymous2026three}. With at most four meanings the rate is $\min\{1,(n^4\Delta^6)^{-1/3}\}$ on every fixed
polynomial window, and the threshold $\gamma_c(K)=K/(2K-2)$ was conjectured for $K$ meanings~\cite{anonymous2026four}. On the
complete graph, the rarity-graded rates contain this rate and threshold as the case $K=4$, $s=3$, and give the three-meaning
rate on polynomial windows. We prove the conjectured threshold for every $K$. \Cref{tab:comparison} lists these rates
beside ours, and \Cref{fig:hero} shows for $K=5$ that all classes share the threshold.

We share with the two
concurrent analyses the model, the moment-projection approach and the sweep identity, and with the four-meaning analysis
the pairing of overlapping motifs and the remainder framework, which we establish for general $K$. Both analyses also treat channel calibration and
query design.
Our proofs differ from theirs in the uniform coefficient bound, which replaces majority-vertex counting and makes the mixed faces,
and with them the classes $s<K-1$, accessible; in the one-sided transfer behind adaptation; and in the sharp constant.

\subsubsection{Semantic Entropy and Noisy Judges}
Semantic entropy was introduced with bidirectional-entailment clustering~\cite{kuhn2023semantic}, and its discrete form,
the entropy of the cluster proportions, detects hallucinations~\cite{farquhar2024detecting}. Later work estimates the
same cluster entropy from fewer samples~\cite{ciosek2025hallucination}, corrects it for meanings not yet
observed~\cite{mccabe2026estimating,mccabe2026seneca}, or changes the uncertainty functional, through a semantic
kernel~\cite{nikitin2024kernel}, pairwise similarities~\cite{nguyen2025beyond} or the degrees and Laplacian spectrum of a
similarity graph~\cite{lin2024generating}. These papers do not derive the minimax entropy risk under a noisy equality channel.
Liu~\cite{liu2026statistically} analyzes the same observation model with a spectral-clustering plug-in, whose guarantee
applies where clustering errors are small. The all-tree estimator does not cluster the answers, and its upper bound holds
for every channel, the weak regime included.

\subsubsection{Entropy Estimation and Block Models}
With observed labels, the entropy on a fixed alphabet has minimax rate $1/n$, and on large alphabets its minimax rates
follow from best polynomial approximation~\cite{wu2016minimax,jiao2015minimax}. Seen through a weak judge, the class $s=0$
has the rate $\min\{1,1/N\}$ on every fixed polynomial window, so the judge replaces $n$ by $N$, and each additional rare
meaning slows the rate further.

For block models and related observation models, the information-theoretic literature
characterizes the recovery of the labels themselves: exact recovery of two communities has a sharp
threshold~\cite{abbe2016exact}. For pairwise difference measurements passed through general channels, the
feasibility of exact recovery is governed by channel divergence measures and the minimum cut of the measurement
graph~\cite{chen2016information}. In the sparse
two-community planted partition model with edge probabilities $a/n$ and $b/n$, neither a partition correlated with the
truth nor consistent estimates of $(a,b)$ exist when $(a-b)^2<2(a+b)$~\cite{mossel2015reconstruction}.

When the
expected degree grows faster than $\log n$ and no two rows of the matrix of edge probabilities coincide, the maximum
likelihood and variational estimators of block-model parameters are asymptotically
normal~\cite{bickel2013asymptotic}. That result is pointwise. Measured uniformly, estimation is much
slower near models whose classes are nearly alike: in the two-class affiliation model with equal proportions and edge
probabilities symmetric about $1/2$, the minimax squared risk for the gap between the within-class and between-class
probabilities is of order $1/n$ rather than $1/n^2$, and lower bounds of this order hold locally around every block
model~\cite{castillo2022uniform}.

Our verdict graph is an affiliation model near the same boundary with the roles
reversed: the gap $\Delta\to0$ is known, and the proportions of the meanings are unknown. On every fixed polynomial window
$n\Delta^2\le n^{-\kappa}$ the rates are set by $N=n^K\Delta^{2K-2}$ rather than by $n$. They hold uniformly over each
rarity class, and the sharp constant holds uniformly over shrinking neighborhoods. The
consistency threshold needs no window: the entropy is uniformly consistently estimable exactly when $N\to\infty$.

\subsubsection{Subgraph Statistics and Moments}
Noise-corrected subgraph products~\cite{chang2022estimation} give our tree statistics, and moments of small subgraphs
have been used to estimate block-model parameters~\cite{bickel2011method,ambroise2012new}. Bravo-Hermsdorff et al.~\cite{bravohermsdorff2022private}
estimate the collision probability $P_2$ from equality indicators of disjoint pairs of samples, each passed through a
channel of the same form as ours. Spectrum estimation also inverts unbiased moment estimates: the moments of a covariance
spectrum are traces of matrix powers, which Kong and Valiant~\cite{kong2017spectrum} estimate without bias by averaging
products of sample inner products along cycles of distinct samples. For equality indicators the product along any
connected graph on $k$ answers is the indicator that they share a meaning, and spanning trees are the connected graphs
with the fewest edges, $k-1$, so each tree product carries the fewest noise factors.

\subsubsection{Moment Hierarchies at Low Signal-to-Noise Ratio}
In i.i.d.\ Gaussian latent-variable models at low signal-to-noise ratio (SNR), the information about the parameter is
graded by moment order. For Gaussian mixtures with a common covariance, Katsevich and Bandeira~\cite{katsevich2023likelihood}
expand the population log-likelihood in powers of the SNR. They show that its maximization reduces to a sequence of
least-squares problems that match the moments one order at a time. In the discrete orbit recovery model at high noise,
the Fisher information at a generic signal has a graded eigenvalue structure determined by the invariant polynomials of
the group, and the maximum likelihood estimator of an invariant of degree $k$ has asymptotic variance
$O(\mathrm{SNR}^{-k}/n)$~\cite{fan2023likelihood}. In Gaussian latent-variable models with a known latent distribution,
optimally weighted moments of the minimal identifying orders match the first-order efficiency of maximum likelihood,
and different directions become informative at different moment orders~\cite{balanov2026generalized}.

These results take
$n\to\infty$ at a fixed noise level and send the SNR to zero afterwards. The joint regime, in which the SNR tends to zero
with $n$, is identified as open in~\cite{balanov2026generalized}. Our verdicts are dependent through the shared labels, and
$\Delta=\Delta_n$ tends to zero with $n$. The power sum $P_k$, an invariant of degree $k$ under relabeling, is estimated at
the scale $N_k=n(n\Delta^2)^{k-1}$, so $n\Delta^2$ takes the place of the SNR in the grading. The least-favorable direction
moves only the last power sum, and the unweighted all-tree estimator attains the local minimax constant in this joint
limit. Our rarity hierarchy has a different source: every class has the scale $N$
of the last power sum, and the number of coordinates that can vanish sets the modulus with which the entropy responds
to that power sum.

\subsubsection{Local Rates Governed by Root Structure}
Hundrieser et al.~\cite{hundrieser2025local} quantify the local minimax risk of recovering an atomic signal from a
Poisson convolution as a function of the clustering structure of its atoms, under a multiscale loss, with estimators
that adapt to that structure. Here the target is a functional rather than the atoms: only roots that approach zero
create the singularity of the entropy, collisions of positive roots do not affect the transfer, and the matching
information bound concerns a dependent graph experiment.

For finite mixtures, Wu and Yang~\cite{wu2020optimal} develop
denoised moment methods, Heinrich and Kahn~\cite{heinrich2018strong} establish local minimax rates, and Le et
al.~\cite{le2026characterizing} characterize configuration-dependent rates of the maximum likelihood estimator. Wu and
Yang~\cite{wu2020optimal} obtain their rates from moment comparison inequalities, which bound the distance between
mixing distributions by the difference of their moments. The one-sided local transfer plays this role for the entropy.
The fibers that move only the constant term show that the exponent of the transfer cannot be improved: along them
$s$ roots within $\eps$ of zero change the coefficients of $\prod_c(z-\theta_c)$ only by order $\eps^s$, while the entropy
moves by order $\eps$. H\"older continuity of the roots of a polynomial in its
coefficients~\cite{parusinski2025perturbation} is the geometric background of the exponent $1/s$. We prove the transfer
directly from the integral representation of Jozsa and Mitchison~\cite{jozsa2015symmetric}.

\subsubsection{Lower Bounds}
Low-degree methods give computational lower bounds for estimation~\cite{schramm2022computational}, including for
distinguishing planted models~\cite{rush2026easier}. Our bounds are information-theoretic: they control the Fisher
information of the full likelihood, a sum over all $K^n$ labelings, through sweeps whose conditional mean given the
verdicts is the observed score. The sharp constant uses the van Trees inequality~\cite{gill1995applications}, and its
attainment uses a fourth-moment bound of the hypercontractive type of Bonami~\cite{bonami1970etude}.

\subsection{Notation}
\label{sec:notation}

We write $[K]:=\{1,\dots,K\}$ and $\Simplex:=\{\theta\in[0,1]^K:\sum_c\theta_c=1\}$. For $\theta\in\Simplex$, $P_k=P_k(\theta):=\sum_c\theta_c^k$ are
the power sums, $P(\theta):=(P_2,\dots,P_K)$, and $e_j=e_j(\theta)$ are the elementary symmetric polynomials, with $e_0=1$ and
$e_1=1$ on $\Simplex$;
\[
  p_\theta(z):=\prod_c(z-\theta_c)=\sum_{j=0}^K(-1)^je_j(\theta)z^{K-j},\qquad
  \pbar_\theta(t):=\prod_c(t+\theta_c)=\sum_{j=0}^Ke_j(\theta)t^{K-j}.
\]
The entropy is $\Ent(\theta):=-\sum_c\theta_c\log\theta_c$, with the natural logarithm and $0\log0:=0$. We write $\Ind{\cdot}$ for an
indicator, $\Bern(r)$ for the Bernoulli law with mean $r$, $\TV$ for the total variation distance, $\|\cdot\|$ for the
Euclidean norm, $\|\cdot\|_1$ for the $\ell_1$ norm and $\mathbf e_c\in\Real^K$ for the $c$-th basis vector. The letters $C$ and $c$ denote
positive constants whose dependence is stated in each result and whose values may change from line to line. The constants
implied by $O(\cdot)$, $\lesssim$ and $\asymp$ depend only on the parameters listed in the statement at hand and, unless stated
otherwise, are uniform in $n$, $\Delta$, $q$ in the stated channel range, the labels, the path coordinate $u$ and the amplitude
of a fiber. An asymptotic statement in $n$ holds along every sequence of admissible parameters and uniformly over the
parameter set stated with it, and its rates of convergence are those displayed.

\section{Problem Formulation}
\label{sec:formulation}

\subsection{Observation Model}

A language model is sampled $n$ times on one prompt. Answer $i$ has a meaning $Z_i\in[K]$, and $Z_1,\dots,Z_n$ are i.i.d.\ from
$\theta\in\Simplex$, where $K\ge3$ is a known upper bound on the number of meanings and zero coordinates are allowed. The target is
the entropy $\Ent(\theta)$. A judge compares every pair of answers and reports $A_{ij}\in\{0,1\}$. Given the meanings, the verdicts
are independent, with
\begin{equation}
  \label{eq:model}
  A_{ij}\mid Z\ \sim\ \Bern\big(q+\Delta\,\Ind{Z_i=Z_j}\big),\qquad 1\le i<j\le n .
\end{equation}
The channel $(q,\Delta)$, $0\le q<q+\Delta\le1$, is known: it can be calibrated from pairs of answers with known meanings,
at no cost in the rate for at most three meanings~\cite{anonymous2026three}. Each verdict models a separate query of the
judge on one pair. We write
\[
  p:=q+\Delta,\qquad W_{ij}:=\Ind{Z_i=Z_j},\qquad \delta:=n\Delta^2,\qquad N_k:=n\delta^{k-1},\qquad N:=N_K=n^K\Delta^{2K-2}.
\]

\subsection{Minimax Risk and Rarity Classes}

For $\Theta\subseteq\Simplex$ the minimax risk is
\[
  R^*(\Theta):=\inf_{\hH}\sup_{\theta\in\Theta}\Exp_\theta\big(\hH-\Ent(\theta)\big)^2,
\]
the infimum running over all estimators $\hH=\hH(A,\Xi)$ with $\Xi$ auxiliary randomness independent of $(Z,A)$. For $0\le s\le K-1$
and $0<\tau<1/(K-s)$, let
\[
  \Rare s\tau:=\{\theta\in\Simplex:\ \text{at least }K-s\text{ coordinates of }\theta\text{ are at least }\tau\},
\]
the distributions in which at most $s$ meanings are rare. Zero coordinates count as rare, and $\tau$ does not depend on $n$.
At $\tau=1/(K-s)$ every member would have $K-s$ coordinates equal to $1/(K-s)$ and entropy $\log(K-s)$, so the restriction
$\tau<1/(K-s)$ excludes only that degenerate class. For $\tau\le1/K$, $\Rare{K-1}\tau=\Simplex$. \Cref{fig:classes} shows
the classes for $K=3$. The rates are
\begin{gather*}
  \psi_s(x):=\min\{1,x^{-1/s}\}\quad(2\le s\le K-1),\\
  \psi_1(x):=\min\{1,x^{-1}\log^2(e+x)\},\qquad
  \psi_0(x):=\min\{1,x^{-1}\}.
\end{gather*}

\begin{figure}[tb]
  \centering
  \includegraphics{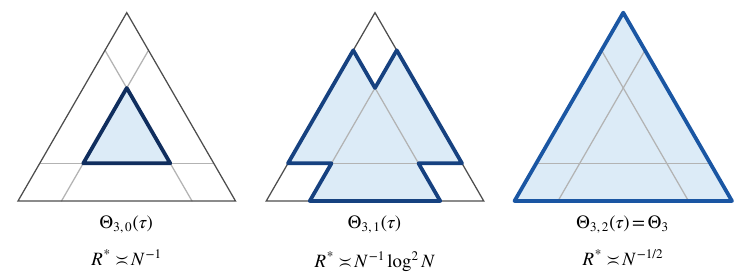}
  \caption{Rarity classes $\Theta_{3,s}(\tau)$ for $K=3$ possible meanings and $\tau=0.2$, in the simplex $\Theta_3$ drawn
  as a triangle whose vertices are the distributions concentrated on one meaning, with gray lines at $\theta_c=\tau$.
  The shaded class $\Theta_{3,s}(\tau)$ holds the distributions with at least $3-s$ coordinates of at least $\tau$, so
  $\Theta_{3,0}(\tau)\subset\Theta_{3,1}(\tau)\subset\Theta_{3,2}(\tau)=\Theta_3$. Under each class is the order of its
  minimax squared risk for large $N=n^3\Delta^4$ under the conditions of \Cref{thm:main}(ii).}
  \label{fig:classes}
\end{figure}

\subsection{The All-Tree Moment Estimator}
\label{sec:estimator}

Let $X_{ij}:=(A_{ij}-q)/\Delta$. Given the meanings, distinct pairs are independent and $\Exp[X_{ij}\mid Z]=W_{ij}$, so the product of
the $X_e$ over the edges of any tree on $k$ answers has mean $P_k$, the probability that $k$ answers share a meaning. Such
noise-corrected subgraph products are those of Chang et al.~\cite{chang2022estimation}. For $2\le k\le K$ let $\Trees_k$ be the set of labeled spanning
trees of the $k$-subsets of $[n]$, so that $|\Trees_k|=\binom nkk^{k-2}$ by Cayley's formula (\Cref{fig:trees} shows the
$16$ trees of a $4$-set). The all-tree average
\begin{equation}
  \label{eq:alltree}
  \hT_k:=\frac1{|\Trees_k|}\sum_{\tr\in\Trees_k}\ \prod_{e\in\tr}X_e
\end{equation}
is therefore unbiased for $P_k$. Newton's identities recover the elementary symmetric polynomials $e_1,\dots,e_K$ of $\theta$,
and hence its coordinates up to permutation, from $(P_2,\dots,P_K)$, since $P_1=1$. The estimator projects
$\hT:=(\hT_2,\dots,\hT_K)$ onto the moment set $\Mom:=\{P(\theta):\theta\in\Simplex\}$ and inverts the moment map
(\Cref{alg:main}).

\begin{algorithm}[tb]
  \caption{All-Tree Moment Estimator}
  \label{alg:main}
  \begin{algorithmic}[1]
    \REQUIRE verdicts $(A_{ij})_{i<j}$, channel $(q,\Delta)$, number $K$ of possible meanings
    \STATE $X_{ij}\leftarrow(A_{ij}-q)/\Delta$ for $i\ne j$
    \STATE $\hT_k\leftarrow\big(\binom nkk^{k-2}\big)^{-1}\sum_{|V|=k}\ \sum_{\tr\text{ spanning tree of }V}\ \prod_{e\in\tr}X_e$ for
      $k=2,\dots,K$
    \STATE $\tP\leftarrow$ a Euclidean nearest point of $(\hT_2,\dots,\hT_K)$ in $\Mom$, chosen by a fixed rule
    \STATE $\htheta\leftarrow$ the roots of $\sum_{j=0}^K(-1)^j\tilde e_jz^{K-j}$, with $\tilde e_j$ computed from $\tP$ by Newton's identities
    \RETURN $\hH:=\Ent(\htheta)$
  \end{algorithmic}
\end{algorithm}

Every $k$-set carries exactly $k$ stars, with distinct centers, and replacing the average over all spanning
trees of a $k$-set by the average over its stars gives the star average
\begin{equation}
  \label{eq:star}
  \hTs_k:=\frac1{k\binom nk}\sum_{i=1}^n\ \sum_{S\subseteq[n]\setminus\{i\},\,|S|=k-1}\ \prod_{j\in S}X_{ij},
\end{equation}
which is also unbiased for $P_k$. We write $\hH^\star$ for the output of \Cref{alg:main} with $\hT$ replaced by
$\hTs:=(\hTs_2,\dots,\hTs_K)$.

\begin{figure}[tb]
  \centering
  \includegraphics{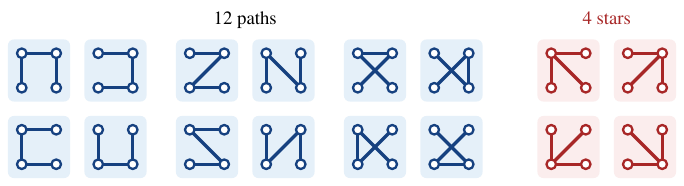}
  \caption{The $16=4^{4-2}$ labeled spanning trees of a set of four answers placed at the corners of a square: $12$
  paths, grouped by shape, and $4$ stars. Given the meanings, the product of the centered verdicts $X_e$ along the
  edges of any of them has conditional mean $1$ if the four answers share a meaning, and $0$ otherwise. The all-tree
  average averages these products over all trees of all $4$-sets, and the star average
  over the stars alone. For $K=4$, restricting to the stars multiplies the leading variance of $\hT_4$ by
  $K^{K-3}=4$.}
  \label{fig:trees}
\end{figure}

\section{Main Results}
\label{sec:results}

\subsection{Minimax Rates Graded by Rarity}

\begin{theorem}[Rarity-Graded Minimax Rates]
  \label{thm:main}
  Fix $K\ge3$, $0\le s\le K-1$ and $0<\tau<1/(K-s)$.
  \begin{enumerate}[label=(\roman*),leftmargin=*,itemsep=1pt,topsep=2pt]
    \item There is $C_U$ depending only on $(K,s,\tau)$ such that for every $n\ge2K$, every channel
      $0\le q<q+\Delta\le1$ and every $\theta\in\Rare s\tau$, \Cref{alg:main} satisfies
      $\Exp_\theta(\hH-\Ent(\theta))^2\le C_U\,\psi_s(\min\{n,N\})$; if $n\Delta^2\le1$, this is $C_U\,\psi_s(N)$.
    \item Fix also $q_0\in(0,1/2]$ and $\kappa\in(0,1)$. There are $c_L>0$ and $n_L$ depending only on
      $(K,s,\tau,q_0,\kappa)$ such that for all $n\ge n_L$, $q\in[q_0,1-q_0]$ and $\Delta>0$ with
      $n\Delta^2\le n^{-\kappa}$, $R^*(\Rare s\tau)\ge c_L\,\psi_s(N)$.
  \end{enumerate}
\end{theorem}

The estimator uses only $(K,q,\Delta)$, so part~(i) holds for every class at once: the same estimator adapts to the unknown
number of rare meanings and to $\tau$. The upper bound needs no lower bound on $q$ and no window, while the lower bound holds
on every fixed polynomial window, with constants that depend on $\kappa$.

\begin{corollary}[Global Rate and Threshold]
  \label{cor:global}
  \begin{enumerate}[label=(\roman*),leftmargin=*,itemsep=1pt,topsep=2pt]
    \item For $\tau\le1/K$, $\Rare{K-1}\tau=\Simplex$, and under the conditions of \Cref{thm:main}(ii),
      $R^*(\Simplex)\asymp\min\{1,N^{-1/(K-1)}\}$, with constants depending only on $(K,q_0,\kappa)$.
    \item Let $q_n\in[q_0,1-q_0]$ and $\Delta_n>0$ with $q_n+\Delta_n\le1$. Estimators $\hH_n$ with
      $\sup_{\theta\in\Rare s\tau}\Exp_\theta(\hH_n-\Ent(\theta))^2\to0$ exist if and only if $n^K\Delta_n^{2K-2}\to\infty$. For
      $\Delta_n=n^{-\gamma}$ this holds exactly when $\gamma<K/(2K-2)$, for every $s$ and $\tau$.
  \end{enumerate}
\end{corollary}

Part~(ii) needs no window condition: if $n^K\Delta_n^{2K-2}$ stays bounded along a subsequence, then eventually
$n\Delta_n^2\le n^{-1/(2K-2)}$ and \Cref{thm:main}(ii) applies. For $K=3$ and $K=4$ the threshold $K/(2K-2)$ equals the values $3/4$ and $2/3$ of the concurrent
analyses, and it decreases to $1/2$ as $K$ grows. \Cref{fig:hero} shows the polynomial decay exponents of the rates
for $K=5$, and \Cref{tab:comparison} sets the rates beside the concurrent results and the sharp local constant.

\subsection{How Rarity Sets the Rate}

\begin{remark}[Why Every Class Has the Same $N$]
  \label{rem:same-N}
  The power sum $P_k$ is estimated at the scale $N_k$ (\Cref{lem:alltree-var}), and $N_k/N_K=\delta^{k-K}\to\infty$ for $k<K$ when
  $\delta\to0$, so the last power sum is the bottleneck. The least-favorable fibers keep $P_2,\dots,P_{K-1}$
  fixed and move $P_K$. The number of coordinates that can vanish sets the modulus with which the entropy responds to that
  last moment, and hence the rate. The threshold is where $N$ stops growing, which is the same for every class.
\end{remark}

\begin{remark}[Where the Exponent Comes From]
  \label{rem:exponent}
  The entropy admits the representation $\Ent(\theta)-\Ent(\theta')=\int_0^\infty\log(\pbar_\theta(t)/\pbar_{\theta'}(t))\,\mathrm dt$ of Jozsa and
  Mitchison~\cite{jozsa2015symmetric}. On $\Rare s\tau$, $\pbar_\theta(t)\ge\tau^{K-s}t^s$ for $t\le1$, and integrating gives
  the one-sided local transfer of \Cref{lem:transfer}: $|\Ent(\theta)-\Ent(\theta')|\le C_T\,\phi_s(D_e(\theta,\theta'))$ for $\theta\in\Rare s\tau$ and every
  $\theta'\in\Simplex$, where $D_e$ is the $\ell_1$ distance between the vectors $(e_2,\dots,e_K)$ of the two points and $\phi_s(x)=x^{1/s}$
  for $s\ge2$. For $s\ge2$ the exponent $1/s$ cannot be improved on $\Rare s\tau$: the endpoints of the mixed fibers have $D_e\le\eps^s$ while their entropies differ by at least a constant multiple of $\eps$
  (\Cref{rem:sharp-exponent}). \Cref{fig:modulus} shows how the entropy responds to $D_e$ along the least-favorable fibers for $K=5$ and every $s$.
\end{remark}

\begin{figure}[t]
  \centering
  \includegraphics{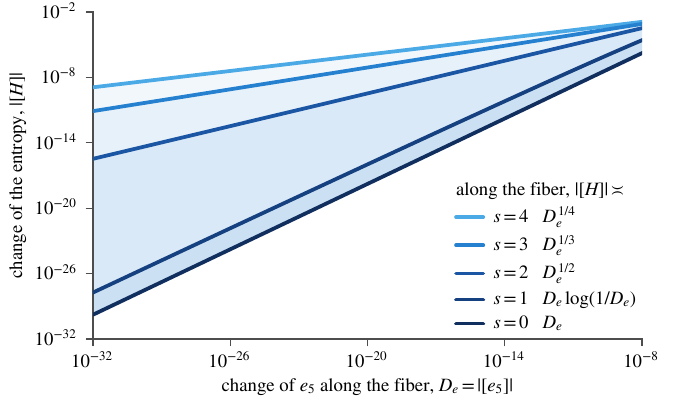}
  \caption{Change $|[\Ent]|$ of the entropy against the change $D_e=|[e_5]|$ of $e_5=\prod_c\theta_c$ between the
  two endpoints of the least-favorable fibers of \Cref{sec:fibers}, for $K=5$ possible meanings and $\tau=0.05$: the
  mixed fiber for $s=2,3,4$, the fiber for $s=1$ and the
  perturbation for $s=0$, at four amplitudes $\eps$ (or $\eta$) per decade, computed from the
  roots of $p_u$ in $160$-digit arithmetic. Only $e_5$ moves along each fiber, so $D_e$ is
  the distance in the one-sided transfer, and the entropy responds with its modulus $\phi_s(D_e)$. At $D_e=N^{-1/2}$ the square of this modulus is of the order of the rate $\psi_s(N)$.}
  \label{fig:modulus}
\end{figure}

\subsection{Sharp Local Minimax Constant}
\label{sec:results-sharp}

On every fixed polynomial window the class $s=0$ has the rate $\min\{1,N^{-1}\}$. At an interior point with distinct probabilities
we identify the constant in front of $N^{-1}$. For $\theta_0\in\Simplex$ and $h>0$ let
$\Rloc(\theta_0;h):=\inf_{\hH}\sup\{\Exp_\theta(\hH-\Ent(\theta))^2:\theta\in\Simplex,\ \|\theta-\theta_0\|\le h\}$, and let
\[
  \Lam(\theta):=\int_0^\infty\frac{\mathrm dt}{\prod_c(t+\theta_c)},\qquad
  \mathcal V(\theta):=\frac{K!}{K^K}\,\Lam(\theta)^2\,(q(1-q))^{K-1}.
\]
The integral $\Lam$ is finite exactly when every coordinate is positive. When the coordinates are also pairwise distinct,
$\Lam=\partial\Ent/\partial e_K$ with $e_1,\dots,e_{K-1}$ fixed.

\begin{theorem}[Sharp Local Constant]
  \label{thm:sharp}
  Fix $K\ge3$, $q\in(0,1)$, $\kappa\in(0,1/(K-1))$ and $\theta_0\in\Simplex$ with pairwise distinct positive coordinates.
  Let $\Delta_n>0$ satisfy $n\Delta_n^2\le n^{-\kappa}$ and $N\to\infty$, and let $h_n\to0$ with $h_n\sqrt{N}\to\infty$. Then
  \[
    \lim_{n\to\infty}N\,\Rloc(\theta_0;h_n)=\mathcal V(\theta_0),
  \]
  and \Cref{alg:main} attains it:
  $\lim_nN\sup_{\|\theta-\theta_0\|\le h_n}\Exp_\theta(\hH-\Ent(\theta))^2=\mathcal V(\theta_0)$.
\end{theorem}

The condition $\kappa<1/(K-1)$ serves only to make the assumptions compatible: $n\Delta^2\le n^{-\kappa}$ with $\kappa\ge1/(K-1)$ forces $N\le1$.

\begin{corollary}[Star Averages]
  \label{cor:star}
  Under the conditions of \Cref{thm:sharp}, the estimator $\hH^\star$ that replaces the all-tree averages in
  \Cref{alg:main} by the star averages \eqref{eq:star} satisfies
  $\lim_nN\sup_{\|\theta-\theta_0\|\le h_n}\Exp_\theta(\hH^\star-\Ent(\theta))^2=K^{K-3}\mathcal V(\theta_0)$; the loss factor $K^{K-3}$ equals $1$, $4$
  and $25$ for $K=3$, $4$ and $5$, respectively.
\end{corollary}

\begin{remark}[Efficiency Gain From Averaging All Trees]
  \label{rem:efficiency}
  The leading noise of both statistics comes from the spanning trees of $K$ answers with noise on every edge, whose noise products are uncorrelated. A
  $K$-set carries $K^{K-2}$ labeled spanning trees but only $K$ stars, so the leading variance of the star average is larger
  by the factor $(K-1)!\,K^{K-2}/K!=K^{K-3}$ (\Cref{lem:sharp-var}).
\end{remark}

\section{Proof of the Upper Bound}
\label{sec:upper}

The upper bound follows by composition. The all-tree averages satisfy $\Exp\|\hT-P(\theta)\|^2\le C/\min\{n,N\}$;
the projection at most doubles the Euclidean error and Newton's identities are Lipschitz on
$[0,1]^{K-1}$, so $\Exp D_e(\htheta,\theta)^2\le C/\min\{n,N\}$; and the one-sided local transfer
with Jensen's inequality (for $s=1$, concavity of $t\mapsto t(1+\frac12\log(1/t))^2$) gives
the rate.

In this section $0\le q<q+\Delta\le1$ and $\theta\in\Simplex$ are arbitrary and $n\ge2K$, except in the first statement of
the next lemma, which needs only $n\ge2k$. No window condition and no lower bound on $q$ are used. Write
$\xi_e:=X_e-W_e$, $r_e:=q+\Delta W_e$ and $v_e:=r_e(1-r_e)$. Given $Z$, the $\xi_e$ are independent with mean zero and variance
$v_e/\Delta^2\le\rho$, where $\rho:=\max\{q(1-q),(q+\Delta)(1-q-\Delta)\}/\Delta^2\le1/(4\Delta^2)$.

\subsection{All-Tree Averages}

\begin{lemma}[All-Tree Variance]
  \label{lem:alltree-var}
  For $2\le k\le K$ and $n\ge2k$, $\Exp\hT_k=P_k$ and
  \begin{gather*}
    \Var(\hT_k)\le\frac{k^2}{4n}+\sum_{j=1}^{k-1}\binom{k-1}j(2k)^{j+1}\frac{\rho^j}{n^{j+1}}
    \le\frac{C_k}{\min\{n,N_k\}},\\
    C_k:=\frac{k^2}4+\sum_{j=1}^{k-1}\binom{k-1}j\frac{(2k)^{j+1}}{4^j}.
  \end{gather*}
  Consequently, for $n\ge2K$, $\Exp\|\hT-P(\theta)\|^2\le C_{(1)}/\min\{n,N\}$ with $C_{(1)}:=\sum_{k=2}^KC_k$, where
  $\hT:=(\hT_2,\dots,\hT_K)$ and $P(\theta):=(P_2,\dots,P_K)$.
\end{lemma}

\begin{proof}
  \emph{Unbiasedness.} Given $Z$, distinct pairs are conditionally independent and
  $\Exp[X_e\mid Z]=W_e$, so $\Exp[\prod_{e\in\tr}X_e\mid Z]=\prod_{e\in\tr}W_e$, the indicator that all
  vertices of $\tr$ share a meaning, whose expectation is $P_k$.

  \emph{Decomposition.} Expanding $\prod_{e\in\tr}(W_e+\xi_e)$ over the set $E'\subseteq\tr$ of edges that
  carry noise,
  \[
    \hT_k=\bar P_k+\sum_{E'\ne\varnothing}\frac{c_{E'}}{|\Trees_k|}\,\xi_{E'},\qquad
    \bar P_k:=\frac1{|\Trees_k|}\sum_{\tr}\prod_{e\in\tr}W_e,\quad c_{E'}:=\sum_{\tr\supseteq E'}\prod_{e\in\tr\setminus E'}W_e,
    \quad\xi_{E'}:=\prod_{e\in E'}\xi_e ,
  \]
  where $E'$ ranges over the nonempty edge sets contained in some $\tr\in\Trees_k$ (forests). Given $Z$,
  $\Exp[\xi_{E'}\mid Z]=0$ for $E'\ne\varnothing$ and $\Exp[\xi_{E'}\xi_{E''}\mid Z]=0$ for $E'\ne E''$, because
  some edge occurs exactly once. Since $\bar P_k$ and the $c_{E'}$ are functions of $Z$,
  \[
    \Var(\hT_k)=\Var(\bar P_k)+\frac1{|\Trees_k|^2}\sum_{E'\ne\varnothing}\Exp\Big[c_{E'}^2\prod_{e\in E'}\frac{v_e}{\Delta^2}\Big].
  \]

  \emph{Signal.} Each $k$-set carries exactly $k^{k-2}$ trees, and $\prod_{e\in\tr}W_e$ is the indicator
  that the vertex set of $\tr$ is monochromatic, so $\bar P_k=\binom nk^{-1}\sum_{|V|=k}\Ind{V\text{ monochromatic}}$.
  Changing one label changes at most $\binom{n-1}{k-1}$ terms, so $\bar P_k$ moves by at most
  $\binom{n-1}{k-1}/\binom nk=k/n$. The Efron--Stein inequality with Popoviciu's inequality (a variable
  with range of length $\ell$ has variance at most $\ell^2/4$) gives $\Var(\bar P_k)\le n(k/n)^2/4=k^2/(4n)$.

  \emph{Noise.} $0\le c_{E'}\le\#\{\tr\supseteq E'\}$ and $\prod_{e\in E'}v_e/\Delta^2\le\rho^{|E'|}$. By double counting,
  \[
    \sum_{E'\ne\varnothing}\#\{\tr\supseteq E'\}^2\rho^{|E'|}=\sum_{\tr\in\Trees_k}\ \sum_{\varnothing\ne E'\subseteq\tr}\#\{\tr'\supseteq E'\}\,\rho^{|E'|}.
  \]
  Fix $\tr$ and $E'\subseteq\tr$ with $|E'|=j$. As a forest, $E'$ touches $v(E')\ge j+1$ vertices. A tree
  containing $E'$ is determined by its $k$-set, which contains the vertices of $E'$ (at most
  $\binom{n-v(E')}{k-v(E')}\le\binom{n-j-1}{k-j-1}$ choices), and by a spanning tree of that set containing
  $E'$ (at most $k^{k-2}$ choices). There are $\binom{k-1}j$ subsets $E'\subseteq\tr$ with $|E'|=j$. Dividing by
  $|\Trees_k|^2$ and summing over $\tr$, the terms with $|E'|=j$ contribute at most
  \[
    \frac1{|\Trees_k|}\binom{k-1}j\binom{n-j-1}{k-j-1}k^{k-2}\rho^j
    =\binom{k-1}j\frac{(k)_{j+1}}{(n)_{j+1}}\rho^j\le\binom{k-1}j\frac{(2k)^{j+1}}{n^{j+1}}\rho^j ,
  \]
  where $(a)_m$ is the falling factorial and $(n)_{j+1}\ge(n/2)^{j+1}$ for $j\le k-1$ and $n\ge2k$.

  \emph{Simplification.} $\rho\le1/(4\Delta^2)=n/(4\delta)$, so $\rho^j/n^{j+1}\le4^{-j}/(n\delta^j)$. If
  $\delta\le1$, then $1/(n\delta^j)\le1/(n\delta^{k-1})=1/N_k$ and $1/n\le1/N_k$; if $\delta\ge1$, then
  $1/(n\delta^j)\le1/n$. In both cases $\Var(\hT_k)\le C_k/\min\{n,N_k\}$. Finally, $\min\{n,N_k\}\ge\min\{n,N\}$
  for $k\le K$, and summing over $k$ gives the last claim.
\end{proof}

\subsection{Projection and Newton's Identities}

The moment set $\Mom:=\{P(\theta):\theta\in\Simplex\}\subseteq[0,1]^{K-1}$ is compact. The estimator takes a
nearest point $\tP$ of $\hT$ in $\Mom$, selected by a fixed rule, and any $\htheta\in\Simplex$ with
$P(\htheta)=\tP$; $\hH:=\Ent(\htheta)\in[0,\log K]$. For $\theta,\theta'\in\Simplex$ let
$D_e(\theta,\theta'):=\sum_{j=2}^K|e_j(\theta)-e_j(\theta')|$.

\begin{lemma}[Projection and Newton's Identities]
  \label{lem:newton}
  \begin{enumerate}[label=(\roman*),leftmargin=*]
    \item For $\theta\in\Simplex$, $\|\tP-P(\theta)\|\le2\|\hT-P(\theta)\|$.
    \item There are polynomials $E_2,\dots,E_K$, depending only on $K$, with $e_j=E_j(P_2,\dots,P_j)$ on
      $\Simplex$, and a constant $c_e$ depending only on $K$ with
      $D_e(\theta,\theta')\le\frac{c_e}2\|P(\theta)-P(\theta')\|$ for $\theta,\theta'\in\Simplex$. Hence
      $D_e(\htheta,\theta)\le c_e\|\hT-P(\theta)\|$, and $\hH$ does not depend on the choice of $\htheta$.
    \item $D_e(\theta,\theta')\le D_e^{\max}:=(1+1/K)^K-2<1$.
  \end{enumerate}
\end{lemma}

\begin{proof}
  (i) $P(\theta)\in\Mom$, so $\|\tP-P\|\le\|\tP-\hT\|+\|\hT-P\|\le2\|\hT-P\|$.
  (ii) Newton's identities $me_m=\sum_{i=1}^m(-1)^{i-1}e_{m-i}P_i$, with $e_0=1$ and $P_1=1$, give the
  $E_m$ by induction. On the convex set $[0,1]^{K-1}\supseteq\Mom$ the mean value theorem gives
  $|E_j(P)-E_j(P')|\le\sup_{[0,1]^{K-1}}\|\nabla E_j\|\,\|P-P'\|$; put $c_e:=2\sum_{j=2}^K\sup_{[0,1]^{K-1}}\|\nabla E_j\|$.
  Then $D_e(\htheta,\theta)\le\frac{c_e}2\|\tP-P(\theta)\|\le c_e\|\hT-P(\theta)\|$. Two points with the same
  moments have the same $e_j$, hence the same polynomial $p_\theta$, the same multiset of
  coordinates and the same entropy.
  (iii) By Maclaurin's inequality $0\le e_j(\theta)\le\binom Kj(e_1/K)^j=\binom KjK^{-j}$, so
  $|e_j(\theta)-e_j(\theta')|\le\binom KjK^{-j}$, and summing over $j\ge2$ gives $(1+1/K)^K-2<e-2<1$.
\end{proof}

\subsection{The One-Sided Local Transfer}
\label{sec:transfer}

Let $\phi_s(x):=x^{1/s}$ for $s\ge2$, $\phi_1(x):=x\log(e/x)$ and $\phi_0(x):=x$ on $(0,1]$, and $\phi_s(0):=0$; each is
continuous and nondecreasing on $[0,1]$ ($\phi_1'(x)=\log(1/x)\ge0$). The transfer rests on the integral representation of
the entropy in \Cref{lem:intrep} of \Cref{app:tools}: on $\Rare s\tau$, $\pbar_\theta(t)\ge\tau^{K-s}t^s$ for $t\le1$, and integrating gives the following
bounds.

\begin{lemma}[One-Sided Local Transfer]
  \label{lem:transfer}
  Let $0\le s\le K-1$.
  \begin{enumerate}[label=(\alph*),leftmargin=*]
    \item \emph{Two-sided bound.} Let $\tau'>0$, and define $\Rare s{\tau'}$ by the condition of
      \Cref{sec:formulation}, so that it is empty for $\tau'>1/(K-s)$; let $\theta,\theta'\in\Rare s{\tau'}$ and
      $a:=D_e(\theta,\theta')\,\tau'^{-(K-s)}$. Then
      \begin{equation*}
        |\Ent(\theta)-\Ent(\theta')|\le D_e\log2+g_s(a),\qquad
        g_s(a):=\begin{cases}\dfrac{\pi}{\sin(\pi/s)}\,a^{1/s}, & s\ge2,\\[1ex]
        (1+a)\log(1+a)-a\log a, & s=1,\\ a, & s=0.\end{cases}
      \end{equation*}
    \item \emph{Localization.} Let $0<\tau<1/(K-s)$, $\theta\in\Rare s\tau$, $\theta'\in\Simplex$ and
      $c_0:=(\tau/(4(K+1)))^K$. If $D_e(\theta,\theta')<c_0$, then $\theta'\in\Rare s{\tau/2}$.
    \item \emph{One-sided bound.} There is $C_T$ depending only on $(K,s,\tau)$ such that
      $|\Ent(\theta)-\Ent(\theta')|\le C_T\,\phi_s(D_e(\theta,\theta'))$ for all $\theta\in\Rare s\tau$ and $\theta'\in\Simplex$.
  \end{enumerate}
\end{lemma}

Part~(c) places no assumption on $\theta'$, and the projection step needs exactly this: $\htheta$ may leave the class, and the
Rouch\'e argument of part~(b) returns it to $\Rare s{\tau/2}$. The bound uses only $K-s$ large coordinates, so coinciding
large coordinates cost nothing. The proof of the sharp constant, by contrast, uses distinct coordinates to invert the moment map.

\begin{proof}
  (a) We use \Cref{lem:intrep} and $|\log(x/y)|\le\log(1+|x-y|/\min\{x,y\})$ for $x,y>0$. Since $e_0=e_1=1$
  for both points, $\pbar_\theta(t)-\pbar_{\theta'}(t)=\sum_{j=2}^K(e_j-e'_j)t^{K-j}$.
  On $t\ge1$ the difference is at most $D_et^{K-2}$ and $\pbar\ge t^K+t^{K-1}$, so the integrand is at
  most $\log(1+D_e/(t(t+1)))\le D_e/(t(t+1))$, whose integral over $[1,\infty)$ is $D_e\log2$.
  On $0<t\le1$ the difference is at most $D_e$. Each point has at least $K-s$ coordinates of at least $\tau'$;
  choosing $K-s$ of them, bounding their factors $t+\theta_c$ of $\pbar(t)$ below by $\tau'$ and the other
  $s$ factors by $t$ gives $\pbar(t)\ge\tau'^{K-s}t^s$, so the integrand is at most $\log(1+at^{-s})$. For $s\ge2$ the substitution
  $t=a^{1/s}x$ gives $\int_0^1\log(1+at^{-s})\,\mathrm dt\le a^{1/s}\int_0^\infty\log(1+x^{-s})\,\mathrm dx$, and an
  integration by parts (the boundary term $x\log(1+x^{-s})$ vanishes at $0$ and, as $s\ge2$, at $\infty$)
  gives $\int_0^\infty\log(1+x^{-s})\,\mathrm dx=s\int_0^\infty\frac{\mathrm dx}{1+x^s}=\frac{\pi}{\sin(\pi/s)}$. For $s=1$,
  $\int_0^1\log(1+a/t)\,\mathrm dt=\int_0^1[\log(t+a)-\log t]\,\mathrm dt=(1+a)\log(1+a)-a\log a$. For $s=0$ the
  integrand is at most $\log(1+a)\le a$.

  (b) Divide $(\tau/2,\tau]$ into $K+1$ half-open intervals of length $\tau/(2(K+1))$. The $K$ coordinates of
  $\theta$ miss at least one of them; let $\varrho$ be its midpoint. Then $\tau/2<\varrho<\tau$ and
  $|\varrho-\theta_c|\ge\tau/(4(K+1))$ for every $c$. On the circle $|z|=\varrho$,
  $|p_\theta(z)|=\prod_c|z-\theta_c|\ge\prod_c|\varrho-\theta_c|\ge c_0$, while
  $p_\theta(z)-p_{\theta'}(z)=\sum_{j=2}^K(-1)^j(e_j-e'_j)z^{K-j}$ has modulus at most $D_e<c_0$ because
  $\varrho<1$. By Rouché's theorem $p_\theta$ and $p_{\theta'}$ have the same number of zeros, counted with
  multiplicity, in $|z|<\varrho$. The zeros of $p_\theta$ there are the coordinates $\theta_c<\varrho$; they are all
  below $\tau$, and $\theta$ has at most $s$ coordinates below $\tau$. Hence $\theta'$ has at most $s$ coordinates
  below $\varrho$, that is, at least $K-s$ coordinates of at least $\varrho>\tau/2$.

  (c) If $D_e<c_0$, both points lie in $\Rare s{\tau/2}$ by (b), and (a) applies with $\tau'=\tau/2$ and
  $a=\Upsilon D_e$, $\Upsilon:=(2/\tau)^{K-s}\ge1$. For $s\ge2$, $D_e\le1$ gives $D_e\le D_e^{1/s}$ and the bound
  $[\log2+\frac{\pi}{\sin(\pi/s)}\Upsilon^{1/s}]D_e^{1/s}$. For $s=1$, $(1+a)\log(1+a)\le a(1+a)\le\Upsilon(1+\Upsilon)D_e$ and
  $-a\log a=\Upsilon D_e(\log(1/D_e)-\log\Upsilon)\le\Upsilon D_e\log(1/D_e)$, so the bound is at most
  $[\log2+\Upsilon(1+\Upsilon)]D_e+\Upsilon D_e\log(1/D_e)\le[\log2+\Upsilon(1+\Upsilon)]\,D_e\log(e/D_e)$. For $s=0$ the bound is
  $(\log2+\Upsilon)D_e$. If $D_e\ge c_0$, then $|\Ent(\theta)-\Ent(\theta')|\le\log K\le\log K\cdot\phi_s(D_e)/\phi_s(c_0)$,
  since $\phi_s$ is nondecreasing on $[0,1]$ and $D_e<1$ by \Cref{lem:newton}(iii). Take $C_T$ to be the
  largest of these constants.
\end{proof}

\subsection{Proof of \texorpdfstring{\Cref{thm:main}}{Theorem 1}(i)}
\label{sec:upper-proof}

Each $\psi_s$ is nonincreasing on $(0,\infty)$; for $s=1$, the derivative of $x^{-1}\log^2(e+x)$ is
$x^{-2}\log(e+x)[2x/(e+x)-\log(e+x)]$, and $\log(e+x)-2x/(e+x)$ has derivative $(x-e)/(e+x)^2$ and minimum
$\log2>0$ at $x=e$.

\begin{proof}
  Fix $\theta\in\Rare s\tau$, and let $D:=D_e(\htheta,\theta)$ and $m:=\min\{n,N\}$. By
  \Cref{lem:alltree-var,lem:newton}, $D\le\min\{D_e^{\max},c_e\|\hT-P\|\}$, $\Exp D^2\le c_e^2C_{(1)}/m=:t_1$
  and $D^2<1$. By \Cref{lem:transfer}(c), $(\hH-\Ent)^2\le C_T^2\phi_s(D)^2$; also $(\hH-\Ent)^2\le(\log K)^2$.

  \emph{$s\ge2$.} By Jensen's inequality for the concave map $x\mapsto x^{1/s}$,
  $\Exp\phi_s(D)^2=\Exp(D^2)^{1/s}\le t_1^{1/s}$, so the risk is at most
  $\min\{(\log K)^2,C_T^2(c_e^2C_{(1)})^{1/s}m^{-1/s}\}\le C_U\min\{1,m^{-1/s}\}$.

  \emph{$s=1$.} Let $g(t):=t(1+\frac12\log(1/t))^2$ on $(0,1]$ and $g(0):=0$, so that
  $\phi_1(D)^2=D^2\log^2(e/D)=g(D^2)$. With $\ell:=\log(1/t)$, $g'(t)=(1+\ell/2)\ell/2\ge0$ and
  $g''(t)=-(\ell+1)/(2t)<0$, so $g$ is continuous, nondecreasing and concave on $[0,1]$, and Jensen's
  inequality gives $\Exp g(D^2)\le g(\Exp D^2)\le g(\min\{1,t_1\})$. If $t_1<1$, then with $c':=c_e^2C_{(1)}$,
  $g(t_1)=\frac{c'}m(1+\frac12\log\frac m{c'})^2$ and $0<\log(m/c')\le\log(e+m)+|\log c'|$; as $\log(e+m)\ge1$,
  $g(t_1)\le c'(\frac32+\frac12|\log c'|)^2m^{-1}\log^2(e+m)$, so the risk is at most
  $\min\{(\log K)^2,C_T^2c'(\frac32+\frac12|\log c'|)^2m^{-1}\log^2(e+m)\}$. If $t_1\ge1$, then $m\le c'$, so
  $\psi_1(m)\ge\psi_1(c')>0$ and the risk is at most $(\log K)^2\le(\log K)^2\psi_1(m)/\psi_1(c')$. In both
  cases the risk is at most $C_U\psi_1(m)$.

  \emph{$s=0$.} The risk is at most $\min\{(\log K)^2,C_T^2t_1\}\le C_U\min\{1,m^{-1}\}$.

  In every case the constant depends only on $(K,s,\tau)$, and the estimator does not use $s$ or $\tau$.
  If $n\Delta^2\le1$, then $N=n\delta^{K-1}\le n$ and $m=N$.
\end{proof}

\subsection{Computing the Estimator}
\label{sec:computing}

Let $X$ be the symmetric $n\times n$ matrix with entries $X_{ij}$ and zero diagonal. For a multigraph $G$ let $\hom(G,X)$ be
the sum of $\prod_{\{u,v\}\in E(G)}X_{\varphi(u)\varphi(v)}$ over all maps $\varphi$ from the vertices of $G$ to $[n]$, and
$\mathrm{inj}(G,X)$ the same sum over injective maps. For fixed $K$ the sum over all labeled spanning trees of all
$k$-sets is a finite combination of homomorphism counts: grouping the trees by shape $F$ gives
$\sum_F\mathrm{inj}(F,X)/|\mathrm{Aut}(F)|$, and M\"obius inversion over the set partitions $\pi$ of $V(F)$ writes the
injective count as $\mathrm{inj}(F,X)=\sum_\pi\mu(0,\pi)\hom(F/\pi,X)$. Partitions that merge two adjacent vertices
contribute zero because $X$ has zero diagonal, and parallel edges of a quotient become entrywise powers of $X$. Each
$\hom(F/\pi,X)$ is a sum over at most $n^k$ maps, so for every fixed $K$ the statistics $\hT_2,\dots,\hT_K$ take a number of
operations polynomial in $n$.

For $k\le5$ the shapes are the edge; the path on three vertices; the star and the path on
four vertices; and the star, the path and the spider with legs $(2,1,1)$ on five vertices. With the quotient multigraphs
contracted by variable elimination, the statistics $\hT_2,\dots,\hT_5$ need one $n\times n$ matrix product and otherwise
$O(n^2)$ operations. The next proposition replaces the exact projection by a grid search.

\begin{proposition}[Approximate Projection]
  \label{prop:approx-proj}
  Let $n\ge2K$, $0\le q<q+\Delta\le1$ and $m:=\min\{n,N\}$.
  \begin{enumerate}[label=(\roman*),leftmargin=*]
    \item Let $\bar\theta\in\Simplex$ satisfy $\|P(\bar\theta)-\hT\|\le\mathrm{dist}(\hT,\Mom)+\eta$ for some $\eta\ge0$, and put
      $\bar H:=\Ent(\bar\theta)$. For every $\theta\in\Simplex$, $D_e(\bar\theta,\theta)\le c_e\|\hT-P(\theta)\|+c_e\eta/2$. If
      $\eta\le L_\eta m^{-1/2}$ for a constant $L_\eta$, then for $0\le s\le K-1$, $0<\tau<1/(K-s)$ and every
      $\theta\in\Rare s\tau$, $\Exp_\theta(\bar H-\Ent(\theta))^2\le\bar C_U\,\psi_s(m)$ with $\bar C_U$ depending only on $(K,s,\tau,L_\eta)$.
    \item Let $h:=1/\lceil m^{1/2}\rceil$ and $G_h:=\{\vartheta\in\Simplex:\vartheta_c/h\in\mathbb Z\text{ for all }c\}$, a set of
      $\binom{1/h+K-1}{K-1}\le(m^{1/2}+2)^{K-1}$ points. Every minimizer $\bar\theta$ of $\|P(\vartheta)-\hT\|$ over $\vartheta\in G_h$
      satisfies the hypothesis of (i) with $\eta=2(K-1)L_Ph\le L_\eta m^{-1/2}$, where $L_P:=(\sum_{k=2}^Kk^2)^{1/2}$ and
      $L_\eta:=2(K-1)L_P$.
    \item Let $\tilde e_0:=1$ and let $\tilde e_1,\dots,\tilde e_K$ be computed from $(1,\hT_2,\dots,\hT_K)$ by Newton's
      identities. Then $\hT\in\Mom$ if and only if all roots of $\sum_{j=0}^K(-1)^j\tilde e_jz^{K-j}$ are real and
      nonnegative, and then these roots form a point $\htheta\in\Simplex$ with $P(\htheta)=\hT$. Let $\check H:=\Ent(\htheta)$
      on $\{\hT\in\Mom\}$ and $\check H:=\bar H$ otherwise, with $\bar\theta$ from (ii). Then $\check H$ satisfies the bound of
      (i), and under the conditions of \Cref{prop:attain} its first display holds with $\check H$ in place of $\hH$;
      hence the conclusion of \Cref{thm:sharp} holds for $\check H$. The same holds for the star averages
      with the second display of \Cref{prop:attain} and \Cref{cor:star}.
  \end{enumerate}
\end{proposition}

\begin{proof}
  (i) As $P(\theta)\in\Mom$, $\mathrm{dist}(\hT,\Mom)\le\|\hT-P(\theta)\|$, so
  $\|P(\bar\theta)-P(\theta)\|\le\|P(\bar\theta)-\hT\|+\|\hT-P(\theta)\|\le2\|\hT-P(\theta)\|+\eta$, and \Cref{lem:newton}(ii) gives the
  first claim. Let $D:=D_e(\bar\theta,\theta)$. By \Cref{lem:newton}(iii), $D<1$, and by $(a+b)^2\le2a^2+2b^2$ and
  \Cref{lem:alltree-var}, $\Exp D^2\le(2c_e^2C_{(1)}+c_e^2L_\eta^2/2)/m$. The proof of \Cref{thm:main}(i)
  uses only these two properties of $D$, \Cref{lem:transfer}(c) for $\theta\in\Rare s\tau$ and a
  second point of $\Simplex$, and $0\le\bar H\le\log K$; with $c_e^2C_{(1)}$ replaced by $2c_e^2C_{(1)}+c_e^2L_\eta^2/2$ it gives
  the bound.

  (ii) For $x,y\in[0,1]$, $|x^k-y^k|\le k|x-y|$, so $|P_k(\vartheta)-P_k(\vartheta')|\le k\|\vartheta-\vartheta'\|_1$ and
  $\|P(\vartheta)-P(\vartheta')\|\le L_P\|\vartheta-\vartheta'\|_1$ on $\Simplex$. As $\Mom$ is compact, there is $\theta^*\in\Simplex$ with
  $\|P(\theta^*)-\hT\|=\mathrm{dist}(\hT,\Mom)$. Rounding $\theta^*_1,\dots,\theta^*_{K-1}$ down to multiples of $h$ and setting
  $\vartheta_K:=1-\sum_{c<K}\vartheta_c\ge\theta^*_K$ gives a point $\vartheta$ with $\|\vartheta-\theta^*\|_1\le2(K-1)h$, and
  $\vartheta\in G_h$ because $1/h$ is an integer. Hence $\|P(\bar\theta)-\hT\|\le\|P(\vartheta)-\hT\|\le\mathrm{dist}(\hT,\Mom)+2(K-1)L_Ph$, and
  $h\le m^{-1/2}$. The points of $G_h$ correspond to the compositions of $1/h$ into $K$ nonnegative parts, and
  $1/h\le m^{1/2}+1$.

  (iii) If $\hT=P(\theta)$ with $\theta\in\Simplex$, Newton's identities give $\tilde e_j=e_j(\theta)$, and the polynomial is
  $p_\theta$, whose roots are the coordinates of $\theta$. Conversely, if the roots $r_1,\dots,r_K$ are real and
  nonnegative, then $\sum_cr_c=\tilde e_1=1$, so $r\in\Simplex$, and $e_j(r)=\tilde e_j$; since Newton's identities determine
  the power sums from the elementary symmetric polynomials, $P(r)=\hT$. On $\{\hT\in\Mom\}$, $\htheta$ satisfies the
  hypothesis of (i) with $\eta=0$, and otherwise $\bar\theta$ satisfies it by (ii); so $\check H$ satisfies the bound of
  (i). In the proof of \Cref{prop:attain}, on the event $\mathcal E$ one has $\hT\in B_1\subset\Mom$ by \Cref{lem:chart}(ii), so
  $\check H=\Ent(\htheta)$ with $P(\htheta)=\hT$. By Newton's identities, $\htheta$ has the multiset of coordinates of
  $(\Pi|_W)^{-1}(\hT)$, so $\check H=\tilde H(\hT)$, with $\Pi$, $W$ and $\tilde H$ as in \Cref{lem:chart} and its proof. On
  $\mathcal E^c$, $0\le\check H\le\log K$. These are the only properties of $\hH$
  that the proof uses, so its first display holds for $\check H$. The proof of \Cref{thm:sharp} then applies,
  since \Cref{prop:vantrees} bounds every estimator. For the star averages replace $\hT$ by $\hTs$.
\end{proof}

The grid of part~(ii) has at most $(m^{1/2}+2)^{K-1}\le(n^{1/2}+2)^{K-1}$ points, and applying Newton's identities and
computing the roots in part~(iii) take a number of operations that depends only on $K$. Hence the estimator $\check H$ runs in time
polynomial in $n$ for every fixed $K$, and for $K\le5$ in one $n\times n$ matrix product and $O(n^2)$ further operations.

\section{Proof of the Lower Bound}
\label{sec:lower}

Throughout this section, $q_0\in(0,1/2]$ and the window exponent
$\kappa\in(0,1)$ are fixed, $q\in[q_0,1-q_0]$ and $0<\Delta\le q_0/2$, so that $p\le1-q_0/2$. On the window $n\Delta^2\le n^{-\kappa}$ the last condition
holds for $n\ge4/q_0^2$, because $\Delta=\sqrt{\delta/n}\le n^{-1/2}$. The proof uses fibers, paths
in the class along which $P_2,\dots,P_{K-1}$ stay fixed while $P_K$ moves. By a Hellinger bound along the fiber and Le Cam's
two-point argument, the rate follows once the Fisher information $I(u)$ of the verdict
graph is bounded along the fiber at the amplitude $\eps=\min\{\eps_1,a_0N^{-1/(2s)}\}$ for $s\ge2$, and at
$\min\{\eps_1,a_0N^{-1/2}\}$ for $s\in\{0,1\}$, with constants $\eps_1,a_0>0$. The capped amplitude covers both regimes:
vanishing risk when $N$ is large and constant risk when $N$ is bounded.

\subsection{Fibers That Move Only the Constant Term}
\label{sec:fibers}

The lower bounds use paths of the following kind, with $r$ the order of the sweep and a dot for $d/du$.

\begin{definition}[Fiber Conditions]
  \label{def:H}
  Let $r\ge1$ be an integer, $L\ge1$ and $\mu,\eta\in(0,1]$. A path $u\mapsto\theta(u)$, $u\in[u_0,u_1]$ with
  $u_1-u_0\le1$, satisfies condition~(H) with parameters $(\mu,\eta,L)$ if for every $u\in[u_0,u_1]$:
  \begin{itemize}[leftmargin=2.6em,itemsep=0pt,topsep=2pt]
    \item[(H0)] $\theta$ is real analytic on an open interval containing $[u_0,u_1]$, $\theta(u)\in\Simplex$ and
      $\theta_c(u)>0$ for all $c$;
    \item[(H1)] $\dot P_j(u)=0$ for $2\le j\le K-1$;
    \item[(H2)] $|\dot P_k(u)|\le L\mu$ for $K\le k\le r$ (void if $r<K$);
    \item[(H3)] $\|\dot\theta(u)\|_1\le L\eta$;
    \item[(H4)] $\max_c|\dot\theta_c(u)|/\theta_c(u)\le L$.
  \end{itemize}
\end{definition}

The coefficients of $p_\theta(z)=\prod_c(z-\theta_c)$ are $\pm e_j(\theta)$. A path that changes only the constant term
$(-1)^Ke_K$ keeps $e_1,\dots,e_{K-1}$ fixed and, by Newton's identities, also $P_2,\dots,P_{K-1}$, so it satisfies~(H1). By
Newton's identity $Ke_K=\sum_{i=1}^K(-1)^{i-1}e_{K-i}P_i$,
\begin{equation}
  \label{eq:d:pk}
  P_K=\sum_{i=1}^{K-1}(-1)^{i-1}e_iP_{K-i}+(-1)^{K-1}Ke_K,\qquad\text{so}\qquad \dot P_K=(-1)^{K-1}K\,\dot e_K
\end{equation}
along such a path. For a function $F$ on the simplex and a path $\theta(u)$ we write
$[F]^b_a:=F(\theta(b))-F(\theta(a))$.

For $2\le s\le K-1$, the mixed fiber starts from a polynomial with $K-s$ distinct roots above $\tau$ and $s$ roots at $\eps$
times the zeros of a Chebyshev polynomial, and perturbs its constant term by an amount $\sigma_\eps\asymp\eps^s$. The small roots then move by order $\eps$, and so does the entropy, while $P_K$ moves by
$K\sigma_\eps$. For $s=K-1$ there is a single macroscopic meaning, and the fiber runs next to a vertex of the simplex.
The case $s=1$
uses a fiber on which the single small root moves from $\eps$ to $2\eps$, with entropy gap at least $\frac12\eps\log(1/\eps)$, and the case $s=0$ perturbs an interior point by $\eta$. The three fiber lemmas below are proved in
\Cref{app:fibers}.

\subsubsection{Distinct Masses}
For an integer $m\ge1$ and $0<\tau<1/m$, let $\bar y(m,\tau):=(1)$ if $m=1$, and otherwise
$\bar y_c:=1/m+(1/m-\tau)(c-(m+1)/2)/m$, $c=1,\dots,m$. The entries are pairwise distinct
and sum to $1$, and the smallest is $1/m-(1/m-\tau)(m-1)/(2m)>1/m-(1/m-\tau)=\tau$.

\subsubsection{The Mixed Fiber}
Let $2\le s\le K-1$, $0<\tau<1/(K-s)$, $m:=K-s$, $\bar y:=\bar y(m,\tau)$, and let
$x_j:=(1+\cos\frac{(2j-1)\pi}{2s})/2$, $j=1,\dots,s$, be the zeros of $T_s(2x-1)$, where $T_s$ is the
Chebyshev polynomial of the first kind; $\sum_jx_j=s/2$. For $\eps\in(0,1/s)$ put $y_c:=\bar y_c(1-\eps s/2)$,
$\sigma_\eps:=2^{1-2s}\eps^s\prod_cy_c$ and, for $u\in[1/4,3/4]$,
\[
  p_u(z):=\prod_{c=1}^m(z-y_c)\prod_{j=1}^s(z-\eps x_j)+(u-\tfrac12)\,\sigma_\eps .
\]

\begin{lemma}[Mixed Fiber]
  \label{lem:fiber-mixed}
  There are $\eps_F>0$ and $C_F$, depending only on $(K,s,\tau)$, such that the following hold for $0<\eps\le\eps_F$ and
  $u\in[1/4,3/4]$.
  \begin{enumerate}[label=(\roman*),leftmargin=*]
    \item $p_u$ has $K$ positive simple real roots; their vector $\theta(u)$ lies in $\Simplex\cap\Rare s\tau$ and is
      real analytic in $u$ on a neighborhood of $[1/4,3/4]$. The $m$ largest roots satisfy
      $|\theta_c(u)-y_c|\le C_F\eps^s$ and $\theta_c(u)>\tau$; the $s$ smallest are $\eps w_j(u,\eps)$, where the
      $w_j$ are real analytic on a neighborhood of $[1/4,3/4]\times[0,\eps_F]$, $w_j(u,0)$ are the $s$ roots of
      $T_s(2w-1)=(-1)^{m+1}(u-\frac12)$ in $(0,1)$, and $w_{\min}\le w_j\le2$ for a constant $w_{\min}>0$
      depending only on $(K,s,\tau)$. At $u=\frac12$ the roots are $y_1,\dots,y_m$ and $\eps x_1,\dots,\eps x_s$.
    \item $e_1,\dots,e_{K-1}$ are constant along the path and $e_K(u)=e_K(\theta(\frac12))+(-1)^K(u-\frac12)\sigma_\eps$;
      hence $\dot P_j=0$ for $j\le K-1$ and $\dot P_K=-K\sigma_\eps$.
    \item The large roots have $|\dot\theta_c|\le C_F\eps^s$; the small roots have $|\dot\theta_c|\le C_F\eps$ and
      $|\dot\theta_c|/\theta_c\le C_F$. Hence $\|\dot\theta\|_1\le C_F\eps$ and
      $|\dot P_k|\le k(m+s2^{k-1})C_F\eps^s$ for every $k\ge K$. Therefore, for every $r\ge1$ the path satisfies~(H) with
      $\mu=\eps^s$, $\eta=\eps$ and $L$ depending only on $(K,s,\tau,r)$.
    \item There are $h_s>0$ depending only on $(K,s)$ and $\eps_H>0$ depending only on $(K,s,\tau)$ such that
      $|\Ent(\theta(\frac34))-\Ent(\theta(\frac14))|\ge h_s\eps/2$ for $\eps\le\eps_H$.
  \end{enumerate}
\end{lemma}

\begin{remark}[Sharpness of the Exponent]
  \label{rem:sharp-exponent}
  For the endpoints $\theta(\frac14),\theta(\frac34)$ of the mixed fiber, $D_e=|[e_K]^{3/4}_{1/4}|=\sigma_\eps/2=2^{-2s}\eps^s\prod_cy_c\le2^{-2s}\eps^s$,
  while $|[\Ent]^{3/4}_{1/4}|\ge h_s\eps/2$ for $\eps\le\eps_H$ (part~(iv) above). A bound $|[\Ent]^{3/4}_{1/4}|\le CD_e^{\alpha}$ on
  $\Rare s\tau$ would give $h_s\eps/2\le C2^{-2s\alpha}\eps^{s\alpha}$ for all small $\eps$, so $\alpha\le1/s$: for $s\ge2$ the exponent
  $1/s$ of the one-sided transfer cannot be improved on $\Rare s\tau$.
\end{remark}

\subsubsection{The Fiber for \texorpdfstring{$s=1$}{s = 1}}
Let $0<\tau<1/(K-1)$ and $b:=\bar y(K-1,\tau)$. For $\eps\in(0,1)$ let $y_j:=(1-\eps)b_j$ and
$F_\eps(z):=(z-\eps)\prod_{j=1}^{K-1}(z-y_j)$; for $u\in[0,1]$ let $\omega(u):=\eps(1+u)$ and
\[
  p_u(z):=F_\eps(z)-F_\eps(\omega(u)).
\]

\begin{lemma}[Fiber for $s=1$]
  \label{lem:fiber-one}
  There are $\eps_F>0$, $C_F$ and $\eps_H>0$, depending only on $(K,\tau)$, such that for $0<\eps\le\eps_F$ and
  $u\in[0,1]$:
  \begin{enumerate}[label=(\roman*),leftmargin=*]
    \item the roots of $p_u$ are $\omega(u)$ and $\theta_j(u)$, $j=1,\dots,K-1$, all positive, real and simple, with
      $|\theta_j(u)-b_j|\le C_F\eps$ and $\theta_j(u)>\tau$; the root vector $\theta(u)$ lies in $\Simplex\cap\Rare1\tau$ and is
      real analytic in $u$;
    \item $e_1,\dots,e_{K-1}$ are constant and $(-1)^Ke_K(u)=F_\eps(0)-F_\eps(\omega(u))$; hence $\dot P_j=0$ for
      $j\le K-1$, $\dot P_K=K\eps F_\eps'(\omega(u))$ and $|\dot P_K|\le C_F\eps$;
    \item $\dot\omega=\eps$ and $\dot\omega/\omega=1/(1+u)\in[\frac12,1]$; the large roots have $|\dot\theta_j|\le C_F\eps$ and
      $|\dot\theta_j|/\theta_j\le C_F\eps/\tau$; hence $\|\dot\theta\|_1\le C_F\eps$, $|\dot P_k|\le kC_F\eps$ for every $k$, and for
      every $r\ge1$ the path satisfies~(H) with $\mu=\eta=\eps$ and $L$ depending only on $(K,\tau,r)$;
    \item for $\eps\le\eps_H$, $\Ent(\theta(1))-\Ent(\theta(0))\ge\frac12\eps\log(1/\eps)$.
  \end{enumerate}
\end{lemma}

\subsubsection{The Perturbation for \texorpdfstring{$s=0$}{s = 0}}
Let $0<\tau<1/K$, $\bar y:=\bar y(K,\tau)$ and $\Lam(\theta):=\int_0^\infty\mathrm dt/\pbar_\theta(t)$. For $\eta\in(0,1]$ and
$u\in[0,1]$ let $\theta(u)$ be the vector of roots of $p_u(z):=p_{\bar y}(z)+(-1)^K\eta u$.

\begin{lemma}[Perturbation for $s=0$]
  \label{lem:fiber-zero}
  There are $\eta_F>0$ and $C_F$, depending only on $(K,\tau)$, such that for $0<\eta\le\eta_F$ and $u\in[0,1]$:
  \begin{enumerate}[label=(\roman*),leftmargin=*]
    \item $p_u$ has $K$ positive simple real roots with $|\theta_c(u)-\bar y_c|\le C_F\eta$ and $\theta_c(u)>\tau$;
      $\theta(u)\in\Simplex\cap\Rare0\tau$ is real analytic in $u$;
    \item $e_K(u)=e_K(\bar y)+\eta u$ and the other $e_j$ are constant; $\dot P_j=0$ for $j\le K-1$ and
      $\dot P_K=(-1)^{K-1}K\eta$;
    \item $|\dot\theta_c|\le C_F\eta$, $|\dot\theta_c|/\theta_c\le C_F\eta/\tau$ and $|\dot P_k|\le kKC_F\eta$; for every $r\ge1$ the path
      satisfies~(H) with parameters $(\eta,\eta,L)$, where $L$ depends only on $(K,\tau,r)$;
    \item if moreover $\eta\le\prod_c\bar y_c$, then $\Ent(\theta(1))-\Ent(\theta(0))\ge\eta\Lam(\bar y)/2$, and
      $0<\Lam(\bar y)<\infty$.
  \end{enumerate}
\end{lemma}

\subsection{Score, Sweeps and Motifs}
\label{sec:score}

Let $u\mapsto\theta(u)$ be a path with positive coordinates, write $P_u$ for the law of $A$ under $\theta(u)$, and let
$I(u):=\Exp_u\Sobs^2$ be its Fisher information, where $\Sobs:=\partial_u\log P_u(A)$ is the observed score.
Let $J_{ij}:=A_{ij}\log\frac{q+\Delta}q+(1-A_{ij})\log\frac{1-q-\Delta}{1-q}$ be the log-likelihood ratio of verdict $A_{ij}$
between a shared and a split meaning, $F_{ic}:=\sum_{j\ne i}J_{ij}\Ind{Z_j=c}$ the field of
meaning $c$ at answer $i$, and $\ell_c:=\dot\theta_c/\theta_c$. The observed score equals $\Exp_u[\SG\mid A]$ with
\[
  \SG:=\sum_{i=1}^n\Phi(F_i),\qquad \Phi(F):=\frac{\sum_c\dot\theta_ce^{F_c}}{\sum_c\theta_ce^{F_c}},
\]
by the Glauber identity (\Cref{lem:glauber,lem:score}). For $U\subseteq[n]$ let $\Pu_Uf:=\Exp_u[f\mid A,Z_{-U}]$, which averages the
labels in $U$ over their posterior, and $\DS_S:=\sum_{U\subseteq S}(-1)^{|S|-|U|}\Pu_U$. The sweep of order $r$ is
$\Sw r:=\sum_{|S|\le r}\DS_S\SG$. Every $\DS_S$ with $S\ne\varnothing$ has conditional mean zero given $A$, so
$\Exp_u[\Sw r\mid A]=\Sobs$ and $I(u)\le\Exp_u\Sw r^2$ (\Cref{lem:sweep}). For fixed labels, $\Sw r$ is an analytic
function of the edge weights near $J=0$. We write $B^{(r)}_M(z)$ for its Taylor coefficient of the monomial
$J^M:=\prod_eJ_e^{m_e}$ of a \emph{motif}, an edge multiset $M$ with multiplicities $m_e$, $v(M)$ vertices, $d(M)$
distinct edges and $|M|:=\sum_em_e$. \Cref{app:sweeps} gives the formal construction. Each $J_e$ is of order $\Delta$, so a motif on $v$ vertices with $d$ distinct edges contributes
about $n^v\Delta^{2d}\,\Exp B^{(r)}_M(Z)^2$ to the second moment of the sweep.

\subsection{The Information Inequality}
\label{sec:information}

\begin{lemma}[Information Inequality]
  \label{lem:info}
  Let $K\ge3$, $q_0\in(0,1/2]$, $r\ge1$, $D\ge r+1$ and $L\ge1$. There are $C_M$ and $n_M$ depending only
  on $(K,q_0,r,D,L)$ such that, if $n\ge n_M$, $q\in[q_0,1-q_0]$, $\delta\le1$ and $\sqrt\delta\log n\le1/C_M$,
  then for every path satisfying condition~(H) of \Cref{def:H} with parameters $(\mu,\eta,L)$ and every $u\in[u_0,u_1]$,
  \begin{equation}
    \label{eq:c:info}
    I(u)\le\Exp_u\Sw r^2\le C_M\Big[N\mu^2+n\delta^r\eta^2+\eta^2n^{r+2}\delta^{D+1}(\log n)^{2(D+1-r)}
    +n^{4D+2}e^{-q_0^2\log^2n/50}\Big].
  \end{equation}
\end{lemma}

For $n\ge n_M$ we have $\Delta=\sqrt{\delta/n}\le n^{-1/2}\le q_0/2$, so \Cref{lem:weights} applies. The proof
splits $\Sw r$ into its Taylor polynomial of degree $D$ in the edge weights and a remainder, and has four steps.
\emph{(a) Small motifs have exact coefficients.} For $v(M)\le r$ the sweep contains the full posterior average over
the labels of the motif's vertices, and $B^{(r)}_M=\partial_u\varkappa_M/M!$, where $\varkappa_M$ is the joint cumulant of
the co-membership indicators $\Ind{Z_k=Z_l}$ of the edges of $M$ and $M!:=\prod_em_e!$. The cumulant is a polynomial in $P_2,\dots,P_{v(M)}$, so its
derivative vanishes when $v(M)\le K-1$ and is $O(\mu)$ otherwise.
\emph{(b) Every coefficient is $O(\|\dot\theta\|_1)$.} With $v:=v(M)$, continued to the complex polydisc of
radius $\varrho_v:=\log(3/2)/\binom v2$ in the edge weights, all posterior normalizers stay at least $1/2$ in
modulus, and the Cauchy estimates give $|B^{(r)}_M|\le9v3^v\varrho_v^{-|M|}\|\dot\theta\|_1$ for every motif and every
labeling. No label assumption is needed.

\emph{(c) Pairing bounds the variance.} Disconnected motifs have zero coefficients, and motifs on
disjoint vertex sets give independent terms, so only overlapping pairs of connected motifs enter the
variance of the Taylor polynomial. Such a pair forms a connected multigraph $U$ with $d(U)$ distinct
edges. Given the labels, the edge weights are independent, with $\Exp[J_e\mid Z]=O(\Delta^2)$ and
$\Exp[J_e^a\mid Z]=O(\Delta^a)$ for $a\ge2$, so every distinct edge of $U$ contributes $O(\Delta^2)$. The at most
$n^{v(U)}\le n^{d(U)+1}$ placements of $U$ therefore contribute at most $n\delta^{d(U)}$ times the product of the two
coefficients. Small motifs with nonzero coefficients have at least $K$ vertices, so $d\ge K-1$, which
gives $N\mu^2$. Large motifs, on at least $r+1$ vertices, have $d\ge r$, which gives $n\delta^r\eta^2$.
\emph{(d) The remainder needs no cancellation.} The score has mean zero, so $\Exp_u\Sw r^2$ is at most twice the variance
of the Taylor polynomial plus four times the second moment of the remainder. The remainder is bounded term by term
on the event that all fields are at most $\sqrt\delta\log n$.

\subsubsection{Coefficients}
The next two lemmas compute the Taylor coefficients of the sweep exactly for small motifs and bound them uniformly for all
motifs.

\begin{lemma}[Small Motifs]
  \label{lem:small-motif}
  Let the path satisfy (H0), (H1) and (H2). For every motif $M$ with $v(M)\le r$ and every $z$,
  $B^{(r)}_M(z)=\partial_u\varkappa_M/M!$, which does not depend on $z$. It vanishes if $v(M)\le\min\{r,K-1\}$, and
  if $K\le v(M)\le r$ then $|B^{(r)}_M|\le\Gamma'_ML\mu$ with $\Gamma'_M$ depending only on $M$.
\end{lemma}

\begin{proof}
  If $v(M)\le r$, the sum over $|S|>r$ in \eqref{eq:a:identity} is empty. By \Cref{lem:identity},
  $\varkappa_M=Q_M(P_2,\dots,P_{v(M)})$ for a polynomial $Q_M$ depending only on $M$. Writing
  $P^{[v]}(\theta):=(P_2(\theta),\dots,P_v(\theta))$, we get
  $\partial_u\varkappa_M=\sum_{k=2}^{v(M)}(\partial_kQ_M)(P^{[v(M)]}(\theta(u)))\,\dot P_k$. By (H1) the terms with
  $k\le K-1$ vanish, and if $v(M)\le K-1$ there are no other terms. If $K\le v(M)\le r$, (H2) and
  $P^{[v(M)]}(\theta)\in[0,1]^{v(M)-1}$ give
  $|\partial_u\varkappa_M|\le v(M)\sup_{[0,1]^{v(M)-1}}\max_k|\partial_kQ_M|\cdot L\mu$.
\end{proof}

\begin{lemma}[Uniform Coefficient Bound]
  \label{lem:uniform}
  Let all $\theta_c(u)>0$. For every motif $M$, with $V:=V(M)$ and $v:=|V|\ge2$, every $r\ge0$ and every
  $z\in[K]^V$,
  \[
    |B^{(r)}_M(z)|\le9v\,3^v\,\varrho_v^{-|M|}\,\|\dot\theta\|_1,\qquad \varrho_v:=\log(3/2)\Big/\binom v2 .
  \]
  In particular, under (H3), $|B^{(r)}_M|\le\Gamma_ML\eta$ with $\Gamma_M:=9v3^v\varrho_v^{-|M|}$ depending only on $M$.
\end{lemma}

\begin{proof}
  Let $E_V$ be the set of pairs inside $V$, $|E_V|=\binom v2$, and work on the motif system of \Cref{app:sweeps}
  on the complex polydisc $\mathbb D:=\{J\in\Cpx^{E_V}:\max_e|J_e|\le\varrho_v\}$. For $U\subseteq V$ and fixed $z_{V\setminus U}$, the
  normalizing constant
  \[
    \mathcal Z_U(J):=\sum_{z'_U}\prod_{k\in U}\theta_{z'_k}\exp\Big(\sum_{e\in E_V,\ e\cap U\ne\varnothing}J_eW_e(z'_U,z_{V\setminus U})\Big)
  \]
  equals $\sum_{z'_U}\prod_k\theta_{z'_k}=1$ at $J=0$. The exponent has modulus at most
  $\binom v2\varrho_v=\log(3/2)$ and $|e^w-1|\le e^{|w|}-1$, so $|\mathcal Z_U(J)-1|\le1/2$, $|\mathcal Z_U|\ge1/2$, and
  $\sum_{z'_U}|\prod_k\theta_{z'_k}\exp(\cdots)|\le3/2$. For $i\in V$ the field
  $F_i=\sum_{j\in V\setminus\{i\}}J_{ij}\mathbf e_{z_j}$ satisfies $|F_{ic}|\le(v-1)\varrho_v\le\log(3/2)$, so in the same way
  $|\sum_c\theta_ce^{F_{ic}}|\ge1/2$ and $|\sum_c\dot\theta_ce^{F_{ic}}|\le\frac32\|\dot\theta\|_1$; hence
  $|\Phi(F_i)|\le3\|\dot\theta\|_1$ and $|\SG^{(V)}|\le3v\|\dot\theta\|_1$ for every labeling. Therefore
  $\Pu^{(V)}_U\SG^{(V)}=\mathcal Z_U^{-1}\sum_{z'_U}(\cdots)\SG^{(V)}(z'_U,z_{V\setminus U})$ is holomorphic on a
  neighborhood of $\mathbb D$ (its denominator does not vanish there) with modulus at most
  $\frac{3/2}{1/2}\cdot3v\|\dot\theta\|_1=9v\|\dot\theta\|_1$ on $\mathbb D$, and $\DS^{(V)}_S\SG^{(V)}$ has modulus at most
  $2^{|S|}\cdot9v\|\dot\theta\|_1$. The Cauchy estimates on a polydisc give $|\coef Mg|\le\sup_{\mathbb D}|g|/\varrho_v^{|M|}$.
  By \Cref{lem:local}(iii), $B^{(r)}_M=\sum_{S\subseteq V,\,|S|\le r}\coef M\DS^{(V)}_S\SG^{(V)}$, so
  $|B^{(r)}_M|\le\sum_{S\subseteq V}2^{|S|}\cdot9v\|\dot\theta\|_1\varrho_v^{-|M|}=3^v\cdot9v\|\dot\theta\|_1\varrho_v^{-|M|}$.
\end{proof}

\begin{remark}[Majority-Vertex Counting]
  \label{rem:majority}
  The concurrent four-meaning analysis \cite{anonymous2026four} bounds these coefficients by counting
  the vertices of $S$ labeled with the single macroscopic meaning. That count needs a single
  macroscopic meaning, for the following reason. If $\max_b|F_b|\le f$, the normalizer $\sum_b\theta_be^{F_b}$ lies in
  $[e^{-f},e^f]$, so $e^{-2f}\theta_c\le\pi_c(F)\le e^{2f}\theta_c$ for every $c$. With one macroscopic meaning and all
  other coordinates at most $\eps$, a posterior average under such fields returns a label to the
  macroscopic meaning except with probability at most $(K-1)e^{2f}\eps$. With two or more meanings of probability
  at least $\tau$, each keeps posterior probability at least $e^{-2f}\tau$, so a posterior average can
  move a label between macroscopic meanings. The uniform bound uses only $\|\dot\theta\|_1$ and makes no
  assumption on the labels or the fields.
\end{remark}

\subsubsection{Pairing}
Let $\Poly_D:=\sum_{1\le|M|\le D}J^MB^{(r)}_M(Z)$, the sum running over the motifs of the complete graph on
$[n]$. Since $\Phi(0)=0$, $\SG$ and hence $\Sw r$ vanish at $J=0$, and $\Poly_D$ is the Taylor polynomial of
degree $D$ of $\Sw r=\mathcal S_r(J;Z)$ at $J=0$.

\begin{lemma}[Pairing]
  \label{lem:pairing}
  Let the path satisfy (H0)--(H3), $q\in[q_0,1-q_0]$, $\Delta\le q_0/2$ and $\delta\le1$. Then
  $\Var(\Poly_D)\le C_D[N\mu^2+n\delta^r\eta^2]$ with $C_D$ depending only on $(K,q_0,r,D,L)$.
\end{lemma}

\begin{proof}
  Write $\mathcal X_M:=J^MB^{(r)}_M(Z)$ and split the motifs into three groups: the \emph{null} group, of
  disconnected motifs (\Cref{lem:local}(iv)) and connected ones with $v(M)\le\min\{r,K-1\}$
  (\Cref{lem:small-motif}), for which $B^{(r)}_M\equiv0$; the \emph{small} group, connected with
  $K\le v(M)\le r$, for which $B^{(r)}_M$ is a constant with $|B^{(r)}_M|\le\Gamma'_ML\mu$; and the \emph{large} group,
  connected with $v(M)\ge r+1$, for which $|B^{(r)}_M|\le\Gamma_ML\eta$ (\Cref{lem:uniform}). With $\Pi_{\mathrm s}$ and
  $\Pi_{\mathrm l}$ the sums over the small and large groups, $\Var(\Poly_D)\le2\Var(\Pi_{\mathrm s})+2\Var(\Pi_{\mathrm l})$.

  Take two placed motifs $M,M'$ of the same group. If $V(M)\cap V(M')=\varnothing$, then $\mathcal X_M$ is a function of
  $(Z_{V(M)},(A_e)_{e\in M})$, and each $A_e$ depends only on the labels of its endpoints and on independent
  edge noise, so $\mathcal X_M$ and $\mathcal X_{M'}$ are independent. Otherwise $U:=M\uplus M'$, with multiplicities
  $a_e:=m_e+m'_e$, is connected. Given $Z$, the edges are independent, so by \eqref{eq:a:moments}
  \[
    |\Exp[\mathcal X_M\mathcal X_{M'}]|=\Big|\Exp\Big[B_MB_{M'}\prod_e\Exp[J_e^{a_e}\mid Z]\Big]\Big|
    \le\Big(\frac{2\Delta}{q_0}\Big)^{\sum_ec(a_e)}\|B_M\|_\infty\|B_{M'}\|_\infty ,
  \]
  and $|\Exp \mathcal X_M\,\Exp \mathcal X_{M'}|$ obeys the same bound because $c(m)+c(m')\ge c(m+m')$ (with $c(0):=0$) and
  $2\Delta/q_0\le1$. As $c\ge2$, $c(a)\le a+1$ and $|M|,|M'|\le D$, we have $2d(U)\le\sum_ec(a_e)\le4D$, so
  \begin{equation}
    \label{eq:c:cov}
    |\Cov(\mathcal X_M,\mathcal X_{M'})|\le2\Big(\frac{2}{q_0}\Big)^{4D}\Delta^{\sum_ec(a_e)}\|B_M\|_\infty\|B_{M'}\|_\infty
    \le2\Big(\frac{2}{q_0}\Big)^{4D}\Delta^{2d(U)}\|B_M\|_\infty\|B_{M'}\|_\infty .
  \end{equation}
  Fix the isomorphism types of $M$ and $M'$ and the way their vertices overlap; for $|M|,|M'|\le D$ there
  are finitely many such triples, their number depending only on $D$, and each has at most $n^{v(U)}$
  placements in $[n]$. As $U$ is connected, $v(U)\le d(U)+1$ and
  $n^{v(U)}\Delta^{2d(U)}\le n^{d(U)+1}\Delta^{2d(U)}=n\delta^{d(U)}$. For small pairs $d(U)\ge d(M)\ge v(M)-1\ge K-1$, so
  $n\delta^{d(U)}\le n\delta^{K-1}=N$, and the product of coefficients is at most $CL^2\mu^2$. For large pairs
  $d(U)\ge v(M)-1\ge r$, so $n\delta^{d(U)}\le n\delta^r$, and the product is at most $CL^2\eta^2$. Summing over
  the finitely many triples gives the claim.
\end{proof}

\subsubsection{Remainder}
Let $\Rem_D:=\Sw r-\Poly_D$ and let $\mathcal G$ be the small-field event of \Cref{lem:small-field}.

\begin{lemma}[Remainder]
  \label{lem:remainder}
  Under the conditions of \Cref{lem:info}, on $\mathcal G$
  \[
    |\Rem_D|\le C\,\eta\,n\,(n\Delta)^r(\sqrt\delta\log n)^{D+1-r},
  \]
  and everywhere $|\Rem_D|\le Cn^{2D}$, with $C$ depending only on $(K,q_0,r,D,L)$.
\end{lemma}

\begin{proof}
  By \eqref{eq:a:local} and \Cref{lem:connected}, $\Sw r=\sum_i\sum_{|S|\le r}\Psi_{i,S}$ with
  $\Psi_{i,S}=g_{i,S}(y)$, $y:=(x,f)$, and $g_{i,S}=\sum_Bx^Bh_B(x_B,f)$ over the edge sets $B$ with $(\Omega,B)$ connected.
  Because $f$ is linear in $J$, the Taylor polynomial of degree $D$ of $\Psi_{i,S}$ in $J$ is the Taylor
  polynomial of degree $D$ of $g_{i,S}$ in $y$ evaluated at $y=y(J)$, and $\mathrm{Tay}_D(x^Bh)=x^B\mathrm{Tay}_{D-|B|}h$, where $\mathrm{Tay}_k$
  denotes the Taylor polynomial of degree $k$ and $\mathrm{Tay}_k:=0$ for $k<0$. Hence
  \begin{equation*}
    \Rem_D=\sum_i\sum_{|S|\le r}\sum_Bx^B\big(h_B-\mathrm{Tay}_{D-|B|}h_B\big)(y).
  \end{equation*}

  \emph{Scale on $\mathcal G$.} On $\mathcal G$, $F^{\mathrm{ext}}_v$ differs from $F_v$ by the weights of at most $r$ edges
  inside $\Omega$, each of modulus at most $2\Delta/q_0$, so $\|F^{\mathrm{ext}}_v\|_\infty\le\sqrt\delta\log n+2r\Delta/q_0$; also
  $\|x\|_1\le\binom{r+1}2\,2\Delta/q_0$. Since $\Delta/(\sqrt\delta\log n)=1/(\sqrt n\log n)\to0$, choose $n_M$ so that
  for $n\ge n_M$ each of $2r\Delta/q_0$, $\binom{r+1}22\Delta/q_0$ and $2\Delta/q_0$ is at most $\sqrt\delta\log n$. Then
  $\|y\|_1\le c_y\sqrt\delta\log n$ with $c_y:=2(r+1)K+1$, and on the segment $\{ty:t\in[0,1]\}$,
  $\exp(2(\|f\|_\infty+\|x\|_1))\le\exp(6\sqrt\delta\log n)\le2$ as soon as $\sqrt\delta\log n\le(\log2)/6$.

  \emph{Each term.} By Taylor's formula with integral remainder and \Cref{lem:derivatives}, for
  $|B|\le D$, $|(h_B-\mathrm{Tay}_{D-|B|}h_B)(y)|\le C\|\dot\theta\|_1\|y\|_1^{D+1-|B|}$; for $|B|\ge D+1$, $\mathrm{Tay}_{D-|B|}h_B=0$ and
  $|h_B(y)|\le C\|\dot\theta\|_1$. Let $s':=|\Omega|-1$; connectedness of $(\Omega,B)$ forces $|B|\ge s'$. As
  $|x^B|\le(2\Delta/q_0)^{|B|}$ and $2\Delta/q_0\le\sqrt\delta\log n\le1$, in both cases
  \[
    |x^B|\cdot\big|(h_B-\mathrm{Tay}_{D-|B|}h_B)(y)\big|\le C\|\dot\theta\|_1\,\Delta^{s'}(\sqrt\delta\log n)^{D+1-s'} .
  \]

  \emph{Counting.} For fixed $i$ and $s'\le r$ there are at most $2\binom{n-1}{s'}\le2n^{s'}$ sets $S$ with
  $|\{i\}\cup S|=s'+1$ (the set $\Omega$ has $\binom{n-1}{s'}$ choices, and $S=\Omega$ or $S=\Omega\setminus\{i\}$), and each pair $(i,S)$
  has at most $2^{\binom{r+1}2}$ sets $B$. On $\mathcal G$, therefore,
  \[
    |\Rem_D|\le C\|\dot\theta\|_1\,n\sum_{s'=0}^r(n\Delta)^{s'}(\sqrt\delta\log n)^{D+1-s'}
    \le C(r+1)L\eta\,n\,(n\Delta)^r(\sqrt\delta\log n)^{D+1-r},
  \]
  because $n\Delta=\sqrt{n\delta}\ge\sqrt\delta\log n$, so that the term $s'=r$ is the largest.

  \emph{Everywhere.} For real $J$, $|\Psi_{i,S}|\le2^r\max_c|\ell_c|\le2^rL$ by \Cref{lem:derivatives} and
  (H4), so $|\Sw r|\le n\sum_{j\le r}\binom nj2^rL\le Cn^{r+1}$. Since $|J_e|\le1$, the number of motifs with
  $|M|\le D$ is at most $2n^{2D}$, and each coefficient is bounded by \Cref{lem:uniform}, we get
  $|\Poly_D|\le Cn^{2D}$. As $r+1\le2D$, $|\Rem_D|\le Cn^{2D}$.
\end{proof}

\subsubsection{Proof of \texorpdfstring{\Cref{lem:info}}{the information inequality}}
The pieces combine as follows.

\begin{proof}
  By \Cref{lem:sweep}, $\Exp_u\Sw r=0$, so $\Exp\Poly_D=-\Exp\Rem_D$ and
  \begin{align*}
    \Exp_u\Sw r^2&\le2\Exp\Poly_D^2+2\Exp\Rem_D^2=2\Var(\Poly_D)+2(\Exp\Rem_D)^2+2\Exp\Rem_D^2\\
    &\le2\Var(\Poly_D)+4\Exp\Rem_D^2 .
  \end{align*}
  By \Cref{lem:remainder,lem:small-field},
  \[
    \Exp\Rem_D^2\le C\eta^2n^2(n\Delta)^{2r}(\delta\log^2n)^{D+1-r}+Cn^{4D}\cdot2Kn\,e^{-q_0^2\log^2n/50},
  \]
  and $n^2(n\Delta)^{2r}\delta^{D+1-r}=n^2(n\delta)^r\delta^{D+1-r}=n^{r+2}\delta^{D+1}$. Together with \Cref{lem:pairing} and
  $I(u)\le\Exp_u\Sw r^2$ this gives \eqref{eq:c:info}.
\end{proof}

\subsection{Proof of \texorpdfstring{\Cref{thm:main}}{Theorem 1}(ii)}
\label{sec:proof-lower}

\begin{proof}
  The window $\delta\le n^{-\kappa}$ gives $\delta\le1$ and, for large $n$, $\sqrt\delta\log n\le n^{-\kappa/2}\log n\le1/C_M$.
  In each of the three cases below we choose a fiber, a sweep order $r$, a Taylor degree $D$ and an
  amplitude; $C:=C_M$ is the constant of \Cref{lem:info} for $(K,q_0,r,D,L)$, with $L$ the constant
  of the fiber lemma. Every family $u\mapsto P_{\theta(u)}$ below lives on the finite set of verdict graphs,
  is positive because $0<q<p<1$ and $\theta_c(u)>0$, and is continuously differentiable because $P_u(A)$ is
  a polynomial in the analytic path $\theta(u)$; so \Cref{lem:hellinger} applies.

  \emph{Case $2\le s\le K-1$.} Let
  \[
    r:=\Big\lfloor\frac{K-1}s+\frac{s-1}{\kappa s}\Big\rfloor+1,\qquad
    D:=\Big\lceil\frac{r+2}\kappa+\frac{K-1}s\Big\rceil+1,
  \]
  so that $r>(K-1)/s+(s-1)/(\kappa s)$ and $D\ge r+1$. Use \Cref{lem:fiber-mixed}, with $\mu=\eps^s$ and
  $\eta=\eps$. Let $a_0:=\min\{1,(4C)^{-1/(2s)}\}$, $\eps_1:=\min\{\eps_F,\eps_H\}$ and
  $\eps:=\min\{\eps_1,a_0N^{-1/(2s)}\}$. Since $\eps\le a_0N^{-1/(2s)}$, that is,
  $\eps^2\le a_0^2n^{-1/s}\delta^{-(K-1)/s}$, the four terms of \eqref{eq:c:info} satisfy:
  \begin{itemize}[leftmargin=*,itemsep=1pt,topsep=2pt]
    \item $CN\eps^{2s}\le Ca_0^{2s}\le1/4$;
    \item $Cn\delta^r\eps^2\le Ca_0^2n^{1-1/s}\delta^{r-(K-1)/s}\le Cn^{1-1/s-\kappa(r-(K-1)/s)}$, because
      $r-(K-1)/s>0$ and $\delta\le n^{-\kappa}$; by the choice of $r$, $\kappa(r-(K-1)/s)>1-1/s$, so the exponent is
      negative;
    \item $C\eps^2n^{r+2}\delta^{D+1}(\log n)^{2(D+1-r)}\le Cn^{r+2-1/s-\kappa(D+1-(K-1)/s)}(\log n)^{2D}$; by the choice
      of $D$, $D+1-(K-1)/s\ge(r+2)/\kappa+2$, so the exponent is at most $-1/s-2\kappa<0$;
    \item $Cn^{4D+2}e^{-q_0^2\log^2n/50}\to0$.
  \end{itemize}
  Hence there is $n_L$ such that for $n\ge n_L$, $\sup_{u\in[1/4,3/4]}I(u)\le1$ and $\int_{1/4}^{3/4}I\le1/2$. By
  \Cref{lem:hellinger},
  \[
    \TV(P_{\theta(1/4)},P_{\theta(3/4)})\le\tfrac12\sqrt{1/2}<\tfrac12 .
  \]
  Both points lie in
  $\Rare s\tau$ (\Cref{lem:fiber-mixed}(i)) and their entropies differ by at least $h_s\eps/2$ (part (iv)), so
  \Cref{lem:lecam} gives
  \[
    R^*(\Rare s\tau)\ge\frac{h_s^2\eps^2}{64}=\frac{h_s^2}{64}\min\{\eps_1^2,a_0^2N^{-1/s}\}\ge c_L\min\{1,N^{-1/s}\},
    \quad c_L:=\frac{h_s^2}{64}\min\{\eps_1^2,a_0^2\}.
  \]
  When $\eps=\eps_1$, which happens when $N$ is bounded, this is the constant-risk part of the bound; the
  estimates above hold in both regimes.

  \emph{Case $s=1$.} Let $r:=K$ and $D:=\lceil(K+2)/\kappa\rceil+K$, and use \Cref{lem:fiber-one} with
  $\mu=\eta=\eps$. Let $a_0:=\min\{\frac12,(4C)^{-1/2}\}$, $\eps_1:=\min\{\eps_F,\eps_H,e^{-1}\}$ and
  $\eps:=\min\{\eps_1,a_0N^{-1/2}\}$, so that $N\eps^2\le a_0^2$. Then:
  \begin{itemize}[leftmargin=*,itemsep=1pt,topsep=2pt]
    \item $CN\eps^2\le1/4$;
    \item $Cn\delta^K\eps^2=C\delta\cdot N\eps^2\le Ca_0^2n^{-\kappa}\to0$;
    \item $C\eps^2n^{K+2}\delta^{D+1}(\log n)^{2(D+1-K)}=C(N\eps^2)\,n^{K+1}\delta^{D+2-K}(\log n)^{2(D+1-K)}
      \le Cn^{K+1-\kappa(D+2-K)}(\log n)^{2D}$, and $\kappa(D+2-K)\ge K+2+2\kappa$, so this tends to $0$;
    \item the last term tends to $0$.
  \end{itemize}
  Hence, for large $n$, $\sup_{u\in[0,1]}I(u)\le1$ and $\TV(P_{\theta(0)},P_{\theta(1)})\le1/2$. The entropy gap is at
  least $\frac12\eps\log(1/\eps)$, and \Cref{lem:lecam} gives $R^*\ge\eps^2\log^2(1/\eps)/64$. If $\eps=\eps_1$,
  then $R^*\ge\eps_1^2\log^2(1/\eps_1)/64\ge c\,\psi_1(N)$, as $\psi_1\le1$ and $\log(1/\eps_1)\ge1$. If
  $\eps=a_0N^{-1/2}<\eps_1\le e^{-1}$, then $\log(1/\eps)\ge1$ and $\log(1/\eps)=\frac12\log N+\log(1/a_0)\ge\frac12\log N$. For
  $N\ge1$, $\max\{1,\frac12\log N\}\ge\frac14(2+\log N)\ge\frac14\log(e+N)$ because $e^2N\ge e+N$; for $N<1$,
  $\log(e+N)\le2$. In both cases $\log(1/\eps)\ge\frac14\log(e+N)$, so
  $R^*\ge\frac{a_0^2}{64\cdot16}N^{-1}\log^2(e+N)\ge c\,\psi_1(N)$.

  \emph{Case $s=0$.} Let $r:=K$ and $D:=\lceil(K+2)/\kappa\rceil+K$, and use \Cref{lem:fiber-zero} with amplitude $\eta$
  and $\mu=\eta$. Let $a_0:=\min\{1,(4C)^{-1/2}\}$ and $\eta:=\min\{\eta_F,\prod_c\bar y_c,a_0N^{-1/2}\}$, so that
  $N\eta^2\le a_0^2$; the four terms are bounded exactly as for $s=1$, and $\sup_{u\in[0,1]}I(u)\le1$ for large
  $n$. The entropy gap is at least $\eta\Lam(\bar y)/2$, and \Cref{lem:lecam} gives
  $R^*\ge\eta^2\Lam(\bar y)^2/64\ge c\min\{1,N^{-1}\}$.

  In all three cases $r$ and $D$ depend only on $(K,s,\kappa)$, the fiber constants only on $(K,s,\tau)$ and
  $r$, and $C_M$ only on $(K,q_0,r,D,L)$; so $c_L$ and $n_L$ depend only on $(K,s,\tau,q_0,\kappa)$.
\end{proof}

\begin{remark}[Role of the Polynomial Window]
  \label{rem:window}
  The uniform coefficient bound needs no assumption on the macroscopic meanings, and the large motifs are
  controlled through the
  factor $n\delta^r$ of the pairing count alone. At the amplitude $\eps\asymp N^{-1/(2s)}$, with $\mu=\eps^s$ and $\eta=\eps$,
  \[
    N\mu^2\asymp1\qquad\text{and}\qquad n\delta^r\eta^2\asymp n^{1-1/s}\delta^{\,r-(K-1)/s},
  \]
  and on the window $\delta\le n^{-\kappa}$ the second term vanishes uniformly once $r>(K-1)/s+(s-1)/(\kappa s)$. The
  remainder vanishes once $D$ is of order $r/\kappa$, and $r$ itself grows like $1/\kappa$. This is where the polynomial
  window enters: each fixed window needs only fixed orders $r$ and $D$, and decay of $n\Delta^2$ slower than every power
  of $n$ would need orders that grow with $n$. For $s\in\{0,1\}$ the amplitude is $N^{-1/2}$ and $r=K$ suffices, since $n\delta^K\eta^2=\delta N\eta^2\to0$.
\end{remark}

\subsection{Proof of \texorpdfstring{\Cref{cor:global}}{Corollary 1}}

\begin{proof}
  (i) If $\tau\le1/K$, the largest coordinate of every $\theta\in\Simplex$ is at least $1/K\ge\tau$, so
  $\Rare{K-1}\tau=\Simplex$; and $K-1\ge2$ gives $\psi_{K-1}(N)=\min\{1,N^{-1/(K-1)}\}$. Take $s=K-1$ and $\tau=1/K$ in
  both parts of \Cref{thm:main}; under $n\Delta^2\le n^{-\kappa}$, $\min\{n,N\}=N$. Hence for $n\ge n_L$, $q\in[q_0,1-q_0]$ and
  $n\Delta^2\le n^{-\kappa}$, $c_L\min\{1,N^{-1/(K-1)}\}\le R^*(\Simplex)\le C_U\min\{1,N^{-1/(K-1)}\}$, with constants depending only on
  $(K,q_0,\kappa)$.

  (ii) Sufficiency: if $N:=n^K\Delta_n^{2K-2}\to\infty$, then $N\le n^K$ because $\Delta_n\le1$, so
  $\min\{n,N\}\to\infty$; the upper bound of \Cref{thm:main}(i) holds for every channel and gives
  $\sup_{\Rare s\tau}\Exp(\hH-\Ent)^2\le C_U\psi_s(\min\{n,N\})\to0$. Necessity: if $N\not\to\infty$, then along a
  subsequence $N\le M<\infty$, so $\delta_n=(N/n)^{1/(K-1)}\le(M/n)^{1/(K-1)}\le n^{-\kappa_0}$ with
  $\kappa_0:=1/(2(K-1))$ as soon as $n\ge M^2$. \Cref{thm:main}(ii) with $\kappa=\kappa_0$ gives, along that
  subsequence, $R^*(\Rare s\tau)\ge c_L\psi_s(N)\ge c_L\psi_s(M)>0$, since $\psi_s$ is nonincreasing; no sequence of
  estimators is uniformly consistent. With $\Delta_n=n^{-\gamma}$, $N=n^{K-(2K-2)\gamma}\to\infty$ if and only if
  $\gamma<K/(2K-2)$.
\end{proof}

\section{Proof of the Sharp Local Constant}
\label{sec:sharp-proof}

Along the interior fiber $p_{\theta_0}(z)+(-1)^Ku$, with coordinate $u=e_K-e_K(\theta_0)$, only $e_K$ moves and
$\dot P_K=(-1)^{K-1}K$. With the sweep order $r=K$, a pair of overlapping motifs with nonzero coefficients carries the factor
$N\delta^{v-K}\Delta^{2\beta+\chi}$ in its variance count, where $v$ is the number of vertices of their union,
$\beta$ its cycle rank and $\chi$ its excess edge multiplicity. The only pairs with $(v-K,\beta,\chi)=(0,0,0)$ are the
spanning trees of $K$ answers, each paired with itself, and each such tree has coefficient $\dot P_K$ whatever its shape. These self-pairings give $I(u)\le L_0(1+o(1))$ with $L_0:=\binom nKK^K\Delta^{2K-2}/(q(1-q))^{K-1}$. The all-tree average $\hT_K$ is unbiased for $P_K$ with
$N\Var(\hT_K)\to(K!/K^{K-2})(q(1-q))^{K-1}$, so the Cramér--Rao inequality gives the
reverse bound, and along the fiber $I(u)=L_0(1+o(1))=K^2/\Var_u(\hT_K)\cdot(1+o(1))$.

The van Trees inequality with a prior of width of order $h_n$ on the fiber
gives the lower bound. For attainment, the moment map is a local diffeomorphism
near $\theta_0$ because the coordinates are distinct, the derivative of the entropy in $P_K$ is $\Lam/K$ in
modulus, and a delta method controlled by fourth moments gives the constant. The fourth moments
follow from a Bonami-type inequality, and they bound the probability of leaving the
chart by $O(N^{-2})$.

Throughout this section $q\in(0,1)$ is fixed, $q_0:=\min\{q,1-q\}$, $\Delta=\Delta_n\to0$, and only $n$ with $\Delta\le q_0/2$ are
considered. The integral
$\Lam(\theta)=\int_0^\infty\mathrm dt/\pbar_\theta(t)$ is finite exactly when all coordinates are
positive, since the integrand is of order $t^{-K}$ at infinity and at least of order $t^{-1}$ at zero when
some $\theta_c=0$. The identity $\Lam=\partial\Ent/\partial e_K$ is used only at pairwise distinct positive coordinates. For distinct coordinates, partial fractions give
$\Lam(\theta)=-\sum_c\log\theta_c/\prod_{d\ne c}(\theta_d-\theta_c)$, and $\Lam$ is positive and continuous near any
$\theta_0$ with distinct positive coordinates.

\subsection{Sharp Variance and Fourth Moments}

\begin{lemma}[Sharp Variance]
  \label{lem:sharp-var}
  There is $C$ depending only on $(K,q)$ such that for $n\ge2K$, $\Delta\le q_0/2$, $\delta\le1$ and every
  $\theta\in\Simplex$,
  \begin{gather*}
    \Big|N\Var(\hT_K)-\frac{K!}{K^{K-2}}(q(1-q))^{K-1}\Big|\le C\Big(\Delta+\frac1n+\delta\Big),\\
    \Big|N\Var(\hTs_K)-(K-1)!\,(q(1-q))^{K-1}\Big|\le C\Big(\Delta+\frac1n+\delta\Big),
  \end{gather*}
  and $\Var(\hT_k)\le C\delta^{K-k}/N$, $\Var(\hTs_k)\le C\delta^{K-k}/N$ for $2\le k\le K-1$.
\end{lemma}

\begin{proof}
  We use the decomposition in the proof of \Cref{lem:alltree-var},
  \[
    \Var(\hT_K)=\Var(\bar P_K)+|\Trees_K|^{-2}\sum_{E'\ne\varnothing}\Exp\Big[c_{E'}^2\prod_{e\in E'}v_e/\Delta^2\Big].
  \]

  \emph{Top term.} An $E'$ with $|E'|=K-1$ is itself a tree of $\Trees_K$, the only tree containing it, so
  $c_{E'}=1$. For each tree $\tr$, $v_e=(q+\Delta W_e)(1-q-\Delta W_e)=q(1-q)[1+\Delta W_e(1-2q-\Delta)/(q(1-q))]$, so
  $\prod_{e\in\tr}v_e=(q(1-q))^{K-1}(1+O(\Delta))$ uniformly in $Z$. The top term equals
  \[
    \frac{(q(1-q))^{K-1}(1+O(\Delta))}{|\Trees_K|\,\Delta^{2K-2}},\qquad
    N\times\text{top term}=\frac{n^K}{\binom nK}\cdot\frac{(q(1-q))^{K-1}}{K^{K-2}}\,(1+O(\Delta)),
  \]
  and $n^K/\binom nK=K!\,n^K/(n)_K=K!(1+O(1/n))$.

  \emph{Other noise terms.} By the double counting in \Cref{lem:alltree-var}, the terms with $|E'|=j$ sum
  to at most $\binom{K-1}j(2K)^{j+1}\rho^j/n^{j+1}\le C/(n\delta^j)$; for $j\le K-2$, multiplying by $N=n\delta^{K-1}$ gives
  at most $C\delta^{K-1-j}\le C\delta$.

  \emph{Signal term.} $\Var(\bar P_K)\le K^2/(4n)$, so $N\Var(\bar P_K)\le(K^2/4)\delta^{K-1}\le C\delta$.

  \emph{Stars.} Every $k$-set carries exactly $k$ stars, so the signal term of the star average is the
  same $\bar P_k$. A star containing a forest $E'$ is determined by its $k$-set, which contains the vertices of
  $E'$, and by its center, so the double counting of \Cref{lem:alltree-var} applies with $k$ in place of
  $k^{k-2}$ and gives the same bound. A star with $K-1\ge2$ edges determines its center, so the only star
  containing a star $E'$ is $E'$ itself, and the top term has $K\binom nK$ in place of $|\Trees_K|$, which
  gives the constant $K!/K=(K-1)!$.

  Finally, for $k\le K-1$, \Cref{lem:alltree-var} gives $\Var(\hT_k)\le C_k/N_k=C_k\delta^{K-k}/N$, and likewise
  for the star average.
\end{proof}

The two constants differ by the factor $(K-1)!\,K^{K-2}/K!=K^{K-3}$. The terms with $1\le|E'|\le K-2$ mix
signal and noise. Two trees that share $j$ noisy edges share at least $j+1$ vertices, so these terms
are smaller than the top term by the factor $\delta^{K-1-j}$.

The next bound is of the hypercontractive type of Bonami~\cite{bonami1970etude}, and we prove the form we need.

\begin{lemma}[Bonami-Type Inequality]
  \label{lem:bonami}
  Let $Y_1,\dots,Y_m$ be independent with $\Exp Y_i=0$, $0<\Exp Y_i^2<\infty$ and $\Exp Y_i^4\le B(\Exp Y_i^2)^2$ for some
  $B\ge1$. Then every multilinear polynomial $f(Y)=\sum_{S\subseteq[m],\,|S|\le d}a_S\prod_{i\in S}Y_i$ of degree at most $d$
  satisfies $\Exp f^4\le(81B^2)^d(\Exp f^2)^2$.
\end{lemma}

\begin{proof}
  Let $\Chc:=81B^2$. We argue by induction on $m$, for all $d\ge0$ at once. For $m=0$, $f$ is constant and
  $\Exp f^4=(\Exp f^2)^2$. For $m\ge1$ write $f=Y_mg+h$, where $g$ and $h$ do not involve $Y_m$, $g$ has degree at
  most $d-1$ ($g=0$ if $d=0$) and $h$ degree at most $d$. Let $\sigma^2:=\Exp Y_m^2$. By independence and
  $\Exp Y_m=0$,
  \[
    \Exp f^2=\sigma^2\Exp g^2+\Exp h^2,\qquad
    \Exp f^4=\Exp Y_m^4\,\Exp g^4+4\,\Exp Y_m^3\,\Exp[g^3h]+6\sigma^2\,\Exp[g^2h^2]+\Exp h^4 .
  \]
  By induction $\Exp g^4\le\Chc^{d-1}(\Exp g^2)^2$ and $\Exp h^4\le\Chc^d(\Exp h^2)^2$. By Cauchy--Schwarz,
  $|\Exp Y_m^3|\le\sqrt B\sigma^3$ and $\Exp[g^2h^2]\le\Chc^{d-1/2}\Exp g^2\,\Exp h^2$; by Hölder's inequality with
  exponents $4/3$ and $4$, $|\Exp[g^3h]|\le\Chc^{d-3/4}(\Exp g^2)^{3/2}(\Exp h^2)^{1/2}$. With $a:=\sigma^2\Exp g^2$,
  $b:=\Exp h^2$ and $a^{3/2}b^{1/2}\le a(a+b)/2$,
  \[
    \Exp f^4\le\Chc^d\Big[a^2\Big(\frac B{\Chc}+\frac{2\sqrt B}{\Chc^{3/4}}\Big)+ab\Big(\frac{2\sqrt B}{\Chc^{3/4}}+\frac6{\Chc^{1/2}}\Big)+b^2\Big].
  \]
  As $\Chc^{1/2}=9B$ and $\Chc^{3/4}=27B^{3/2}$, the first bracket is $\frac1{81B}+\frac2{27B}\le1$ and the second is
  $\frac2{27B}+\frac2{3B}\le2$. Hence $\Exp f^4\le\Chc^d(a+b)^2=\Chc^d(\Exp f^2)^2$.
\end{proof}

\begin{lemma}[Fourth Moments]
  \label{lem:fourth}
  There is $C$ depending only on $(K,q_0)$ such that for $n\ge2K$, $q\in[q_0,1-q_0]$, $\Delta\le q_0/2$,
  $\delta\le1$, every $\theta\in\Simplex$ and $2\le k\le K$, $\Exp(\hT_k-P_k)^4\le C/N_k^2$; hence
  $\Exp\|\hT-P(\theta)\|^4\le C/N^2$. The same bounds hold for the star averages.
\end{lemma}

\begin{proof}
  Write $\hT_k-P_k=Q_k+(\bar P_k-P_k)$ with $Q_k:=|\Trees_k|^{-1}\sum_{E'\ne\varnothing}c_{E'}\xi_{E'}$ and $\bar P_k=\Exp[\hT_k\mid Z]$.

  \emph{Noise.} Given $Z$, the $\xi_e=(A_e-r_e)/\Delta$ are independent with mean zero, and
  $r_e\in\{q,q+\Delta\}\subseteq[q_0,1-q_0/2]$. A centered $\Bern(r)$ variable satisfies
  $\Exp(A-r)^4/(\Exp(A-r)^2)^2=((1-r)^3+r^3)/(r(1-r))\le1/(r(1-r))\le4/q_0=:B$, because $r(1-r)\ge q_0/4$ on that
  interval. Given $Z$, $Q_k$ is a multilinear polynomial of degree at most $k-1$ in the $\xi_e$ (each
  $\xi_{E'}$ is a product over distinct edges), with coefficients determined by $Z$, so \Cref{lem:bonami}
  gives $\Exp[Q_k^4\mid Z]\le(81B^2)^{k-1}(\Exp[Q_k^2\mid Z])^2$. The noise estimate in the proof of
  \Cref{lem:alltree-var} holds for every $Z$, since it uses only $0\le c_{E'}\le\#\{\tr\supseteq E'\}$ and
  $v_e/\Delta^2\le\rho$, so $\Exp[Q_k^2\mid Z]\le C/N_k$ and $\Exp Q_k^4\le C/N_k^2$.

  \emph{Signal.} $\bar P_k$ is a function of $Z$ with $\Exp \bar P_k=P_k$ that moves by at most $k/n$ when one
  label changes (\Cref{lem:alltree-var}). McDiarmid's inequality gives
  $\Prob(|\bar P_k-P_k|\ge t)\le2\exp(-2nt^2/k^2)$, so
  $\Exp(\bar P_k-P_k)^4=\int_0^\infty4t^3\Prob(|\bar P_k-P_k|\ge t)\,\mathrm dt\le8\int_0^\infty t^3e^{-2nt^2/k^2}\mathrm dt=k^4/n^2$.

  \emph{Combination.} $(x+y)^4\le8(x^4+y^4)$ and $N_k\le n$ for $\delta\le1$ give $\Exp(\hT_k-P_k)^4\le C/N_k^2$. As
  $\|x\|^4\le(K-1)\sum_kx_k^4$ and $N_k\ge N$, $\Exp\|\hT-P\|^4\le C/N^2$. For the star averages the argument is the
  same.
\end{proof}

\subsection{The Interior Fiber and the Sharp Information Bound}

\begin{lemma}[Interior Fiber]
  \label{lem:fiber-int}
  Let $\theta_0\in\Simplex$ have pairwise distinct positive coordinates. There are $w_0>0$, $C_0$ and
  $L_{\theta_0}\ge1$, depending only on $\theta_0$, such that for $|u|\le w_0$ the polynomial $p_{\theta_0}(z)+(-1)^Ku$ has $K$
  positive simple real roots; let $\theta(u)$ be their vector, ordered as $\theta_0$. Then:
  \begin{enumerate}[label=(\roman*),leftmargin=*]
    \item $\theta(u)$ is real analytic, $\theta(u)\in\Simplex$, $\|\theta(u)-\theta_0\|\le C_0|u|$, and the coordinates are
      pairwise distinct and at least $\frac12\min_c\theta_{0,c}$;
    \item $e_K(\theta(u))=e_K(\theta_0)+u$ and the other $e_j$ are constant; $\dot P_j=0$ for $j\le K-1$ and
      $\dot P_K=(-1)^{K-1}K$;
    \item $\|\dot\theta\|_1\le L_{\theta_0}$, $\max_c|\dot\theta_c|/\theta_c\le L_{\theta_0}$ and $|\dot P_k|\le kL_{\theta_0}$; hence for every $r\ge1$ the path
      satisfies~(H) with $\mu=\eta=1$ and $L=rL_{\theta_0}$;
    \item $\frac{\mathrm d}{\mathrm du}\Ent(\theta(u))=\Lam(\theta(u))$.
  \end{enumerate}
\end{lemma}

\begin{proof}
  (i)--(iii) The coordinates of $\theta_0$ are distinct, so the implicit function theorem gives analytic
  simple roots, with $\dot\theta_c=(-1)^{K+1}/p_u'(\theta_c)$ and $|p_u'(\theta_c)|$ bounded below for $|u|\le w_0$; only the
  constant term changes, so the roots still sum to $1$, and \eqref{eq:d:pk} gives $\dot P_K$. (iv) By
  \Cref{lem:intrep} and $\pbar_{\theta(u)}(t)=\pbar_{\theta_0}(t)+u$,
  $\Ent(\theta(u))-\Ent(\theta_0)=\int_0^\infty\log(1+u/\pbar_{\theta_0}(t))\,\mathrm dt$. Take $w_0\le\frac12\prod_c\theta_{0,c}$; then
  $|u|/\pbar_{\theta_0}(t)\le\frac12$, and the $u$-derivative of the integrand, $1/(\pbar_{\theta_0}(t)+u)$, is bounded by
  $2/\max\{\prod_c\theta_{0,c},t^K\}$, which is integrable. Differentiating under the integral gives
  $\int_0^\infty\mathrm dt/\pbar_{\theta(u)}(t)=\Lam(\theta(u))$.
\end{proof}

Take the sweep order $r=K$. By \Cref{lem:small-motif}, motifs on at most $K-1$ vertices have coefficient
zero and motifs on $K$ vertices have the constant coefficient $\partial_u\varkappa_M/M!$; by \Cref{lem:uniform},
$|B^{(K)}_M|\le\Gamma_M\|\dot\theta\|_1\le\Gamma_ML_{\theta_0}$ for every motif.

\emph{Coefficient of a spanning tree.} Let $\tr$ be a simple spanning tree of a $K$-set. In the expansion of the joint
cumulant $\varkappa_\tr$ over the partitions of the edges of $\tr$, the partition with a single block contributes
$\Exp\prod_{e\in\tr}W_e=P_K$; every other partition has blocks that are proper subsets of $\tr$, whose
components have at most $K-1$ vertices, so they contribute only $P_2,\dots,P_{K-1}$. Hence
$\partial\varkappa_\tr/\partial P_K=1$ and $B^{(K)}_\tr=\dot P_K=(-1)^{K-1}K$, whatever the shape of $\tr$.

Let $M,M'$ be two motifs with nonzero coefficients that share a vertex. By
\Cref{lem:local}(iv) both are connected. Let $U:=M\uplus M'$ with multiplicities $a_e:=m_e+m'_e$, $v:=v(U)$,
$d:=d(U)$, the cycle rank $\beta:=d-v+1\ge0$ and the excess edge multiplicity $\chi:=\sum_e(a_e-2)_+$.

\begin{lemma}[Unique Leading Overlap]
  \label{lem:overlap}
  Let $M,M'$ share a vertex, with $B^{(K)}_M$ and $B^{(K)}_{M'}$ not identically zero. Then $v\ge K$, and
  $(v-K,\beta,\chi)=(0,0,0)$ if and only if $M=M'=\tr$ is a simple spanning tree of a $K$-set.
\end{lemma}

\begin{proof}
  Nonzero coefficients force $M,M'$ to be connected with $v(M),v(M')\ge K$ (\Cref{lem:local}(iv) and
  \Cref{lem:small-motif}), so $v\ge K$. Let $v=K$, $\beta=0$ and $\chi=0$. As $V(M),V(M')\subseteq V(U)$ have at
  least $K=|V(U)|$ elements, $V(M)=V(M')=V(U)$. By $\beta=0$ the set of distinct edges of $U$ is a tree on
  these $K$ vertices. $M$ is connected and covers the $K$ vertices, so its distinct edges form a
  connected spanning subgraph of that tree, hence the whole tree; the same holds for $M'$. Therefore every
  edge has $a_e\ge2$, and $\chi=0$ forces $a_e=2$, that is, $m_e=m'_e=1$ and $M=M'=\tr$. The converse is clear.
\end{proof}

Since $c(a)=2+(a-2)_+$ in \eqref{eq:a:moments}, $\sum_ec(a_e)=2d+\chi$, and with $d=v-1+\beta$,
\begin{equation}
  \label{eq:f:factor}
  n^v\Delta^{\sum_ec(a_e)}=n^v\Delta^{2(v-1+\beta)+\chi}=N\,\delta^{v-K}\,\Delta^{2\beta+\chi}.
\end{equation}
By \Cref{lem:overlap}, for every overlap pattern other than the self-pairing of a tree, the factor
$\delta^{v-K}\Delta^{2\beta+\chi}$ is at most $\max\{\delta,\Delta\}$, as $\delta,\Delta\le1$. \Cref{fig:overlap} shows the
self-pairing of a tree and, for each of $v-K$, $\beta$ and $\chi$, a pair in which it equals one and the other two
vanish. Let
\begin{gather*}
  L_0:=\binom nKK^{K-2}\Big(\frac{\Delta^2}{q(1-q)}\Big)^{K-1}\dot P_K^2=\binom nK\frac{K^K\Delta^{2K-2}}{(q(1-q))^{K-1}},\\
  \frac N{L_0}=\frac{n^K}{\binom nK}\cdot\frac{(q(1-q))^{K-1}}{K^K}\longrightarrow\frac{K!}{K^K}(q(1-q))^{K-1}.
\end{gather*}

\begin{figure}[t]
  \centering
  \includegraphics{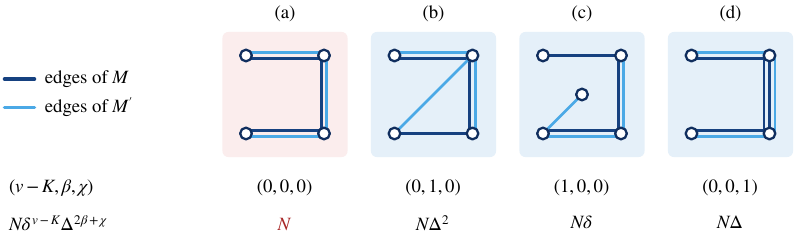}
  \caption{Pairs $(M,M')$ of motifs that share a vertex, for $K=4$, drawn as the multigraph $U=M\uplus M'$ with the
  edges of $M$ dark and those of $M'$ light. Under each pair are the triple $(v-K,\beta,\chi)$ and the factor $N\delta^{v-K}\Delta^{2\beta+\chi}$, where $v$ is the number of vertices, $\beta$ the cycle rank and $\chi$ the excess edge multiplicity
  of $U$. (a) A spanning tree paired with itself. No other pair of motifs with
  nonzero coefficients has the factor $N$. (b) Two spanning trees of the same $4$-set, a path and a star. (c) Spanning trees of two
  $4$-sets that share three answers. (d) A spanning tree with one edge doubled, paired with the tree.}
  \label{fig:overlap}
\end{figure}

\begin{lemma}[Sharp Information Bound]
  \label{lem:sharp-info}
  Fix $K\ge3$, $q\in(0,1)$, $\kappa\in(0,1)$ and $\theta_0\in\Simplex$ with pairwise distinct positive coordinates; let
  $\theta(u)$, $|u|\le w_0$, be the interior fiber of \Cref{lem:fiber-int}, $I(u)$ the Fisher information of the
  verdicts along it, and $D:=\lceil(K+2)/\kappa\rceil+K$. There are $C$ and $n_0$ depending only on $(K,q,\kappa,\theta_0)$
  such that for $n\ge n_0$, $\delta\le n^{-\kappa}$, $N\ge1$ and $|u|\le w_0$,
  \[
    I(u)\le\Exp_u\Sw K^2\le L_0(1+\eps_n),\qquad \eps_n:=C\big(\Delta+\delta+\varrho_n^{1/2}\big),
  \]
  where $\varrho_n:=n^{K+1}\delta^{D+2-K}(\log n)^{2(D+1-K)}+n^{4D+2}e^{-q_0^2\log^2n/50}$; moreover $\eps_n\to0$.
\end{lemma}

\begin{proof}
  With $r=K$ write $\Sw K=\Poly_D+\Rem_D$.

  \emph{Polynomial part.} $\Var(\Poly_D)=\sum_{(M,M')}\Cov(\mathcal X_M,\mathcal X_{M'})$ with $\mathcal X_M:=J^MB^{(K)}_M(Z)$; vertex-disjoint pairs
  contribute nothing (proof of \Cref{lem:pairing}). For the self-pairing of a tree,
  $\Var(\mathcal X_\tr)=\dot P_K^2[\Exp(J^\tr)^2-(\Exp J^\tr)^2]$. With $\beta_J:=\log\frac{p(1-q)}{q(1-p)}=\frac\Delta{q(1-q)}+O(\Delta^2)$,
  $\Var(J_e\mid Z)=\beta_J^2r_e(1-r_e)=\frac{\Delta^2}{q(1-q)}(1+O(\Delta))$, and $|\Exp[J_e\mid Z]|\le2\Delta^2/q_0^2$ by
  \Cref{lem:weights}; hence $\Exp[J_e^2\mid Z]=\frac{\Delta^2}{q(1-q)}(1+O(\Delta))$ uniformly in $Z$ and
  $\Exp(J^\tr)^2=(\frac{\Delta^2}{q(1-q)})^{K-1}(1+O(\Delta))$, while $|\Exp J^\tr|\le(2\Delta/q_0)^{2(K-1)}$, whose square is
  smaller by the factor $O(\Delta^{2(K-1)})$. There are $|\Trees_K|=\binom nKK^{K-2}$ such placements, which gives
  $L_0(1+O(\Delta))$. Every other pair that shares a vertex has, by \eqref{eq:c:cov},
  $|\Cov(\mathcal X_M,\mathcal X_{M'})|\le2(2/q_0)^{4D}\Delta^{\sum_ec(a_e)}\|B_M\|_\infty\|B_{M'}\|_\infty$ with bounded coefficients; for given types and
  overlap pattern there are at most $n^v$ placements, so by \eqref{eq:f:factor} and \Cref{lem:overlap} the
  contribution is at most $CN\max\{\delta,\Delta\}$. There are finitely many triples, their number depending on
  $D$. As $L_0\ge cN$ with $c$ depending only on $(K,q)$, $\Var(\Poly_D)\le L_0(1+C(\Delta+\delta))$.

  \emph{Remainder.} By \Cref{lem:fiber-int}, (H) holds with $\eta=1$, and \Cref{lem:remainder} with the
  computation in the proof of \Cref{lem:info} gives
  $\Exp\Rem_D^2\le Cn^{K+2}\delta^{D+1}(\log n)^{2(D+1-K)}+Cn^{4D+2}e^{-q_0^2\log^2n/50}$; dividing by $L_0\ge cN=cn\delta^{K-1}$
  and using $N\ge1$ gives $\Exp\Rem_D^2\le CL_0\varrho_n$. Since $\delta\le n^{-\kappa}$ and
  $\kappa(D+2-K)\ge K+2+2\kappa$, $\varrho_n\le n^{-1-2\kappa}(\log n)^{2D}+n^{4D+2}e^{-q_0^2\log^2n/50}\to0$.

  \emph{Combination.} $\Exp\Poly_D=-\Exp\Rem_D$, so $\Exp\Poly_D^2\le\Var(\Poly_D)+\Exp\Rem_D^2$. As $\Sw K=\Poly_D+\Rem_D$ and
  $2|\Poly_D\Rem_D|\le\varrho_n^{1/2}\Poly_D^2+\varrho_n^{-1/2}\Rem_D^2$,
  \begin{align*}
    \Exp_u\Sw K^2&\le(1+\varrho_n^{1/2})\Exp\Poly_D^2+(1+\varrho_n^{-1/2})\Exp\Rem_D^2\\
    &\le(1+\varrho_n^{1/2})\Var(\Poly_D)+(2+\varrho_n^{1/2}+\varrho_n^{-1/2})\Exp\Rem_D^2 ,
  \end{align*}
  which gives $\Exp_u\Sw K^2\le L_0[1+C(\Delta+\delta+\varrho_n^{1/2})]$ once $\varrho_n\le1$, and
  $I(u)\le\Exp_u\Sw K^2$ by \Cref{lem:sweep}. The constants are uniform in $|u|\le w_0$, because they depend on
  the fiber only through $L_{\theta_0}$ and $\theta_0$.
\end{proof}

\begin{remark}[Information Equality]
  \label{rem:info-eq}
  Along the interior fiber of \Cref{lem:fiber-int} with the coordinate $u=e_K-e_K(\theta_0)$, under the
  conditions of \Cref{lem:sharp-info},
  \[
    I(u)=L_0(1+o(1))\qquad\text{and}\qquad K^2/\Var_u(\hT_K)=L_0(1+o(1))\qquad\text{uniformly in }|u|\le w_0 .
  \]
  Indeed, $\hT_K$ is unbiased for $P_K(\theta(u))$, and $P_u$ is a positive family on a finite set,
  differentiable in $u$; the Cramér--Rao inequality gives $\Var_u(\hT_K)\ge\dot P_K^2/I(u)$, that is,
  $I(u)\ge K^2/\Var_u(\hT_K)$. By \Cref{lem:sharp-var},
  $K^2/\Var_u(\hT_K)=K^2N/[(K!/K^{K-2})(q(1-q))^{K-1}](1+O(\Delta+1/n+\delta))$, which is $L_0(1+O(\Delta+1/n+\delta))$ because
  $L_0=NK^K/(K!(q(1-q))^{K-1})\cdot(1+O(1/n))$. Together with \Cref{lem:sharp-info}, this gives $I(u)=L_0(1+o(1))$. The proof of
  the sharp constant does not use this lower bound on $I(u)$.
\end{remark}

\subsection{Lower Bound via the Van Trees Inequality}

We use the one-dimensional van Trees inequality \cite{gill1995applications}: let $\varpi$ be an
absolutely continuous density on $[a,b]$ with $\varpi(a)=\varpi(b)=0$ and $\mathcal I(\varpi):=\int(\varpi')^2/\varpi<\infty$, let
$u\mapsto P_u(A)$ be positive and continuously differentiable for each $A$, and let $\zeta$ be continuously
differentiable. Then every estimator $T=T(A,\Xi)$ satisfies
\[
  \int_a^b\Exp_u(T-\zeta(u))^2\,\varpi(u)\,\mathrm du\ \ge\ \frac{\big(\int_a^b\zeta'(u)\varpi(u)\,\mathrm du\big)^2}{\int_a^bI(u)\varpi(u)\,\mathrm du+\mathcal I(\varpi)} ,
\]
where the auxiliary randomness $\Xi$, independent of $A$, does not change the Fisher information.

\begin{proposition}[Van Trees Lower Bound]
  \label{prop:vantrees}
  Let $K\ge3$, $q\in(0,1)$, $\kappa\in(0,1/(K-1))$ and $\theta_0\in\Simplex$ with pairwise distinct positive
  coordinates. Let $\Delta_n>0$ satisfy $n\Delta_n^2\le n^{-\kappa}$ and $N\to\infty$, and let $h_n\to0$ with
  $h_n\sqrt{N}\to\infty$. Then $\liminf_nN\Rloc(\theta_0;h_n)\ge \mathcal V(\theta_0)$.
\end{proposition}

\begin{proof}
  Let $\varpi_1(v):=\cos^2(\pi v/2)$ on $[-1,1]$, a density with $\mathcal I(\varpi_1)=\pi^2$; let $w_n:=\min\{w_0,h_n/C_0\}$
  and $\varpi_n(u):=w_n^{-1}\varpi_1(u/w_n)$, so $\mathcal I(\varpi_n)=\pi^2/w_n^2$. By \Cref{lem:fiber-int}(i),
  $\|\theta(u)-\theta_0\|\le h_n$ for $|u|\le w_n$, so the local minimax risk is at least the Bayes risk under $\varpi_n$
  along the fiber. Let $\zeta(u):=\Ent(\theta(u))$. Then $\zeta'=\Lam(\theta(u))$ by \Cref{lem:fiber-int}(iv), and the van
  Trees inequality with \Cref{lem:sharp-info} gives
  \[
    \Rloc(\theta_0;h_n)\ge\frac{\big(\min_{|u|\le w_n}\Lam(\theta(u))\big)^2}{L_0(1+\eps_n)+\pi^2/w_n^2}.
  \]
  For large $n$, $w_n=h_n/C_0$. Since $\binom nKK^K\ge n^K/K!$ for $n\ge2K$, we have $L_0\ge N/(K!\,(q(1-q))^{K-1})$, so
  $(\pi^2/w_n^2)/L_0\le\pi^2C_0^2K!\,(q(1-q))^{K-1}/(h_n^2N)\to0$. Moreover, $\min_{|u|\le w_n}\Lam(\theta(u))\to\Lam(\theta_0)$
  and $N/L_0\to(K!/K^K)(q(1-q))^{K-1}$. Hence $\liminf N\Rloc\ge\Lam(\theta_0)^2(K!/K^K)(q(1-q))^{K-1}=\mathcal V(\theta_0)$.
\end{proof}

\subsection{Attainment}

Let $\Sigma:=\{\theta\in\Real^K:\sum_c\theta_c=1\}$ and $\Pi(\theta):=(P_2(\theta),\dots,P_K(\theta))$.

\begin{lemma}[Local Chart]
  \label{lem:chart}
  Let $\theta_0\in\Simplex$ have pairwise distinct positive coordinates. There are $\rho_1>0$ and $r_U>0$, depending only
  on $\theta_0$, and a $C^\infty$ function $\tilde H$ on $B_1:=B(\Pi(\theta_0),\rho_1)\subset\Real^{K-1}$ such that:
  \begin{enumerate}[label=(\roman*),leftmargin=*]
    \item for $\theta\in\Simplex$ with $\|\theta-\theta_0\|\le r_U$, $\Pi(\theta)\in B(\Pi(\theta_0),\rho_1/2)$ and $\Ent(\theta)=\tilde H(\Pi(\theta))$;
    \item if $\hT\in B_1$, then $\hT\in\Mom$ and \Cref{alg:main} outputs $\hH=\tilde H(\hT)$;
    \item $\partial\tilde H/\partial P_K\,(\Pi(\theta))=(-1)^{K-1}\Lam(\theta)/K$ for $\theta$ as in (i);
    \item $\sup_{B_1}(\|\nabla\tilde H\|+\|\nabla^2\tilde H\|)<\infty$.
  \end{enumerate}
\end{lemma}

\begin{proof}
  The differential of $\Pi$ on the tangent space $\{v:\sum_cv_c=0\}$ of $\Sigma$ is $v\mapsto(k\sum_c\theta_c^{k-1}v_c)_{k=2}^K$. If it
  maps $v$ to zero, then the $K\times K$ matrix $[k\theta_c^{k-1}]_{k,c}$, whose row $k=1$ encodes $\sum_cv_c=0$, maps $v$
  to zero; this matrix is $\mathrm{diag}(1,\dots,K)$ times the transposed Vandermonde matrix of $\theta$, which is
  nonsingular for distinct coordinates, so $v=0$. By the inverse function theorem, there are open
  neighborhoods $W$ of $\theta_0$ in $\Sigma$ and $\mathcal O$ of $\Pi(\theta_0)$ in $\Real^{K-1}$ such that $\Pi:W\to\mathcal O$ is a $C^\infty$
  diffeomorphism; shrink $W$ into the interior of $\Simplex$, with pairwise distinct coordinates bounded
  below. Let $\tilde H:=\Ent\circ(\Pi|_W)^{-1}$, which is $C^\infty$ because $\Ent$ is smooth on the
  positive orthant. Choose $\rho_1$ with $\bar B(\Pi(\theta_0),\rho_1)\subset\mathcal O$, which gives (iv) by compactness, and
  $r_U$ such that $\{\theta\in\Sigma:\|\theta-\theta_0\|\le r_U\}\subset W$ and $\Pi$ maps this set into $B(\Pi(\theta_0),\rho_1/2)$,
  which gives (i).
  (ii) $\hT\in B_1\subset\Pi(W)\subseteq\Mom$, so the nearest point is $\hT$ itself, and every $\htheta$ with $P(\htheta)=\hT$ has
  the multiset of coordinates of $(\Pi|_W)^{-1}(\hT)$ by Newton's identities; so $\hH=\tilde H(\hT)$.
  (iii) On $W$, Newton's identities make $e_1,\dots,e_{K-1}$ functions of $P_2,\dots,P_{K-1}$ alone, and, by
  \eqref{eq:d:pk}, $P_K=f(e_1,\dots,e_{K-1})+(-1)^{K-1}Ke_K$; so, with $P_2,\dots,P_{K-1}$ fixed,
  $\partial e_K/\partial P_K=(-1)^{K-1}/K$ and the other $e_j$ do not move. Moving $e_K$ by $t$ with $e_1,\dots,e_{K-1}$ fixed
  changes $\Ent$ by $\int_0^\infty\log(1+t/\pbar_\theta(t'))\,\mathrm dt'$ (\Cref{lem:intrep}), whose derivative at $t=0$ is
  $\Lam(\theta)$ by the dominated convergence argument of \Cref{lem:fiber-int}(iv). The chain rule gives (iii).
\end{proof}

\begin{proposition}[Attainment]
  \label{prop:attain}
  Let $q\in(0,1)$, let $\theta_0\in\Simplex$ have pairwise distinct positive coordinates, and let $r_U$ be as in
  \Cref{lem:chart}. There are $C$ and $n_0$ depending only on $(K,q,\theta_0)$ such that for $n\ge n_0$, $\Delta\le q_0/2$,
  $\delta\le1$, $N\ge1$ and every $\theta\in\Simplex$ with $\|\theta-\theta_0\|\le r_U$,
  \begin{gather*}
    \big|N\Exp_\theta(\hH-\Ent(\theta))^2-\mathcal V(\theta)\big|\le C\big(\Delta+\delta^{1/2}+n^{-1}+N^{-1/2}\big),\\
    \big|N\Exp_\theta(\hH^\star-\Ent(\theta))^2-K^{K-3}\mathcal V(\theta)\big|\le C\big(\Delta+\delta^{1/2}+n^{-1}+N^{-1/2}\big).
  \end{gather*}
\end{proposition}

\begin{proof}
  Fix $\theta$ and write $P:=\Pi(\theta)$, $g:=\nabla\tilde H(P)$, $a:=g\cdot(\hT-P)$ and $\mathcal E:=\{\|\hT-P\|<\rho_1/2\}$. On $\mathcal E$,
  $\hT\in B_1$ because $P\in B(\Pi(\theta_0),\rho_1/2)$, so, by \Cref{lem:chart},
  $\hH-\Ent(\theta)=\tilde H(\hT)-\tilde H(P)=a+R$ with $|R|\le C_2\|\hT-P\|^2$, $C_2:=\frac12\sup_{B_1}\|\nabla^2\tilde H\|$, the segment from
  $P$ to $\hT$ lying in the convex set $B_1$. By \Cref{lem:fourth} and Markov's inequality,
  $\Prob(\mathcal E^c)\le(2/\rho_1)^4\Exp\|\hT-P\|^4\le C/N^2$, $\Exp R^2\mathbf 1_{\mathcal E}\le C/N^2$ and
  $\Exp a^4\le\|g\|^4\Exp\|\hT-P\|^4\le C/N^2$.

  \emph{Upper bound.} As $N\ge1$, $(a+R)^2\le(1+N^{-1/2})a^2+(1+N^{1/2})R^2$, and $(\hH-\Ent)^2\le(\log K)^2$ on $\mathcal E^c$,
  so $\Exp(\hH-\Ent)^2\le(1+N^{-1/2})\Exp a^2+2CN^{-3/2}+C(\log K)^2/N^2$.

  \emph{Lower bound.} $(a+R)^2\ge(1-N^{-1/2})a^2-N^{1/2}R^2$, and
  $\Exp a^2\mathbf 1_{\mathcal E}\ge\Exp a^2-(\Exp a^4)^{1/2}\Prob(\mathcal E^c)^{1/2}\ge\Exp a^2-C/N^2$, so
  $\Exp(\hH-\Ent)^2\ge(1-N^{-1/2})\Exp a^2-CN^{-3/2}$.

  \emph{Main term.} $\Exp a^2=\Var(g\cdot\hT)=\sum_{k,l}g_kg_l\Cov(\hT_k,\hT_l)$. For $(k,l)\ne(K,K)$,
  $|\Cov(\hT_k,\hT_l)|\le(\Var\hT_k\,\Var\hT_l)^{1/2}\le C\delta^{1/2}/N$ by \Cref{lem:sharp-var}. By
  \Cref{lem:chart}(iii), $g_K^2=\Lam(\theta)^2/K^2$, and by \Cref{lem:sharp-var},
  \[
    N\Exp a^2=\frac{\Lam(\theta)^2}{K^2}\cdot\frac{K!}{K^{K-2}}(q(1-q))^{K-1}+O(\Delta+n^{-1}+\delta+\delta^{1/2})=\mathcal V(\theta)+O(\Delta+n^{-1}+\delta^{1/2}),
  \]
  uniformly in $\|\theta-\theta_0\|\le r_U$, since $\|g\|$ and $\Lam$ are bounded there; in particular $N\Exp a^2\le C$.
  The two bounds give $|N\Exp(\hH-\Ent)^2-N\Exp a^2|\le CN^{-1/2}$, which proves the first
  display. For the star average the proof is the same: \Cref{lem:sharp-var,lem:fourth} hold for
  $\hTs$, with $(K-1)!$ in place of $K!/K^{K-2}$, and $\frac{\Lam^2}{K^2}(K-1)!(q(1-q))^{K-1}=K^{K-3}\mathcal V(\theta)$.
\end{proof}

\subsection{Proofs of \texorpdfstring{\Cref{thm:sharp}}{Theorem 2} and \texorpdfstring{\Cref{cor:star}}{Corollary 2}}

\begin{proof}
  \emph{Upper bound.} As $h_n\to0$, $h_n\le r_U$ for large $n$. The conditions $\delta_n\le n^{-\kappa}$ and $N\to\infty$
  give $\delta_n\to0$ and $\Delta_n\le\sqrt{\delta_n}\to0$. By \Cref{prop:attain},
  $N\sup_{\|\theta-\theta_0\|\le h_n}\Exp_\theta(\hH-\Ent(\theta))^2\le\sup_{\|\theta-\theta_0\|\le h_n}\mathcal V(\theta)+o(1)\to \mathcal V(\theta_0)$, since $\mathcal V$ is continuous at
  $\theta_0$. Hence $\limsup N\Rloc(\theta_0;h_n)\le \mathcal V(\theta_0)$, and the local worst-case risk of \Cref{alg:main} has the same
  upper limit.

  \emph{Lower bound.} This is \Cref{prop:vantrees}.

  Together, the two bounds give $N\Rloc(\theta_0;h_n)\to \mathcal V(\theta_0)$. As $\Rloc$ is at most the local worst-case risk of $\hH$,
  that risk times $N$ also tends to $\mathcal V(\theta_0)$. \Cref{cor:star} follows in the same way from the second
  display of \Cref{prop:attain}: $N$ times the local worst-case risk of $\hH^\star$ has upper limit at most $\sup_{\|\theta-\theta_0\|\le h_n}K^{K-3}\mathcal V(\theta)+o(1)$, and $N$ times its risk at
  $\theta=\theta_0$ is $K^{K-3}\mathcal V(\theta_0)+o(1)$.
\end{proof}

\section{Numerical Checks of the Mechanisms}
\label{sec:numerics}

The computations below check the mechanisms that the proofs use:
the all-tree and star averages, the coefficient structure of the sweep, the fibers, the transfer, and the variance
and information constants. All checks are implemented in Python 3.12, with mpmath~\cite{mpmath2023} for arbitrary precision,
NumPy~\cite{harris2020array} for double and 80-bit extended precision, SciPy~\cite{virtanen2020scipy} for the
Nelder--Mead searches~\cite{nelder1965simplex} and NetworkX~\cite{hagberg2008exploring} for the enumeration of the
motif graphs.
The code of every check, figure and table, the saved outputs, the seeds and the pinned package versions, with
instructions that reproduce the reported summaries, are available at
\url{https://github.com/yahiko-l/semantic-entropy-rarity}.

\subsection{All-Tree and Star Averages}
\label{sec:num-trees}

The all-tree averages, implemented through homomorphism counts, and the star averages
are checked against exhaustive enumeration over
every $k$-set with all its labeled spanning trees and all its stars, for $k=2,\dots,5$. The test inputs are $50$ symmetric
matrices with zero diagonal ($n\in\{5,6,8,10,12\}$, five kinds of entries including Cauchy draws and
centered verdicts, two draws each). The entries are dyadic rationals, so the references are exact
integer computations. The largest relative error of the floating-point implementation is
$3.6\times10^{-11}$ for the all-tree averages and $5.3\times10^{-15}$ for the star averages, and the same code
run on exact rational input reproduces the references exactly.

\subsection{Coefficients of the Sweep}
\label{sec:num-coefficients}

We evaluate the exact expressions of the Taylor coefficients $B^{(5)}_M(z)$ of the sweep of order $5$ in double precision
(IEEE 754 binary64),
over all label assignments $z$ of each motif, for $K=5$ and $s\in\{2,3,4\}$. The fiber is that
of \Cref{lem:fiber-mixed} with $\bar y$ replaced by $(0.5,0.3,0.2)$, $(0.6,0.4)$ and $(1)$ for $s=2,3,4$. These
entries are distinct and exceed $\tau$ for every $\tau<0.2$. The fiber moves the constant term by $0.9(2u-1)\sigma_\eps$, a
linear reparametrization under which derivatives in $u$ are $1.8$ times
those of the lemma. All quantities are evaluated at the midpoint $u=1/2$, where
$\|\dot\theta\|_1/\eps$ lies in $[0.588,0.775]$ for all $s$ and $\eps$. The values of $\eps$ are $10^{-2}$,
$3\times10^{-3}$, $10^{-3}$, $3\times10^{-4}$ and $10^{-4}$.

The motifs are of three kinds. The first kind comprises all $52$
connected graphs on six vertices with at least nine edges, including the complete graph. The second starts from eight
small shapes: the star and the path on four vertices, the path on five vertices, the triangle, the triangle with a
pendant edge, the four-cycle, the star on five vertices and the spider with legs $(2,1,1)$. It comprises one motif per
edge orbit of each shape with that edge doubled, three tripled motifs, and the eight simple shapes themselves. The third comprises the $24$
distinct motifs selected
from a scan at $s\in\{2,3\}$ and sweep orders $r=2,\dots,5$: for each $s$ and $r$, the motifs on more than $r$ vertices
with the three smallest measured growth exponents of $\Exp B^2$ in $\eps$, or the five smallest for $r=5$. Six of these $24$
motifs also belong to the first two kinds, so the three kinds make up $95$ distinct motifs, $30$ of them on at most five
vertices and $65$ on six or seven.

\begin{table}[!htbp]
  \centering
  \caption{Taylor Coefficients of the Sweep of Order 5 on the Mixed Fibers}
  \label{tab:mechanism}
  \small
  \begin{tabular}{@{}ccccccc@{}}
\toprule
 & & \multicolumn{2}{c}{30 motifs on at most $5$ vertices} & \multicolumn{3}{c}{65 motifs on $6$ or $7$ vertices} \\
\cmidrule(lr){3-4}\cmidrule(l){5-7}
$s$ & macroscopic & identity & slope of $|\partial_u\varkappa_M|$ & largest & last & smallest slope \\
 & meanings & metric & ($5$ vertices) & $R(\varepsilon)$ & ratio & of $\max_z|B^{(5)}_M|$ \\
\midrule
$2$ & $3$ & $1.99\times10^{-19}$ & $[1.989,\,2.002]$ & $0.080$ & $[0.331,\,1.001]$ & $0.997$ \\
$3$ & $2$ & $6.76\times10^{-20}$ & $[2.951,\,2.999]$ & $0.178$ & $[0.111,\,1.010]$ & $0.971$ \\
$4$ & $1$ & $1.13\times10^{-19}$ & $[3.994,\,4.000]$ & $0.152$ & $[0.013,\,0.335]$ & $1.982$ \\
\bottomrule
\end{tabular}

  \par\smallskip
  \begin{minipage}{0.96\textwidth}\footnotesize
    Taylor coefficients $B^{(5)}_M(z)$ of the sweep of order $5$ in the edge weights, for motifs
      (edge multisets) $M$ and label assignments $z$ of their vertices, at the midpoint of the mixed fibers
      of the text ($K=5$, $s$ rare meanings), evaluated in double precision over all label assignments. $\varkappa_M$ is
      the joint cumulant of the
      co-membership indicators of the edges of $M$ and $M!$ the product of the factorials of the edge
      multiplicities. Identity metric: the largest value of
      $\max_z|B^{(5)}_M(z)-\partial_u\varkappa_M/M!|/\max\{1,\max_z|B^{(5)}_M(z)|\}$ over the $30$ motifs on at
      most five vertices and all $\varepsilon$. Slope: the range of the local slopes of
      $\log|\partial_u\varkappa_M|$ against $\log\varepsilon$ on the last two segments, over the $13$
      five-vertex motifs. $R(\varepsilon):=\max_z|B^{(5)}_M(z)|/\varepsilon$; largest $R(\varepsilon)$: over the
      $65$ motifs on six or seven vertices and all $\varepsilon$; last ratio: the range of
      $R(10^{-4})/R(3\times10^{-4})$ over these motifs, near $1$ for a coefficient of order $\varepsilon$ and
      near $1/3$ or below for one that decays at least like $\varepsilon^2$; smallest slope: the smallest
      local slope of $\log\max_z|B^{(5)}_M(z)|$ against $\log\varepsilon$ on the last two segments. Rounded
      to the digits shown.
  \end{minipage}
\end{table}

\Cref{tab:mechanism} summarizes
the results and \Cref{fig:mechanism} plots the largest normalized coefficient. On the motifs with at
most five vertices the coefficients agree with $\partial_u\varkappa_M/M!$ to the level of
floating-point rounding, and on the tested five-vertex motifs $|\partial_u\varkappa_M|$ scales like $\eps^s$. On six or
seven vertices the largest tested coefficient is consistent with order $\eps$ for $s=2,3$ and decays faster
for $s=4$. Apart from one motif at $s=4$, a recomputation in extended (long double) precision of the $22$
motif, class and $\eps$ combinations with the smallest coefficients agrees
with double precision on $\Exp B^2$ to within a relative difference of $2.4\times10^{-5}$. For that motif the
largest coefficient is below $3\times10^{-16}$, and the relative differences are $3.43\times10^{-4}$ and
$8.15\times10^{-3}$ at $\eps=3\times10^{-4}$ and $10^{-4}$, respectively. Substituting these extended-precision
values into \Cref{tab:mechanism} changes one entry: the lower end of the last-ratio range for $s=4$ is $0.012$, whereas
double precision gives $0.013$.

\begin{figure}[!htbp]
  \centering
  \includegraphics{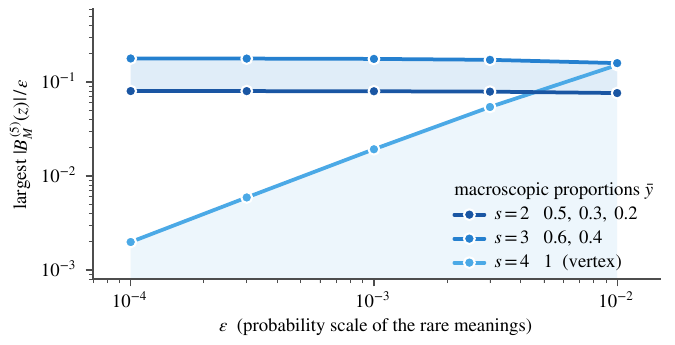}
  \caption{Largest Taylor coefficient $B^{(5)}_M(z)$ of the sweep of order $5$, over the $65$ tested
  motifs $M$ on six or seven vertices and all their label assignments $z$, divided by the probability
  scale $\varepsilon$ of the rare meanings, at the midpoint of a mixed fiber for $K=5$ with $s$ rare meanings and
  macroscopic proportions $(0.5,0.3,0.2)$, $(0.6,0.4)$ and $(1)$ for $s=2,3,4$. The coefficients are evaluated in double precision. The largest tested coefficient is
  consistent with order $\varepsilon$, the order of the uniform coefficient bound, when two or three
  meanings are macroscopic ($s=3,2$), and decays faster at the vertex ($s=4$).}
  \label{fig:mechanism}
\end{figure}

\subsection{The Fibers}
\label{sec:num-fibers}

\Cref{tab:fiber} checks \Cref{lem:fiber-one} in $80$-digit arithmetic. Along the fiber,
$P_2,\dots,P_{K-1}$ stay constant to within $2\times10^{-81}$, the roots are real, simple and positive, and
the increment $[P_K]^1_0:=P_K(\theta(1))-P_K(\theta(0))$ equals $KF_\eps(2\eps)$ to working precision. Let
$\beta_j:=(-1)^{K-1}\prod_ib_i/(b_j\prod_{i\ne j}(b_j-b_i))$. The reference limit of the last column is
\begin{equation}
  \label{eq:g:fiber-limit}
  \lim_{\eps\to0}\frac{[H]^1_0-\eps\log(1/\eps)}{\eps}=-2\log2+1-\sum_{j=1}^{K-1}\beta_j\log b_j .
\end{equation}
Indeed, in the proof of \Cref{lem:fiber-one}(iv) the small root contributes $\eps\log(1/\eps)-2\eps\log2$
exactly. For the large roots, $\dot\theta_j=\eps F_\eps'(\omega)/F_\eps'(\theta_j)$, and
$F_\eps'(\omega(u))/F_\eps'(\theta_j(u))\to F_0'(0)/F_0'(b_j)=\beta_j$ uniformly in $u\in[0,1]$, so
$(\theta_j(1)-\theta_j(0))/\eps\to\beta_j$ and $(h(\theta_j(1))-h(\theta_j(0)))/\eps\to-(1+\log b_j)\beta_j$ with $h(x)=-x\log x$.
Since $\sum_c\dot\theta_c=0$, we have $\eps+\sum_j\dot\theta_j=0$, hence $\sum_j\beta_j=-1$, and summing gives \eqref{eq:g:fiber-limit}.
The column $[P_K]^1_0/\eps$ tends to $KF_0'(0)=(-1)^{K-1}K\prod_jb_j$. In \Cref{tab:fiber} the last two columns approach these
limits as $\eps$ decreases, and the third stays bounded, as \Cref{lem:fiber-one}(i) states. At every tabulated $\eps$
the entropy gap $[H]^1_0$ exceeds $\frac12\eps\log(1/\eps)$.

\begin{table}[!htbp]
  \centering
  \caption{The Fiber for $s=1$ in 80-Digit Arithmetic}
  \label{tab:fiber}
  \small
  \begin{tabular}{@{}ccccc@{}}
\toprule
$K$ & $\varepsilon$ & $\max_{u,j}|\theta_j(u)-b_j|/\varepsilon$ & $[P_K]^1_0/\varepsilon$ & $([H]^1_0-\varepsilon\log(1/\varepsilon))/\varepsilon$ \\
\midrule
$4$ & $10^{-3}$ & $6.88$ & $-0.127749$ & $-3.05336$ \\
 & $10^{-4}$ & $7.14$ & $-0.130385$ & $-3.05845$ \\
 & $10^{-5}$ & $7.17$ & $-0.130650$ & $-3.05896$ \\
 & $10^{-6}$ & $7.17$ & $-0.130677$ & $-3.05901$ \\
 & limit & & $-0.13068$ & $-3.05902$ \\
\midrule
$5$ & $10^{-4}$ & $27.09$ & $+0.016107$ & $-3.69178$ \\
 & $10^{-5}$ & $26.74$ & $+0.016164$ & $-3.69301$ \\
 & $10^{-6}$ & $26.71$ & $+0.016169$ & $-3.69313$ \\
 & limit & & $+0.01617$ & $-3.69315$ \\
\bottomrule
\end{tabular}

  \par\smallskip
  \begin{minipage}{0.96\textwidth}\footnotesize
    The fiber for $s=1$: $\theta(u)$ has coordinates $\varepsilon(1+u)$ and the
      roots $\theta_1(u),\dots,\theta_{K-1}(u)$ of $F_\varepsilon(z)-F_\varepsilon(\varepsilon(1+u))$ other than
      $\varepsilon(1+u)$, where
      $F_\varepsilon(z)=(z-\varepsilon)\prod_{j=1}^{K-1}(z-(1-\varepsilon)b_j)$; computed in $80$-digit
      arithmetic for $K=4$ with $b=(0.45,0.33,0.22)$ and $K=5$ with $b=(0.35,0.28,0.22,0.15)$ in place of
      $\bar y(K-1,\tau)$, on $41$
      equally spaced values of $u\in[0,1]$. Only values of $\varepsilon$ for which the whole fiber is
      real-rooted are shown. For these $b$ the numerical thresholds, found by bisection on a $201$-point
      grid in $u$, are $\varepsilon<9.02\times10^{-3}$ for $K=4$ and $\varepsilon<7.13\times10^{-4}$ for $K=5$.
      The maximum in the third column runs over $u$ and $j=1,\dots,K-1$. In the last two columns,
      $[P_K]^1_0:=P_K(\theta(1))-P_K(\theta(0))=KF_\varepsilon(2\varepsilon)$ with $P_K(\theta)=\sum_c\theta_c^K$,
      and $[H]^1_0:=H(\theta(1))-H(\theta(0))$ is the change of the entropy. The limits are
      $(-1)^{K-1}K\prod_jb_j$ and the closed form \eqref{eq:g:fiber-limit}. Rounded to the digits shown.
  \end{minipage}
\end{table}

For \Cref{fig:modulus} the three least-favorable fibers are evaluated for $K=5$ and $\tau=0.05$ in $160$-digit arithmetic,
with the masses $\bar y(m,\tau)$ of the lemmas. For each $s$ the amplitudes form a grid of four values per decade,
from a value with $D_e\le10^{-32}$ up to $0.1$ for $s\ge2$, $1.8\times10^{-4}$ for $s=1$ and $3.2\times10^{-8}$ for $s=0$. At
the next values of the grid the last two fibers have complex roots. At both endpoints and at seven interior points of
every fiber on the grid the roots are real, simple and positive, and at least $K-s$ of them are no smaller than $\tau$. The power
sums $P_2,\dots,P_4$ agree at the two endpoints to within $4\times10^{-161}$. At the grid points with
$10^{-32}\le D_e\le10^{-8}$, which is the range of the figure, the ratio $|[H]|/\phi_s(D_e)$ varies by less than $3\%$ for $s=2,3,4$ and by a
factor $1.8$ for $s=1$, where it converges logarithmically as $D_e\to0$. For $s=0$, the ratio $|[H]|/\eta$ agrees with
$\Lam(\bar y(K,\tau))=168.86476$ to within a relative difference of $10^{-5}$.

\subsection{The Transfer}

\Cref{tab:transfer} checks \Cref{lem:transfer} on $K\in\{4,5\}$, $s\in\{0,1,2,K-1\}$ and
$\tau\in\{0.8,0.3\}/(K-s)$, in $50$-digit arithmetic with seed $20260928$. The point $\theta$ has $s$ small
coordinates, each drawn equal to $0$, to $\tau10^{-U}$ with $U$ uniform on $(0,15)$, or to $\tau U'$ with $U'$ uniform
on $(0.5,1)$, with probabilities $0.2$, $0.6$ and $0.2$, respectively. The small coordinates are rescaled jointly when their sum exceeds
$0.9(1-(K-s)\tau)$. The $K-s$ large coordinates are $\tau$ plus uniformly random shares of the remaining mass,
and when $K-s\ge2$ one of them equals $\tau$ with probability $1/3$. The point $\theta'$ comes from eight families: multiplicative and
additive perturbations of $\theta$ (4000 and 3000 pairs); moves of $e_K$ alone, small (4000) and large
(2000); independent Dirichlet draws with concentration in $\{0.05,0.2,1,5,50\}$ (4000); independent
draws of the construction of $\theta$ (2000); moves of one coordinate that is at least $\tau/2$, the smallest such coordinate or a uniformly chosen
one, to $0$ or below $\tau/2$, with the removed mass added to another coordinate (3000); and
Nelder--Mead searches for the $\theta'$ with $s+1$ coordinates below $\tau/2$ that minimizes $D_e$ (300). Each row
of \Cref{tab:transfer} has $22300$ pairs, or $22000$ where the search family is vacuous because $s=K-1$
and $K\tau/2\le1$. The two-sided bound is checked with $\tau'=\tau/2$ and with each pair's largest admissible
$\tau'$. There
are no violations in $356200$ pairs: none of the localization among the $150630$ pairs with
$D_e<c_0$, and none of the two-sided bound among the $305588$ pairs with both points in
$\Theta_{K,s}(\tau/2)$. The localization holds with a wide margin: on the tested pairs, every $\theta'$ outside
$\Theta_{K,s}(\tau/2)$ has $D_e\ge3.1\times10^4\,c_0$.

\begin{table}[!htbp]
  \centering
  \caption{Checks of the One-Sided Local Transfer}
  \label{tab:transfer}
  \scriptsize
  \setlength{\tabcolsep}{3.2pt}
\begin{tabular}{@{}rrrrrrrrrrr@{}}
\toprule
$K$ & $s$ & $\tau$ & $c_0$ & pairs & $D_e<c_0$ & viol. & one-sided & both in $\Theta_{K,s}(\tfrac\tau2)$ & two-sided & min $D_e/c_0$ \\
\midrule
$4$ & $0$ & $0.2$ & $1.00\times10^{-8}$ & $22300$ & $8863$ & $0$ & $2.4\times10^{-3}$ & $16648$ & $2.4\times10^{-3}$ & $5.5\times10^{5}$ \\
$4$ & $0$ & $0.075$ & $1.98\times10^{-10}$ & $22300$ & $7657$ & $0$ & $1.9\times10^{-4}$ & $17032$ & $2.0\times10^{-4}$ & $1.1\times10^{6}$ \\
$4$ & $1$ & $0.2667$ & $3.16\times10^{-8}$ & $22300$ & $9385$ & $0$ & $0.071$ & $18568$ & $0.096$ & $4.7\times10^{4}$ \\
$4$ & $1$ & $0.1$ & $6.25\times10^{-10}$ & $22300$ & $8311$ & $0$ & $0.016$ & $19041$ & $0.084$ & $1.6\times10^{5}$ \\
$4$ & $2$ & $0.4$ & $1.60\times10^{-7}$ & $22300$ & $10244$ & $0$ & $0.180$ & $19687$ & $0.180$ & $3.1\times10^{4}$ \\
$4$ & $2$ & $0.15$ & $3.16\times10^{-9}$ & $22300$ & $9712$ & $0$ & $0.093$ & $20572$ & $0.093$ & $1.2\times10^{6}$ \\
$4$ & $3$ & $0.8$ & $2.56\times10^{-6}$ & $22300$ & $11396$ & $0$ & $0.133$ & $20494$ & $0.334$ & $7.2\times10^{4}$ \\
$4$ & $3$ & $0.3$ & $5.06\times10^{-8}$ & $22000$ & $11030$ & $0$ & $0.096$ & $22000$ & $0.252$ & $\text{--}$ \\
$5$ & $0$ & $0.16$ & $1.32\times10^{-11}$ & $22300$ & $8556$ & $0$ & $6.3\times10^{-4}$ & $16506$ & $6.3\times10^{-4}$ & $1.2\times10^{8}$ \\
$5$ & $0$ & $0.06$ & $9.77\times10^{-14}$ & $22300$ & $7533$ & $0$ & $3.2\times10^{-5}$ & $17005$ & $3.2\times10^{-5}$ & $9.9\times10^{8}$ \\
$5$ & $1$ & $0.2$ & $4.02\times10^{-11}$ & $22300$ & $8996$ & $0$ & $0.031$ & $18153$ & $0.051$ & $4.8\times10^{6}$ \\
$5$ & $1$ & $0.075$ & $2.98\times10^{-13}$ & $22300$ & $7903$ & $0$ & $5.6\times10^{-3}$ & $18578$ & $0.060$ & $2.0\times10^{7}$ \\
$5$ & $2$ & $0.2667$ & $1.69\times10^{-10}$ & $22300$ & $9977$ & $0$ & $0.117$ & $19295$ & $0.117$ & $8.0\times10^{6}$ \\
$5$ & $2$ & $0.1$ & $1.26\times10^{-12}$ & $22300$ & $9101$ & $0$ & $0.053$ & $19890$ & $0.053$ & $9.7\times10^{7}$ \\
$5$ & $4$ & $0.8$ & $4.12\times10^{-8}$ & $22300$ & $11119$ & $0$ & $0.066$ & $20119$ & $0.306$ & $4.4\times10^{6}$ \\
$5$ & $4$ & $0.3$ & $3.05\times10^{-10}$ & $22000$ & $10847$ & $0$ & $0.053$ & $22000$ & $0.255$ & $\text{--}$ \\
\midrule
\multicolumn{4}{@{}l}{total} & $356200$ & $150630$ & $0$ & $0.180$ & $305588$ & $0.334$ & $3.1\times10^{4}$ \\
\bottomrule
\end{tabular}

  \par\smallskip
  \begin{minipage}{0.96\textwidth}\footnotesize
    Pairs $(\theta,\theta')$ with $\theta\in\Theta_{K,s}(\tau)$ and $\theta'\in\Theta_K$, in
      $50$-digit arithmetic; sampling and pair counts in the text, $D_e$ and $c_0=(\tau/(4(K+1)))^K$ as in
      \Cref{lem:transfer}. Columns: the number of pairs with $D_e<c_0$; violations of the
      localization and of the two-sided bound, counted over the pairs to which each
      applies; the largest ratio of $|H(\theta)-H(\theta')|$ to the two-sided bound with $\tau'=\tau/2$ over the
      pairs with $D_e<c_0$ (one-sided), and the number of pairs with both points in $\Theta_{K,s}(\tau/2)$ with
      the largest ratio over them (two-sided); the smallest $D_e/c_0$ among the $\theta'$ outside
      $\Theta_{K,s}(\tau/2)$ (a dash where none occurred). Rounded to the digits shown.
  \end{minipage}
\end{table}

\subsection{Variances of the Star and All-Tree Averages}
\label{sec:num-variance}

\Cref{tab:variance-exact} and \Cref{fig:variance}(a, b) give exact variances at $n\in\{K,K+1\}$, computed in
rational arithmetic by two routes that agree exactly in all $80$ settings: a decomposition over pairs of trees or
stars, and a direct enumeration of all verdict outcomes. As $\Delta$ decreases, the
ratio of the star to the all-tree variance approaches $K^{K-3}$ from below, and both normalized
variances approach $1$. At $\Delta=10^{-3}$ the ratio agrees with $K^{K-3}$ to within a relative difference of
$5.3\times10^{-6}$.

\begin{table}[!htbp]
  \centering
  \caption{Exact Variances of the All-Tree and Star Averages at Fixed $n$}
  \label{tab:variance-exact}
  \scriptsize
  \begin{tabular}{@{}cccccccccc@{}}
\toprule
 & & & \multicolumn{3}{c}{$\mathrm{Var}(\hat T^\star_K)/\mathrm{Var}(\hat T_K)$} & \multicolumn{3}{c}{$\nu_T$} & $\nu_S$ \\
\cmidrule(lr){4-6}\cmidrule(lr){7-9}\cmidrule(l){10-10}
$K$ & $n$ & $q$ & $\Delta=10^{-1}$ & $10^{-2}$ & $10^{-3}$ & $10^{-1}$ & $10^{-2}$ & $10^{-3}$ & $10^{-3}$ \\
\midrule
$4$ & $4$ & $0.3$ & $3.871281$ & $3.998734$ & $3.999987$ & $1.222173$ & $1.017574$ & $1.001719$ & $1.0017154$ \\
$4$ & $4$ & $0.5$ & $3.895025$ & $3.998932$ & $3.999989$ & $1.027125$ & $1.000267$ & $1.000003$ & $1.0000000$ \\
$4$ & $5$ & $0.3$ & $3.813419$ & $3.998108$ & $3.999981$ & $1.258022$ & $1.017878$ & $1.001722$ & $1.0017168$ \\
$4$ & $5$ & $0.5$ & $3.843904$ & $3.998404$ & $3.999984$ & $1.053407$ & $1.000519$ & $1.000005$ & $1.0000012$ \\
$5$ & $5$ & $0.3$ & $24.020839$ & $24.990460$ & $24.999905$ & $1.232056$ & $1.017651$ & $1.001719$ & $1.0017155$ \\
$5$ & $5$ & $0.5$ & $24.201749$ & $24.991943$ & $24.999919$ & $1.032984$ & $1.000322$ & $1.000003$ & $1.0000000$ \\
$5$ & $6$ & $0.3$ & $23.677714$ & $24.986789$ & $24.999868$ & $1.262692$ & $1.017907$ & $1.001722$ & $1.0017165$ \\
$5$ & $6$ & $0.5$ & $23.901186$ & $24.988848$ & $24.999888$ & $1.055180$ & $1.000534$ & $1.000005$ & $1.0000009$ \\
\bottomrule
\end{tabular}

  \par\smallskip
  \begin{minipage}{0.96\textwidth}\footnotesize
    Exact variances of the all-tree average $\hat T_K$ and the star average $\hat T^\star_K$ at
      fixed $n\in\{K,K+1\}$ as $\Delta$ decreases, in exact rational arithmetic, for
      $\theta=(0.4,0.3,0.2,0.1)$ ($K=4$) and $\theta=(0.3,0.25,0.2,0.15,0.1)$ ($K=5$), with the normalizations
      $\nu_T:=\mathrm{Var}(\hat T_K)\binom nKK^{K-2}\Delta^{2(K-1)}/(q(1-q))^{K-1}$ and
      $\nu_S:=\mathrm{Var}(\hat T^\star_K)K\binom nK\Delta^{2(K-1)}/(q(1-q))^{K-1}$, both of which tend to $1$ as
      $\Delta\to0$ at fixed $n$. Rounded to the digits shown.
  \end{minipage}
\end{table}

\Cref{tab:variance-mc} and \Cref{fig:variance}(c) report a Monte Carlo
experiment for $K=4$ and $\theta=(0.4,0.3,0.2,0.1)$ with $1000$ replicates per cell and an independent random
stream per cell, derived from seed $20260928$. In these cells $N=n\delta^3$ ranges from $0.025$ to $6.4$. The
standard errors of the ratio in \Cref{tab:variance-mc}, the error bars of \Cref{fig:variance}(c) and the ratio
$z$-scores use the delta method, and a bootstrap with $2000$ resamples (seed $20260929$) gives standard errors
within $5\%$ of these. The estimates agree with the exact finite-$n$ values: the $36$ variance $z$-scores have mean square
$1.01$ and the $18$ ratio $z$-scores mean square $1.31$. As $\delta$ decreases the exact ratio rises toward $4$, with
corrections of order $\delta$: at $\delta=0.05$ it is between $3.71$ and $3.76$.

\begin{table}[!htbp]
  \centering
  \caption{Monte Carlo Variances of the All-Tree and Star Averages for $K=4$}
  \label{tab:variance-mc}
  \scriptsize
  \begin{tabular}{@{}ccccccccc@{}}
\toprule
 & & & \multicolumn{2}{c}{$\mathrm{Var}(\hat T^\star_4)/\mathrm{Var}(\hat T_4)$} & \multicolumn{2}{c}{$\tilde\nu_T$} & \multicolumn{2}{c}{$\tilde\nu_S$} \\
\cmidrule(lr){4-5}\cmidrule(lr){6-7}\cmidrule(l){8-9}
$q$ & $n$ & $\delta$ & Monte Carlo & exact & Monte Carlo & exact & Monte Carlo & exact \\
\midrule
$0.3$ & $200$ & $0.2$ & $3.227\,(0.164)$ & $3.069$ & $1.737\,(0.082)$ & $1.854$ & $1.359\,(0.067)$ & $1.380$ \\
$0.3$ & $200$ & $0.1$ & $3.132\,(0.175)$ & $3.467$ & $1.435\,(0.076)$ & $1.416$ & $1.090\,(0.054)$ & $1.191$ \\
$0.3$ & $200$ & $0.05$ & $3.258\,(0.184)$ & $3.713$ & $1.416\,(0.071)$ & $1.222$ & $1.119\,(0.055)$ & $1.101$ \\
$0.3$ & $400$ & $0.2$ & $2.769\,(0.140)$ & $3.064$ & $1.972\,(0.092)$ & $1.807$ & $1.345\,(0.066)$ & $1.363$ \\
$0.3$ & $400$ & $0.1$ & $3.410\,(0.182)$ & $3.464$ & $1.496\,(0.070)$ & $1.383$ & $1.257\,(0.061)$ & $1.180$ \\
$0.3$ & $400$ & $0.05$ & $3.673\,(0.195)$ & $3.711$ & $1.227\,(0.059)$ & $1.196$ & $1.110\,(0.052)$ & $1.093$ \\
$0.3$ & $800$ & $0.2$ & $3.100\,(0.149)$ & $3.061$ & $1.768\,(0.083)$ & $1.778$ & $1.360\,(0.062)$ & $1.351$ \\
$0.3$ & $800$ & $0.1$ & $3.466\,(0.174)$ & $3.462$ & $1.355\,(0.057)$ & $1.364$ & $1.166\,(0.050)$ & $1.172$ \\
$0.3$ & $800$ & $0.05$ & $3.885\,(0.218)$ & $3.710$ & $1.226\,(0.057)$ & $1.181$ & $1.181\,(0.055)$ & $1.087$ \\
$0.5$ & $200$ & $0.2$ & $3.059\,(0.159)$ & $3.176$ & $1.710\,(0.087)$ & $1.632$ & $1.269\,(0.062)$ & $1.257$ \\
$0.5$ & $200$ & $0.1$ & $3.377\,(0.191)$ & $3.538$ & $1.341\,(0.071)$ & $1.308$ & $1.099\,(0.056)$ & $1.123$ \\
$0.5$ & $200$ & $0.05$ & $3.544\,(0.196)$ & $3.755$ & $1.208\,(0.059)$ & $1.164$ & $1.038\,(0.047)$ & $1.060$ \\
$0.5$ & $400$ & $0.2$ & $2.995\,(0.169)$ & $3.174$ & $1.594\,(0.086)$ & $1.612$ & $1.175\,(0.055)$ & $1.260$ \\
$0.5$ & $400$ & $0.1$ & $3.648\,(0.203)$ & $3.537$ & $1.263\,(0.060)$ & $1.291$ & $1.134\,(0.055)$ & $1.124$ \\
$0.5$ & $400$ & $0.05$ & $4.012\,(0.215)$ & $3.754$ & $1.067\,(0.050)$ & $1.148$ & $1.054\,(0.050)$ & $1.061$ \\
$0.5$ & $800$ & $0.2$ & $3.057\,(0.160)$ & $3.173$ & $1.634\,(0.078)$ & $1.602$ & $1.240\,(0.057)$ & $1.262$ \\
$0.5$ & $800$ & $0.1$ & $3.562\,(0.182)$ & $3.536$ & $1.277\,(0.056)$ & $1.282$ & $1.129\,(0.053)$ & $1.125$ \\
$0.5$ & $800$ & $0.05$ & $4.031\,(0.212)$ & $3.753$ & $1.159\,(0.053)$ & $1.139$ & $1.159\,(0.052)$ & $1.061$ \\
\bottomrule
\end{tabular}

  \par\smallskip
  \begin{minipage}{0.96\textwidth}\footnotesize
    Monte Carlo variances of the all-tree average $\hat T_4$ and the star average $\hat T^\star_4$
      for $K=4$ and $\theta=(0.4,0.3,0.2,0.1)$, at $\delta=n\Delta^2$, $1000$ replicates per row, with standard
      errors in parentheses (from the empirical fourth central moment for the variances and by the delta
      method for the ratio), against the exact finite-$n$ values. Normalizations:
      $\tilde\nu_T:=N\mathrm{Var}(\hat T_4)/[(4!/4^2)(q(1-q))^3]$ with $N=n^4\Delta^6$, whose limit as $\Delta\to0$
      at fixed $n$ is $n^4/(24\binom n4)$, about $1.0306$ at $n=200$ and $1.0075$ at $n=800$, not $1$; and
      $\tilde\nu_S:=4\binom n4\Delta^6\mathrm{Var}(\hat T^\star_4)/(q(1-q))^3$, whose limit as $\Delta\to0$ at
      fixed $n$ is $1$. Rounded to the digits shown.
  \end{minipage}
\end{table}

\begin{figure}[!htbp]
  \centering
  \includegraphics[width=\linewidth]{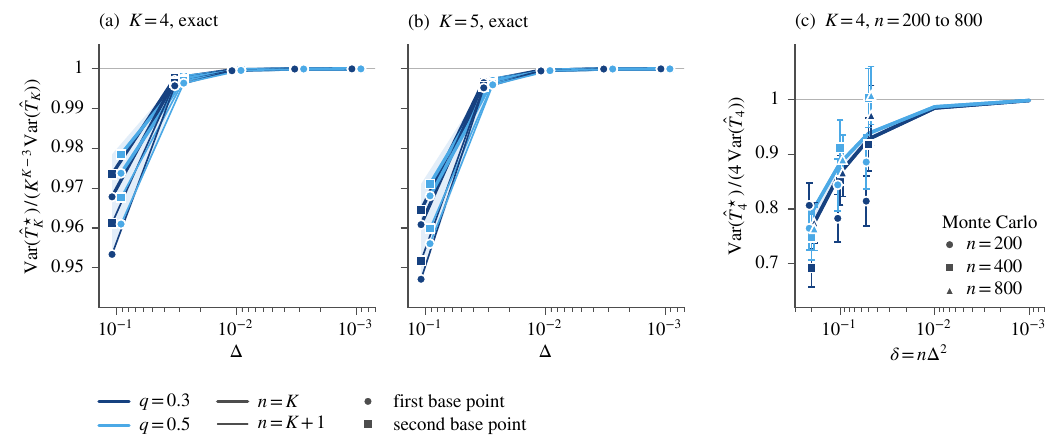}
  \caption{Star versus all-tree averages of the centered verdicts, with $q=0.3$ in dark blue and $q=0.5$ in
  light blue. The gray line marks the limit $1$, and the shading spans the exact values drawn in each panel.
  \textbf{(a, b)} The exact ratio $\mathrm{Var}(\hat T^\star_K)/(K^{K-3}\mathrm{Var}(\hat T_K))$
  for $K=4$ and $K=5$ at fixed $n=K$ (thick lines) and $n=K+1$ (thin lines) as $\Delta$ decreases, in exact rational
  arithmetic, at two base points per $K$: circles $\theta=(0.4,0.3,0.2,0.1)$ and
  $(0.3,0.25,0.2,0.15,0.1)$, squares $\theta=(0.35,0.3,0.2,0.15)$ and $(0.28,0.24,0.2,0.16,0.12)$, with the $q=0.3$
  series offset horizontally to the left of the $q=0.5$ series for legibility.
  \textbf{(c)} $K=4$, $\theta=(0.4,0.3,0.2,0.1)$: the same ratio against $\delta=n\Delta^2$, exact finite-$n$
  values at $n=200$ and $n=800$ (lines, which coincide at this scale) and Monte Carlo estimates with $1000$
  replicates each and delta-method standard errors (circles $n=200$, squares $n=400$, triangles $n=800$, offset
  horizontally in that order for legibility). All panels check the variance constant behind
  the loss factor $K^{K-3}$.}
  \label{fig:variance}
\end{figure}

\subsection{Fisher Information Along the Interior Fiber}

\Cref{tab:fisher} reports the Fisher information of the verdict graph at $u=0$ on the interior fiber, for $n\in\{K,K+1\}$, four base points and $q\in\{0.3,0.5\}$, computed by summing over all
$2^{\binom n2}$ verdict graphs without any expansion in $\Delta$. The ratio $I(0)/L_0$ approaches $1$ as $\Delta$ decreases, and
the deviation is proportional to $\Delta$ for $q=0.3$ and of smaller order for $q=0.5$.

\begin{table}[!htbp]
  \centering
  \caption{Fisher Information Along the Interior Fiber}
  \label{tab:fisher}
  \scriptsize
  \setlength{\tabcolsep}{4pt}
\begin{tabular}{@{}cccccccc@{}}
\toprule
 & & & & \multicolumn{3}{c}{$I(0)/L_0$} & \\
\cmidrule(lr){5-7}
$K$ & $\theta_0$ & $n$ & $q$ & $\Delta=10^{-2}$ & $10^{-3}$ & $10^{-4}$ & $(I(0)/L_0-1)/\Delta$ at $\Delta=10^{-4}$ \\
\midrule
$4$ & $(0.4,0.3,0.2,0.1)$ & $4$ & $0.3$ & $0.9831455$ & $0.9982887$ & $0.9998286$ & $-1.7140$ \\
$4$ & $(0.4,0.3,0.2,0.1)$ & $4$ & $0.5$ & $1.0000900$ & $1.0000009$ & $1.0000000$ & $+0.0001$ \\
$4$ & $(0.4,0.3,0.2,0.1)$ & $5$ & $0.3$ & $0.9833249$ & $0.9982905$ & $0.9998286$ & $-1.7138$ \\
$4$ & $(0.4,0.3,0.2,0.1)$ & $5$ & $0.5$ & $1.0002446$ & $1.0000024$ & $1.0000000$ & $+0.0002$ \\
$4$ & $(0.35,0.3,0.2,0.15)$ & $4$ & $0.3$ & $0.9845438$ & $0.9984312$ & $0.9998429$ & $-1.5712$ \\
$4$ & $(0.35,0.3,0.2,0.15)$ & $4$ & $0.5$ & $1.0000872$ & $1.0000009$ & $1.0000000$ & $+0.0001$ \\
$4$ & $(0.35,0.3,0.2,0.15)$ & $5$ & $0.3$ & $0.9847728$ & $0.9984335$ & $0.9998429$ & $-1.5709$ \\
$4$ & $(0.35,0.3,0.2,0.15)$ & $5$ & $0.5$ & $1.0002839$ & $1.0000028$ & $1.0000000$ & $+0.0003$ \\
$5$ & $(0.3,0.25,0.2,0.15,0.1)$ & $5$ & $0.3$ & $0.9831627$ & $0.9982888$ & $0.9998286$ & $-1.7140$ \\
$5$ & $(0.3,0.25,0.2,0.15,0.1)$ & $5$ & $0.5$ & $1.0001132$ & $1.0000011$ & $1.0000000$ & $+0.0001$ \\
$5$ & $(0.3,0.25,0.2,0.15,0.1)$ & $6$ & $0.3$ & $0.9834181$ & $0.9982914$ & $0.9998286$ & $-1.7137$ \\
$5$ & $(0.3,0.25,0.2,0.15,0.1)$ & $6$ & $0.5$ & $1.0003327$ & $1.0000033$ & $1.0000000$ & $+0.0003$ \\
$5$ & $(0.28,0.24,0.2,0.16,0.12)$ & $5$ & $0.3$ & $0.9838487$ & $0.9983574$ & $0.9998355$ & $-1.6454$ \\
$5$ & $(0.28,0.24,0.2,0.16,0.12)$ & $5$ & $0.5$ & $1.0001238$ & $1.0000012$ & $1.0000000$ & $+0.0001$ \\
$5$ & $(0.28,0.24,0.2,0.16,0.12)$ & $6$ & $0.3$ & $0.9841283$ & $0.9983603$ & $0.9998355$ & $-1.6451$ \\
$5$ & $(0.28,0.24,0.2,0.16,0.12)$ & $6$ & $0.5$ & $1.0003638$ & $1.0000036$ & $1.0000000$ & $+0.0004$ \\
\bottomrule
\end{tabular}

  \par\smallskip
  \begin{minipage}{0.96\textwidth}\footnotesize
    Fisher information $I(u)$ of the verdict graph with respect to the coordinate
      $u=e_K(\theta)-e_K(\theta_0)$, $e_K(\theta)=\prod_c\theta_c$, along the interior fiber, at $u=0$, computed by summing over all verdict graphs in extended (long double)
      precision at fixed $n\in\{K,K+1\}$, relative to $L_0=\binom nKK^K\Delta^{2K-2}/(q(1-q))^{K-1}$, and the
      first-order coefficient $(I(0)/L_0-1)/\Delta$ at $\Delta=10^{-4}$. Rounded to the digits shown.
  \end{minipage}
\end{table}

\section{Conclusion}
\label{sec:conclusion}

Observed through a weak pairwise judge with a known channel, the semantic entropy is as hard to estimate as its rare
meanings make it. With at most $s$ meanings below a fixed probability, the minimax squared risk is of order $\psi_s(N)$ on
every fixed polynomial window, where $N=n^K\Delta^{2K-2}$ is the scale of the last power sum. For every $K\ge3$ the
entropy is uniformly consistently estimable exactly when $N\to\infty$. One all-tree moment estimator attains the rate of every class
without knowing $s$ or $\tau$. At interior points with pairwise distinct probabilities it attains the sharp local constant,
which an average over stars multiplies by $K^{K-3}$.

Adaptation rests on a one-sided local transfer, which bounds the
entropy error of the projected estimate even when the estimate leaves the class. The exponent $1/s$ of the transfer
cannot be improved for $s\ge2$. The matching lower bounds rest on an information inequality for
finite-order sweeps whose coefficients are controlled uniformly by Cauchy estimates. The inequality needs no single
macroscopic meaning and so reaches the mixed faces of the simplex. In practice, these results show how many answers a weak judge
needs: their number must grow faster than $\Delta^{-(2K-2)/K}$, and, for a given accuracy, far fewer suffice when few
meanings can be rare.

The polynomial windows cover $\Delta=n^{-\gamma}$ for every fixed $\gamma>1/2$. Whether the matching lower bounds hold on the
whole window $n\Delta^2\le c_1$, where $n\Delta^2$ may decay more slowly than every power of $n$ or stay bounded away from zero,
is open in general. Natural extensions let the number of meanings grow with $n$, calibrate the channel
from data, and allow heterogeneous or dependent verdicts.

\section*{Acknowledgment}
This work was supported by the Young Doctoral Talent Incubation Program of the Second Affiliated Hospital, Army Medical
University, under Grant 2025YQB034.

\bibliographystyle{IEEEtran}
\bibliography{references}

\clearpage
\appendix
\phantomsection\addcontentsline{toc}{section}{Appendix}%
\noindent{\Large\bfseries Appendix}\par
\section{Proofs of the General Tools}
\label{app:tools}

This appendix proves the general tools behind the lower bounds. They extend to $K$ meanings the tools of the concurrent
analyses of at most three \cite{anonymous2026three} and at most four meanings
\cite{anonymous2026four}: the bounds on the edge weights, the Glauber identity
and the score representation, the Hellinger path bound and the two-point bound, the sweeps, the
locality of their coefficients, the exact identity for the full sweep, the local decomposition of
the sweep with its derivative bounds, and the integral representation of the entropy. In this appendix the sweeps
and the local decomposition cover every order, whereas the three-meaning analysis uses a single sweep, and the
four-meaning analysis writes the local decomposition and its derivative bounds for the sweep of order three. We
state and prove them here so that the paper is self-contained.
Throughout this appendix $q\in[q_0,1-q_0]$ with $q_0\in(0,1/2]$ and $0<\Delta\le q_0/2$, so that $p\le1-q_0/2$.

\subsection{Edge Weights, Fields and the Score}

Put $\lambda:=\log(p/q)>0$, $\nu:=\log\frac{1-p}{1-q}<0$ and $J_{ij}:=\nu+(\lambda-\nu)A_{ij}$ for $i\ne j$,
the log-likelihood ratio of a verdict. Since $\Bern(a;p)/\Bern(a;q)=\exp(\nu+(\lambda-\nu)a)$ for
$a\in\{0,1\}$,
\begin{equation}
  \label{eq:a:likelihood}
  \Prob(A\mid Z)=\prod_{i<j}\Bern(A_{ij};q)\cdot\exp\Big(\sum_{i<j}J_{ij}W_{ij}\Big).
\end{equation}
The field of meaning $c$ at vertex $i$ is $F_{ic}:=\sum_{j\ne i}J_{ij}\Ind{Z_j=c}$, so
$F_i=\sum_{j\ne i}J_{ij}\mathbf e_{Z_j}\in\Real^K$, and $\pi_c(F):=\theta_ce^{F_c}/\sum_b\theta_be^{F_b}$ is the
tilted law.

\begin{lemma}[Edge Weights]
  \label{lem:weights}
  \begin{enumerate}[label=(\roman*),leftmargin=*]
    \item $0<\lambda\le\Delta/q_0$, $0<-\nu\le2\Delta/q_0$ and $|J_{ij}|\le2\Delta/q_0$.
    \item Given $Z$, the weights $J_{ij}$, $i<j$, are conditionally independent,
      $|\Exp[J_{ij}\mid Z]|\le2\Delta^2/q_0^2$ and $\Var(J_{ij}\mid Z)\le3\Delta^2/q_0^2$.
    \item For every integer $a\ge1$,
      \begin{equation}
        \label{eq:a:moments}
        \big|\Exp[J_{ij}^a\mid Z]\big|\le\Big(\frac{2\Delta}{q_0}\Big)^{c(a)},
        \qquad c(1):=2,\quad c(a):=a\ \ (a\ge2).
      \end{equation}
  \end{enumerate}
\end{lemma}

\begin{proof}
  (i) $\lambda=\log(1+\Delta/q)\le\Delta/q\le\Delta/q_0$. With $x:=\Delta/(1-q)\in(0,1)$ and
  $\log(1-x)\ge-x/(1-x)$ we get $\nu\ge-\Delta/(1-p)\ge-2\Delta/q_0$. Hence
  $|J_{ij}|\le\max\{\lambda,-\nu\}\le2\Delta/q_0$.

  (ii) Each $J_{ij}$ is a function of $A_{ij}$ alone, whence the conditional independence. Given $Z$,
  $\Exp[J_{ij}\mid Z]=\nu+(\lambda-\nu)r$ with $r\in\{q,p\}$; for $r=p$ this is the Kullback--Leibler
  divergence $D(p\,\|\,q)$ and for $r=q$ it is $-D(q\,\|\,p)$. By $D(x\,\|\,y)\le(x-y)^2/(y(1-y))$,
  $q(1-q)\ge q_0/2$ and $p(1-p)\ge q_0\cdot q_0/2$, both have modulus at most $2\Delta^2/q_0^2$. Finally,
  $\Var(J_{ij}\mid Z)=(\lambda-\nu)^2r(1-r)\le\frac14(3\Delta/q_0)^2\le3\Delta^2/q_0^2$.

  (iii) For $a\ge2$ use (i); for $a=1$ use (ii) and $2\Delta^2/q_0^2\le(2\Delta/q_0)^2$.
\end{proof}

\begin{lemma}[Glauber Identity]
  \label{lem:glauber}
  Let $\theta_c>0$ for all $c$ and $U\subseteq[n]$. Given $(A,Z_{-U})$, the conditional law of $Z_U$ is
  proportional to
  \begin{equation}
    \label{eq:a:blockweight}
    z'_U\ \mapsto\ \prod_{k\in U}\theta_{z'_k}\cdot
    \exp\Big(\sum_{e=\{k,l\}:\,e\cap U\ne\varnothing}J_e\Ind{z_k=z_l}\Big),
    \qquad z:=(z'_U,Z_{-U}).
  \end{equation}
  In particular $\Prob(Z_i=c\mid A,Z_{-i})=\pi_c(F_i)$.
\end{lemma}

\begin{proof}
  By \eqref{eq:a:likelihood}, $\Prob(Z_U=z'_U,A\mid Z_{-U})$ equals $\prod_{k\in U}\theta_{z'_k}$ times
  $\prod_{i<j}\Bern(A_{ij};q)\exp(\sum_eJ_eW_e(z))$, and the factors of the edges that do not meet $U$
  and the Bernoulli factors do not depend on $z'_U$. Normalizing gives the claim; for $U=\{i\}$ the
  weight of $c$ is $\theta_ce^{F_{ic}}$.
\end{proof}

\begin{lemma}[Score Representation]
  \label{lem:score}
  Let $u\mapsto\theta(u)\in\Simplex$ be continuously differentiable on an interval with all
  $\theta_c(u)>0$. Write $P_u$ for the law of $A$ under $\theta(u)$, $\Sobs(A):=\partial_u\log P_u(A)$,
  $\ell_c:=\dot\theta_c/\theta_c$ and
  \[
    \Phi(F):=\frac{\sum_c\dot\theta_ce^{F_c}}{\sum_c\theta_ce^{F_c}}=\sum_c\pi_c(F)\,\ell_c,
    \qquad \SG:=\sum_{i=1}^n\Phi(F_i).
  \]
  Then $\Sobs=\Exp_u[\sum_i\ell_{Z_i}\mid A]=\Exp_u[\SG\mid A]$, $\Exp_u\Sobs=0$, $\Phi(0)=0$, and
  $|\Phi(F)|\le\max_c|\ell_c|$ for every real $F$.
\end{lemma}

\begin{proof}
  $P_u(A)=\sum_z\prod_i\theta_{z_i}(u)\Prob(A\mid Z=z)$ is a polynomial in $\theta(u)$; differentiating term
  by term gives $\partial_uP_u(A)=\sum_z(\sum_i\ell_{z_i})\prod_i\theta_{z_i}\Prob(A\mid z)$, that is,
  $\Sobs=\Exp_u[\sum_i\ell_{Z_i}\mid A]$. By \Cref{lem:glauber} and the tower property,
  $\Exp_u[\ell_{Z_i}\mid A]=\Exp_u[\sum_c\pi_c(F_i)\ell_c\mid A]=\Exp_u[\Phi(F_i)\mid A]$. Next,
  $\Exp_u\Sobs=\sum_A\partial_uP_u(A)=0$, and $\Phi(0)=\sum_c\dot\theta_c=0$ because $\theta(u)\in\Simplex$.
  For real $F$, $\Phi(F)$ is an average of the $\ell_c$ under the probability vector $\pi(F)$.
\end{proof}

\begin{lemma}[Derivatives of $\Phi$]
  \label{lem:phi}
  For $F\in\Cpx^K$ with $\sum_b\theta_be^{F_b}\ne0$ the formulas of \Cref{lem:score} define complex
  weights $\pi_c(F)$ and $\Phi(F)=\sum_c\pi_c(F)\ell_c$. For every $m\ge1$ and all meanings $c_1,\dots,c_m$, the
  derivative $\partial_{c_1}\cdots\partial_{c_m}\Phi(F)$ is the joint cumulant, under $\pi(F)$, of the variables
  $\ell,\mathbf 1_{c_1},\dots,\mathbf 1_{c_m}$ on $[K]$, where $\mathbf 1_c$ is the indicator of $c$. If
  $F\in\Real^K$ and $g:=\max_c|F_c|$, then $|\Phi(F)|\le\|\dot\theta\|_1e^{2g}$ and
  $|\partial_{c_1}\cdots\partial_{c_m}\Phi(F)|\le C_m\|\dot\theta\|_1e^{2g}$ with
  $C_m:=\mathrm{Bell}(m+1)\,m!$.
\end{lemma}

\begin{proof}
  $\mathcal K(F):=\log\sum_b\theta_be^{F_b}$ is the cumulant generating function of the family $\pi(F)$ of
  laws on $[K]$, so $\partial_{c_0}\cdots\partial_{c_m}\mathcal K$ is the joint cumulant of
  $\mathbf 1_{c_0},\dots,\mathbf 1_{c_m}$ under $\pi(F)$; since $\Phi=\sum_{c_0}\ell_{c_0}\partial_{c_0}\mathcal K$,
  multilinearity of cumulants gives the formula. In the expansion of this cumulant over the
  $\mathrm{Bell}(m+1)$ partitions of $\{0,1,\dots,m\}$, with coefficients of modulus at most $m!$, the
  variable $\ell$ enters only through the block containing $0$, whose moment is $\Phi(F)$, the product
  $\pi_c(F)\ell_c=\dot\theta_ce^{F_c}/\sum_b\theta_be^{F_b}$ for a single meaning $c$, or zero. For real $F$,
  $\sum_b\theta_be^{F_b}\ge e^{-g}$ and $|\sum_c\dot\theta_ce^{F_c}|\le\|\dot\theta\|_1e^g$, so each of these
  has modulus at most $\|\dot\theta\|_1e^{2g}$; the other blocks are probabilities.
\end{proof}

\subsection{Two-Point Bounds}

\begin{lemma}[Hellinger Path Bound]
  \label{lem:hellinger}
  Let $\{P_u\}_{u\in[u_0,u_1]}$, $u_1-u_0\le1$, be distributions on a finite set with $P_u(A)>0$ and
  $u\mapsto P_u(A)$ continuously differentiable, and let $I(u):=\sum_AP_u(A)(\partial_u\log P_u(A))^2$.
  Then $\TV(P_{u_0},P_{u_1})\le\frac12(\int_{u_0}^{u_1}I(u)\,\mathrm du)^{1/2}$.
\end{lemma}

\begin{proof}
  For every $A$,
  \[
    \sqrt{P_{u_1}(A)}-\sqrt{P_{u_0}(A)}=\int_{u_0}^{u_1}\frac12\sqrt{P_u(A)}\,\partial_u\log P_u(A)\,\mathrm du,
  \]
  whose square is at most $\frac14\int_{u_0}^{u_1}P_u(A)(\partial_u\log P_u(A))^2\mathrm du$ by
  Cauchy--Schwarz on an interval of length at most one. Summing over $A$,
  $\sum_A(\sqrt{P_{u_1}}-\sqrt{P_{u_0}})^2\le\frac14\int I$. Finally,
  $\sum_A|P-Q|\le(\sum_A(\sqrt P-\sqrt Q)^2)^{1/2}(\sum_A(\sqrt P+\sqrt Q)^2)^{1/2}\le2(\sum_A(\sqrt P-\sqrt Q)^2)^{1/2}$,
  so $\TV=\frac12\sum_A|P-Q|\le\frac12(\int I)^{1/2}$.
\end{proof}

\begin{lemma}[Two Points]
  \label{lem:lecam}
  Let $\Theta\subseteq\Simplex$, $\omega>0$ and $\theta^{(0)},\theta^{(1)}\in\Theta$ with
  $|\Ent(\theta^{(0)})-\Ent(\theta^{(1)})|\ge\omega$ and $\TV(P_{\theta^{(0)}},P_{\theta^{(1)}})\le1/2$. Then
  every estimator $\hH=\hH(A,\Xi)$, where $\Xi$ is auxiliary randomness independent of $(Z,A)$, satisfies
  $\max_{k\in\{0,1\}}\Exp_{\theta^{(k)}}(\hH-\Ent(\theta^{(k)}))^2\ge\omega^2/16$; hence
  $\inf_{\hH}\sup_{\theta\in\Theta}\Exp_\theta(\hH-\Ent(\theta))^2\ge\omega^2/16$.
\end{lemma}

\begin{proof}
  The pair $(A,\Xi)$ has laws $P_{\theta^{(k)}}\otimes\mathcal L(\Xi)$, whose total variation distance equals
  $\TV(P_{\theta^{(0)}},P_{\theta^{(1)}})$. The events $E_k:=\{|\hH-\Ent(\theta^{(k)})|<\omega/2\}$ are
  disjoint, so $\Prob_0(E_0^c)+\Prob_1(E_1^c)\ge\Prob_0(E_0^c)+\Prob_1(E_0)\ge1-\TV\ge1/2$. By Markov's
  inequality $\sum_k\Exp_k(\hH-\Ent(\theta^{(k)}))^2\ge(\omega/2)^2\cdot\frac12$, and the maximum is at
  least half the sum.
\end{proof}

\subsection{Sweeps and Motifs}
\label{app:sweeps}

Fix a path as in \Cref{lem:score} and a value $u$. For $U\subseteq[n]$ let $\Pu_Uf:=\Exp_u[f\mid A,Z_{-U}]$,
with $\Pu_\varnothing:=\mathrm{Id}$, and for $S\subseteq[n]$ let $\DS_S:=\sum_{U\subseteq S}(-1)^{|S|-|U|}\Pu_U$. The
sweep of order $r\ge0$ is $\Sw r:=\sum_{S\subseteq[n],\,|S|\le r}\DS_S\SG$.

\begin{lemma}[Sweeps]
  \label{lem:sweep}
  \begin{enumerate}[label=(\roman*),leftmargin=*]
    \item For every $f(A,Z)$ and every nonempty $S$, $\Exp_u[\DS_Sf\mid A]=0$. Consequently
      $\Exp_u[\Sw r\mid A]=\Sobs$, $\Exp_u\Sw r=0$ and $I(u)\le\Exp_u\Sw r^2$ for every $r\ge0$.
    \item $\sum_{S\subseteq[n]}\DS_S=\Pu_{[n]}$; in particular $\sum_{S\subseteq[n]}\DS_S\SG=\Sobs$.
  \end{enumerate}
\end{lemma}

\begin{proof}
  (i) By the tower property $\Exp_u[\Pu_Uf\mid A]=\Exp_u[f\mid A]$ for every $U$, so
  $\Exp_u[\DS_Sf\mid A]=\Exp_u[f\mid A]\sum_{U\subseteq S}(-1)^{|S|-|U|}=0$ for $S\ne\varnothing$. With
  \Cref{lem:score}, $\Exp_u[\Sw r\mid A]=\Exp_u[\SG\mid A]=\Sobs$, and the conditional Jensen inequality gives
  $\Sobs^2\le\Exp_u[\Sw r^2\mid A]$, whose expectation is $I(u)=\Exp_u\Sobs^2$.
  (ii) $\sum_S\sum_{U\subseteq S}(-1)^{|S\setminus U|}\Pu_U=\sum_U\Pu_U\sum_{S\supseteq U}(-1)^{|S\setminus U|}=\Pu_{[n]}$,
  since the inner sum vanishes unless $U=[n]$; and $\Pu_{[n]}\SG=\Exp_u[\SG\mid A]=\Sobs$.
\end{proof}

\subsubsection{The Sweep as an Analytic Function of the Edge Weights}
Fix a labeling $z\in[K]^n$. The operators $\Pu_U$ depend on $A$ only through $J=(J_{kl})_{k<l}$. For a
vector $J$ let
\[
  w(z):=\prod_v\theta_{z_v}\prod_{k<l}e^{J_{kl}W_{kl}(z)},\qquad
  \Pu_Uf(z):=\frac{\sum_{z'_U}w(z'_U,z_{-U})\,f(z'_U,z_{-U})}{\sum_{z'_U}w(z'_U,z_{-U})} ,
\]
which at $J=J(A)$ is the conditional expectation by \Cref{lem:glauber}. With $\SG$ written as the
function $\sum_i\Phi(\sum_{j\ne i}J_{ij}\mathbf e_{z_j})$ of $(J,z)$, the sweep becomes $\Sw r=\mathcal S_r(J;z)$ for an
explicit function $\mathcal S_r$, a finite combination of ratios of exponential polynomials in $J$
whose denominators equal $1$ at $J=0$; so $\mathcal S_r$ is analytic on a complex neighborhood of
$J=0$. A \emph{motif} is an edge multiset $M$ on $[n]$ with multiplicities $m_e$; we write
$J^M:=\prod_eJ_e^{m_e}$, $M!:=\prod_em_e!$, $|M|:=\sum_em_e$, $V(M)$ for the set of vertices it touches,
$v(M):=|V(M)|$ and $d(M)$ for the number of distinct edges, and
\[
  B^{(r)}_M(z):=\coef M\mathcal S_r:=\frac1{M!}\,\partial^M\mathcal S_r(J;z)\Big|_{J=0}.
\]
The same notation $\coef M$ is used for the coefficient of $J^M$ in any function of $J$ analytic near
$0$. A motif is connected if it is connected as a multigraph on $V(M)$.

\subsubsection{Motif Systems}
For a finite vertex set $V$, the motif system on $V$ is the same construction with the vertices of
$V$ only: the weight $w_V(z):=\prod_{v\in V}\theta_{z_v}\prod_{\{k,l\}\subseteq V}e^{J_{kl}W_{kl}}$, the fields,
$\Phi$, the operators $\Pu^{(V)}_U$ for $U\subseteq V$, their combinations $\DS^{(V)}_S$, and
$\SG^{(V)}:=\sum_{i\in V}\Phi(F^{(V)}_i)$ all use only the edges inside $V$.

\begin{lemma}[Locality]
  \label{lem:local}
  Let $M$ be a motif, $V:=V(M)$, $z\in[K]^n$ and $S\subseteq[n]$.
  \begin{enumerate}[label=(\roman*),leftmargin=*]
    \item If $S\not\subseteq V$, then $\coef M\DS_S\SG=0$.
    \item If $S\subseteq V$, then $\coef M\DS_S\SG(z)=\coef M\DS^{(V)}_S\SG^{(V)}(z_V)$.
    \item Hence $B^{(r)}_M(z)$ depends only on $z_V$, and
      $B^{(r)}_M=\coef M\sum_{S\subseteq V,\,|S|\le r}\DS^{(V)}_S\SG^{(V)}$.
    \item If $M$ is disconnected, then $B^{(r)}_M=0$.
  \end{enumerate}
\end{lemma}

\begin{proof}
  The coefficient $\coef M$ depends only on the restriction of a function to the subspace
  $\{J_e=0,\ e\notin M\}$, and we compute there.
  (i) Take $w\in S\setminus V$. On the subspace every edge weight at $w$ vanishes, so $\SG$ does not
  depend on $z_w$ (the fields of the other vertices give $z_w$ the weight $0$, and
  $\Phi(F_w)=\Phi(0)=0$), and in the weight \eqref{eq:a:blockweight} the coordinate $z_w$ is not
  coupled to the others. Hence $\Pu_{U\cup\{w\}}\SG=\Pu_U\SG$ for every $U\subseteq S\setminus\{w\}$, and
  $\DS_S=\sum_{U\subseteq S\setminus\{w\}}(-1)^{|S\setminus\{w\}|-|U|}(\Pu_{U\cup\{w\}}-\Pu_U)$ annihilates $\SG$.
  (ii) On the subspace the vertices outside $V$ are isolated: their terms $\Phi$ vanish, and the
  fields and posteriors of the vertices of $V$ involve only edges inside $V$.
  (iii) follows from (i) and (ii).
  (iv) Let $M$ split into nonempty vertex-disjoint parts $M_1,M_2$. On the subspace there is no edge
  between $V(M_1)$ and $V(M_2)$. For $i\in V(M_1)$, $\Phi(F_i)$ contains no weight of $M_2$, and the
  posterior of $Z_U$ factorizes between the two parts, so $\Pu_U\Phi(F_i)=\Pu_{U\cap V(M_1)}\Phi(F_i)$
  contains no weight of $M_2$; the case $i\in V(M_2)$ is symmetric, and vertices outside $V(M)$
  contribute nothing. Hence no term contains the monomial $J^M$.
\end{proof}

Let $\varkappa_M$ denote the joint cumulant of $\{W_e\}_{e\in M}$, each $W_e$ repeated $m_e$ times, under
$Z\sim\theta^{\otimes n}$; equivalently $\log\Exp_\theta\exp(\sum_eJ_eW_e)=\sum_M\varkappa_MJ^M/M!$.

\begin{lemma}[Exact Identity]
  \label{lem:identity}
  Let the path be real analytic. For every motif $M$, $V:=V(M)$, every $r\ge0$ and every $z$,
  \begin{equation}
    \label{eq:a:identity}
    B^{(r)}_M(z)=\frac{\partial_u\varkappa_M}{M!}
    -\sum_{S\subseteq V,\ |S|>r}\coef M\big(\DS^{(V)}_S\SG^{(V)}\big)(z).
  \end{equation}
  Moreover $\varkappa_M$ is a polynomial in $P_2,\dots,P_{v(M)}$ whose coefficients depend only on $M$.
\end{lemma}

\begin{proof}
  On the motif system, \Cref{lem:sweep}(ii) gives $\sum_{S\subseteq V}\DS^{(V)}_S=\Pu^{(V)}_V$, so, by
  \Cref{lem:local}(iii), $B^{(r)}_M=\coef M\Pu^{(V)}_V\SG^{(V)}-\sum_{|S|>r}\coef M\DS^{(V)}_S\SG^{(V)}$. The
  operator $\Pu^{(V)}_V$ averages all labels of $V$ with the weight $w_V$. As a function of $z_i$,
  $w_V(z)$ is $\theta_{z_i}e^{F_{i,z_i}}$ times a factor free of $z_i$, and $F_i$ does not depend on $z_i$,
  so $\sum_{z_i}w_V\,\ell_{z_i}=\Phi(F_i)\sum_{z_i}w_V$. Therefore
  \[
    \Pu^{(V)}_V\SG^{(V)}=\frac{\sum_zw_V(z)\sum_{i\in V}\ell_{z_i}}{\sum_zw_V(z)}
    =\partial_u\log\sum_zw_V(z)=\partial_u\log\Exp_\theta\exp\Big(\sum_{\{k,l\}\subseteq V}J_{kl}W_{kl}\Big),
  \]
  using $\partial_uw_V=w_V\sum_{i\in V}\ell_{z_i}$. The function $\log\sum_zw_V(z)$ is jointly analytic in
  $(u,J)$ for $J$ near $0$, because the path is real analytic and $\sum_zw_V=1$ at $J=0$; so $\coef M$
  commutes with $\partial_u$, and taking $\coef M$ gives $\partial_u\varkappa_M/M!$. Finally, a joint
  cumulant is a polynomial in the joint moments of its sub-multisets, with coefficients depending only on
  the multiset. Moreover, $\Exp\prod_{e\in M'}W_e=\prod_QP_{v(Q)}$, the product running over the connected
  components $Q$ of $M'$: the vertices of a component share a meaning, and different components involve
  disjoint sets of independent labels. Each component has between $2$ and $v(M)$ vertices.
\end{proof}

\subsection{Local Decomposition and Small Fields}

For $i\in[n]$ and $S\subseteq[n]$ with $|S|\le r$, let $\Omega:=\{i\}\cup S$, let $E_\Omega$ be the set of pairs
inside $\Omega$, and define the external fields $F^{\mathrm{ext}}_v:=\sum_{m\notin \Omega}J_{vm}\mathbf e_{z_m}\in\Real^K$
for $v\in \Omega$. By linearity,
\begin{equation}
  \label{eq:a:local}
  \Sw r=\sum_{i=1}^n\ \sum_{|S|\le r}\Psi_{i,S},\qquad
  \Psi_{i,S}:=\DS_S\Phi(F_i)=g_{i,S}(x,f;z_\Omega),
\end{equation}
with $x:=J_{E_\Omega}$ and $f:=(F^{\mathrm{ext}}_v)_{v\in \Omega}$, where $g_{i,S}$ is explicit: for $U\subseteq S$ the
conditional weight of $Z_U$ is proportional to
\[
  \Big(\prod_{k\in U}\theta_{z'_k}\Big)\exp\Big[\sum_{k\in U}F^{\mathrm{ext}}_{k,z'_k}
  +\sum_{k\in U,\ l\in \Omega\setminus U}J_{kl}\Ind{z'_k=z_l}+\sum_{k<l,\ k,l\in U}J_{kl}\Ind{z'_k=z'_l}\Big],
\]
and $F_i=F^{\mathrm{ext}}_i+\sum_{k\in \Omega\setminus\{i\}}J_{ik}\mathbf e_{z_k}$. The variables $x$ and $f$ are separate
arguments of $g_{i,S}$, and $f$ is linear in $J$.

\begin{lemma}[Connected Decomposition]
  \label{lem:connected}
  If the graph $(\Omega,\mathrm{supp}\,x)$ is disconnected, then $g_{i,S}(x,f)=0$. Consequently
  \begin{gather*}
    g_{i,S}(x,f)=\sum_{B\subseteq E_\Omega:\ (\Omega,B)\text{ connected}}x^B\,h_B(x_B,f),\\
    h_B(x_B,f):=\int_{[0,1]^B}\partial_{x_B}g_{i,S}\big(t\odot x_B,0_{E_\Omega\setminus B},f\big)\,\mathrm dt .
  \end{gather*}
\end{lemma}

\begin{proof}
  Let $\Omega$ split into nonempty parts $\Omega_1\ni i$ and $\Omega_2$ with no edge of $\mathrm{supp}\,x$ between them.
  Then $\Phi(F_i)$ does not depend on $z_{\Omega_2}$, and the conditional weight of $Z_U$ factorizes between
  $\Omega_1$ and $\Omega_2$: the local weights between the parts vanish, and the external fields are fixed
  vectors that do not couple labels. Hence $\Pu_U\Phi(F_i)=\Pu_{U\cap \Omega_1}\Phi(F_i)$ and
  \[
    \DS_S\Phi(F_i)=\sum_{U_1\subseteq S\cap \Omega_1}(-1)^{|S\cap \Omega_1|-|U_1|}\Pu_{U_1}\Phi(F_i)
    \sum_{U_2\subseteq S\cap \Omega_2}(-1)^{|S\cap \Omega_2|-|U_2|}=0,
  \]
  since $S\cap \Omega_2=\Omega_2\ne\varnothing$. For the decomposition, Möbius inversion in the coordinates of $x$
  gives $g_{i,S}=\sum_{B\subseteq E_\Omega}g_B(x_B)$ with $g_B:=\sum_{B'\subseteq B}(-1)^{|B\setminus B'|}g_{i,S}(x_{B'},0,f)$. If
  $(\Omega,B)$ is disconnected, so is $(\Omega,B')$ for every $B'\subseteq B$, and $g_B=0$. If it is connected, $g_B$
  vanishes whenever $x_e=0$ for some $e\in B$, and the fundamental theorem of calculus in each coordinate
  of $B$ gives $g_B=x^Bh_B$, because $\partial_{x_B}g_B=\partial_{x_B}g_{i,S}(\cdot,0,f)$.
\end{proof}

\begin{lemma}[Derivative Bounds]
  \label{lem:derivatives}
  For every multi-index $\alpha$ and real $(x,f)$,
  \begin{equation}
    \label{eq:a:derivatives}
    \big|\partial^\alpha g_{i,S}(x,f)\big|\ \le\ C_{|\alpha|,r,K}\,\|\dot\theta\|_1\,
    \exp\Big(2\big(\|F^{\mathrm{ext}}_i\|_\infty+\|x\|_1\big)\Big),
  \end{equation}
  and each $h_B$ obeys the same bound with $C_{|\alpha|+|B|,r,K}$ in place of $C_{|\alpha|,r,K}$. Moreover
  $|\Psi_{i,S}|\le2^{|S|}\max_c|\ell_c|$ for real $J$.
\end{lemma}

\begin{proof}
  $g_{i,S}$ is a signed sum of at most $2^r$ functions $\Pu_U\Phi(F_i)$, $U\subseteq S$. Each is a sum over
  $z'_U$ of normalized weights times $\Phi$ evaluated at $F^{\mathrm{ext}}_i+\sum_{k\in \Omega\setminus\{i\}}x_{ik}\mathbf e_{z''_k}$,
  where $z''_k$ is $z'_k$ or $z_k$. The normalized weights are softmax functions of logits that are
  affine in $(x,f)$ with coefficients in $\{0,1\}$, so each derivative of order $m$ of a weight is
  bounded by a constant depending only on $m$. By \Cref{lem:phi} the derivatives of $\Phi$ of order $m$
  are at most $C_m\|\dot\theta\|_1e^{2g}$ with $g\le\|F^{\mathrm{ext}}_i\|_\infty+\|x\|_1$. The Leibniz and chain rules
  give \eqref{eq:a:derivatives}; differentiating under the integral, whose integrand is smooth on the
  compact cube $[0,1]^B$ and whose arguments satisfy the same bound on $g$, gives the bound for $h_B$.
  Finally, for real $J$ each $\Pu_U\Phi(F_i)$ is an average of values of $\Phi$, each at most
  $\max_c|\ell_c|$ in modulus (\Cref{lem:score}).
\end{proof}

\begin{lemma}[Small Fields]
  \label{lem:small-field}
  Let $\mathcal G:=\{\max_{v,c}|F_{vc}|\le\sqrt\delta\log n\}$. There is $n_0$ depending only on $q_0$ such
  that $\delta\le1$ and $n\ge n_0$ imply $\Prob(\mathcal G^c)\le2Kn\,e^{-q_0^2\log^2n/50}$.
\end{lemma}

\begin{proof}
  Given $Z$, $F_{vc}$ is a sum of at most $n$ conditionally independent terms $J_{vm}$. By
  \Cref{lem:weights}, after centering each is bounded by $4\Delta/q_0$, the conditional variances sum to at
  most $3\delta/q_0^2$, and the conditional mean has modulus at most $2\delta/q_0^2$, which is at most
  $\frac12\sqrt\delta\log n$ once $\log n\ge4/q_0^2$. Bernstein's inequality with $t=\frac12\sqrt\delta\log n$,
  variance proxy $3\delta/q_0^2$ and range $4\Delta/q_0$, together with $\Delta/\sqrt\delta=n^{-1/2}$, gives
  $\Prob(|F_{vc}|>\sqrt\delta\log n\mid Z)\le2\exp(-q_0^2\log^2n/50)$ for $n\ge n_0(q_0)$. A union bound
  over the $nK$ pairs $(v,c)$ gives the claim.
\end{proof}

\subsection{The Entropy in Symmetric Coordinates}

\begin{lemma}[Integral Representation]
  \label{lem:intrep}
  For $\theta,\theta'\in\Simplex$, zero coordinates allowed,
  \begin{equation*}
    \Ent(\theta)-\Ent(\theta')=\int_0^\infty\log\frac{\pbar_\theta(t)}{\pbar_{\theta'}(t)}\,\mathrm dt,
  \end{equation*}
  the integrand being absolutely integrable on $(0,\infty)$.
\end{lemma}

\begin{proof}
  For $x\in[0,1]$ and $M>0$, $\int_0^M\log(t+x)\,\mathrm dt=(M+x)\log(M+x)-M-x\log x$. Summing over $c$ and
  subtracting the same identity for $\theta'$,
  \[
    \int_0^M\log\frac{\pbar_\theta(t)}{\pbar_{\theta'}(t)}\,\mathrm dt=\Ent(\theta)-\Ent(\theta')
    +\sum_c\big[\phi_M(\theta_c)-\phi_M(\theta'_c)\big],\qquad \phi_M(x):=(M+x)\log(M+x).
  \]
  By Taylor's theorem $\phi_M(x)=M\log M+x(\log M+1)+x^2/(2M)+O(M^{-2})$ with an absolute constant, uniformly in
  $x\in[0,1]$, and $\sum_c\theta_c=\sum_c\theta'_c=1$, so the bracket sum is $(P_2(\theta)-P_2(\theta'))/(2M)+O(M^{-2})\to0$, with a
  constant depending only on $K$. For
  integrability: on $t\ge1$, $\pbar_\theta(t)-\pbar_{\theta'}(t)=\sum_{j\ge2}(e_j(\theta)-e_j(\theta'))t^{K-j}$ has
  modulus at most $D_e\,t^{K-2}$ with $D_e=D_e(\theta,\theta')=\sum_{j\ge2}|e_j(\theta)-e_j(\theta')|$, while
  $\pbar\ge t^K+t^{K-1}$ because $e_0=e_1=1$ and $e_j\ge0$; with $|\log(x/y)|\le\log(1+|x-y|/\min\{x,y\})$ the
  integrand is at most $D_e/(t(t+1))$. On $(0,1]$, $|\log\pbar_\theta(t)|\le K|\log t|+K\log2$, which is
  integrable, and the same holds for $\theta'$.
\end{proof}

\section{Proofs of the Fiber Lemmas}
\label{app:fibers}

This appendix proves the fiber lemmas, with the constructions and the notation of
\Cref{sec:fibers}.

\subsection{The Mixed Fiber}

\begin{proof}[Proof of \Cref{lem:fiber-mixed}]
  (i) \emph{Small roots.} Let $Y(\eps,w):=\prod_c(1-\eps w/y_c)$. Since
  $\prod_c(\eps w-y_c)=(-1)^m(\prod_cy_c)Y$ and $\prod_j(w-x_j)=2^{1-2s}T_s(2w-1)$ (both sides are monic of
  degree $s$ with the same roots),
  \[
    p_u(\eps w)=\eps^s2^{1-2s}\Big(\prod_cy_c\Big)\,G(u,\eps,w),\qquad G(u,\eps,w):=(-1)^mY(\eps,w)\,T_s(2w-1)+(u-\tfrac12).
  \]
  $G$ is real analytic in $(u,\eps,w)$. At $\eps=0$, $G=(-1)^mT_s(2w-1)+u-\frac12$, whose zeros solve
  $T_s(2w-1)=c_u:=(-1)^{m+1}(u-\frac12)$ with $|c_u|\le1/4$. Writing $2w-1=\cos\varphi$ with $\varphi\in(0,\pi)$, the
  equation $\cos(s\varphi)=c_u$ has exactly one solution in each interval $((i-1)\pi/s,i\pi/s)$, $i=1,\dots,s$, and
  at these points $|\partial_wT_s(2w-1)|=2s|\sin(s\varphi)|/\sin\varphi\ge2s\sqrt{15}/4$. Since $|c_u|\le1/4<1$, the
  solutions stay a distance depending only on $s$ away from $\varphi\in\{0,\pi\}$, that is, from $w\in\{0,1\}$,
  uniformly in $u\in[1/4,3/4]$. The analytic implicit function theorem and compactness give
  $\eps_1>0$ and real-analytic $w_j(u,\eps)$ on a neighborhood of $[1/4,3/4]\times[0,\eps_1]$ that are simple,
  mutually separated zeros of $G(u,\eps,\cdot)$, with $|\partial_wG(u,\eps,w_j)|\ge g_0>0$; shrinking $\eps_1$
  gives $w_{\min}\le w_j\le2$. Hence the $\eps w_j$ are $s$ positive real roots of $p_u$. At $u=\frac12$,
  $p_{1/2}$ has the roots $\eps x_j$, and by uniqueness in the implicit function theorem these are the
  $\eps w_j(\frac12,\eps)$.

  \emph{Large roots.} Let $\Psi_c(z,\eps,\zeta):=\prod_{c'}(z-y_{c'}(\eps))\prod_j(z-\eps x_j)+\zeta$. At
  $(\bar y_c,0,0)$, $\Psi_c=0$ and $\partial_z\Psi_c=\prod_{c'\ne c}(\bar y_c-\bar y_{c'})\,\bar y_c^s\ne0$. The analytic
  implicit function theorem gives a root $z_c(\eps,\zeta)$ near $\bar y_c$ with $z_c(\eps,0)=y_c(\eps)$ and
  $|\partial_\zeta z_c|\le C$; then $\theta_c(u):=z_c(\eps,(u-\frac12)\sigma_\eps)$ is real analytic in $u$ and
  $|\theta_c-y_c|\le C\sigma_\eps\le C\eps^s$. For small $\eps$, $\theta_c>\tau$ because $\bar y_c>\tau$, and the large
  roots are distinct from each other and from the small roots.

  These $m+s=K$ distinct real roots are all the roots of $p_u$; they are simple and positive. Only the
  constant term of $p_u$ depends on $u$, so their sum is $e_1=\sum_cy_c+\eps\sum_jx_j=(1-\eps s/2)+\eps s/2=1$.
  Hence $\theta(u)\in\Simplex$, with $m=K-s$ coordinates above $\tau$.

  (ii) $p_u(z)=\sum_j(-1)^je_j(\theta(u))z^{K-j}$ and only the constant term $(-1)^Ke_K$ depends on $u$, with
  increment $(u-\frac12)\sigma_\eps$. By \eqref{eq:d:pk},
  $\dot P_K=(-1)^{K-1}K\dot e_K=(-1)^{K-1}K(-1)^K\sigma_\eps=-K\sigma_\eps$.

  (iii) Differentiating $p_u(\theta_c(u))=0$ gives $p_u'(\theta_c)\dot\theta_c+\sigma_\eps=0$. At a large root,
  $|p_u'(\theta_c)|$ is bounded below by a positive constant, so $|\dot\theta_c|\le C\sigma_\eps\le C\eps^s$. At a small
  root, differentiating $p_u(\eps w)=\eps^s2^{1-2s}(\prod y)G$ in $w$ gives
  $p_u'(\eps w_j)=\eps^{s-1}2^{1-2s}(\prod y)\,\partial_wG(u,\eps,w_j)$, so
  $|\dot\theta_c|=\sigma_\eps/|p_u'|=\eps/|\partial_wG|\le\eps/g_0$ and $|\dot\theta_c|/\theta_c\le1/(g_0w_{\min})$. Finally,
  $\dot P_k=k\sum_c\theta_c^{k-1}\dot\theta_c$: the large roots contribute at most $kmC\eps^s$ and the small ones at
  most $ks(2\eps)^{k-1}\eps/g_0\le ks2^{k-1}\eps^s/g_0$, since $k\ge K>s$ and $\eps\le1$. Conditions (H0)--(H4)
  follow, the interval $[1/4,3/4]$ having length $1/2$.

  (iv) By \Cref{lem:intrep}, and because $\pbar_{\theta(u)}(t)=\sum_je_j(\theta(u))t^{K-j}$ changes only through
  $e_K(\frac34)-e_K(\frac14)=(-1)^K\sigma_\eps/2$, the substitution $t=\eps v$ gives
  \[
    [\Ent]^{3/4}_{1/4}=\eps\int_0^\infty f_\eps(v)\,\mathrm dv,\qquad
    f_\eps(v):=\log\frac{\pbar_{\theta(3/4)}(\eps v)}{\pbar_{\theta(1/4)}(\eps v)},
  \]
  with $\pbar_{\theta(u)}(\eps v)=\prod_{c\le m}(\eps v+\theta_c(u))\cdot\eps^s\prod_j(v+w_j(u,\eps))$. Pointwise, the products
  over the large roots tend to $\prod_c\bar y_c$ for both values of $u$, so
  $f_\eps(v)\to f_0(v):=\log\big(\prod_j(v+w_j(\frac34,0))/\prod_j(v+w_j(\frac14,0))\big)$. For domination,
  $\min\{\pbar_{\theta(3/4)},\pbar_{\theta(1/4)}\}(\eps v)\ge\tau^m\eps^s(v+w_{\min})^s$, so
  $|f_\eps(v)|\le\frac{\sigma_\eps/2}{\tau^m\eps^s(v+w_{\min})^s}\le\frac C{(v+w_{\min})^s}$, which is integrable on $[0,\infty)$ for
  $s\ge2$. By dominated convergence, $[\Ent]^{3/4}_{1/4}/\eps\to H_s^*:=\int_0^\infty f_0(v)\,\mathrm dv$. The monic polynomial
  $Q_u(w):=\prod_j(w-w_j(u,0))$ equals $2^{1-2s}[T_s(2w-1)+(-1)^m(u-\frac12)]$ (both are monic of degree $s$ with the
  same $s$ simple roots), so $\prod_j(v+w_j(u,0))=(-1)^sQ_u(-v)$ and
  \[
    \prod_j\big(v+w_j(\tfrac34,0)\big)-\prod_j\big(v+w_j(\tfrac14,0)\big)=(-1)^s2^{1-2s}(-1)^m\cdot\tfrac12=(-1)^K2^{-2s}.
  \]
  Both products are positive for $v\ge0$, so $f_0(v)=\log\big(1+(-1)^K2^{-2s}/\prod_j(v+w_j(\frac14,0))\big)$ is
  nonzero with the constant sign $(-1)^K$, and $H_s^*\ne0$. The $w_j(u,0)$ depend only on $(s,m,u)$, so
  $h_s:=|H_s^*|$ depends only on $(K,s)$; choose $\eps_H$ with $|[\Ent]^{3/4}_{1/4}/\eps-H_s^*|\le h_s/2$ for $\eps\le\eps_H$.
\end{proof}

\subsection{The Fiber for \texorpdfstring{$s=1$}{s = 1}}

\begin{proof}[Proof of \Cref{lem:fiber-one}]
  (i) $p_u(z)=(z-\omega)G(z;\eps,\omega)$, where $G(z;\eps,\omega):=(F_\eps(z)-F_\eps(\omega))/(z-\omega)$ is a monic polynomial of
  degree $K-1$ in $z$ whose coefficients are polynomials in $(\eps,\omega)$. At $(\eps,\omega)=(0,0)$,
  $F_0(z)=z\prod_j(z-b_j)$ and $G(z;0,0)=\prod_j(z-b_j)$, whose roots $b_j$ are pairwise distinct with
  $|\partial_zG|=\prod_{j'\ne j}|b_j-b_{j'}|>0$. The implicit function theorem gives roots $\theta_j(\eps,\omega)$, analytic
  near $(0,0)$, with $|\theta_j-b_j|\le C(\eps+|\omega|)\le3C\eps$ for $\omega\in[\eps,2\eps]$, and
  $\theta_j(u):=\theta_j(\eps,\eps(1+u))$ is analytic in $u$. For small $\eps$ the $\theta_j$ exceed $\tau$ (as $b_j>\tau$), are
  pairwise distinct, and differ from $\omega(u)\le2\eps<\tau$. Hence $p_u$ has $K$ distinct positive real roots. The
  coefficient of $z^{K-1}$ in $F_\eps$ is $-(\eps+\sum_jy_j)=-1$, and subtracting the constant $F_\eps(\omega)$ does not
  change it, so the roots sum to $1$.

  (ii) All coefficients of $p_u$ except the constant term are those of $F_\eps$, which do not depend on
  $u$; the constant term is $F_\eps(0)-F_\eps(\omega(u))=(-1)^Ke_K(u)$. Hence $\dot e_K=(-1)^{K+1}\eps F_\eps'(\omega)$ and, by
  \eqref{eq:d:pk}, $\dot P_K=K\eps F_\eps'(\omega)$.

  (iii) Differentiating $p_u(\theta_j(u))=0$ and using $p_u'=F_\eps'$ gives
  $F_\eps'(\theta_j)\dot\theta_j-F_\eps'(\omega)\dot\omega=0$, so $\dot\theta_j=\eps F_\eps'(\omega)/F_\eps'(\theta_j)$. As $\eps\to0$,
  $F_\eps'(\theta_j)\to F_0'(b_j)=b_j\prod_{j'\ne j}(b_j-b_{j'})\ne0$, and $|F_\eps'(\omega)|\le C$; so $|\dot\theta_j|\le C\eps$.

  (iv) With $h(x):=-x\log x$, $[\Ent]^1_0=(h(2\eps)-h(\eps))+\sum_j(h(\theta_j(1))-h(\theta_j(0)))$. The first difference is
  $-2\eps\log(2\eps)+\eps\log\eps=\eps\log(1/\eps)-2\eps\log2$. In the second, $|\theta_j(1)-\theta_j(0)|\le C\eps$ and
  $|h'(x)|=|1+\log x|\le1+\log(1/\tau)$ on $[\tau,1]$, so its modulus is at most $C'\eps$. Hence
  $[\Ent]^1_0\ge\eps\log(1/\eps)-(2\log2+C')\eps\ge\frac12\eps\log(1/\eps)$ once $\log(1/\eps)\ge2(2\log2+C')$.
\end{proof}

\subsection{The Perturbation for \texorpdfstring{$s=0$}{s = 0}}

\begin{proof}[Proof of \Cref{lem:fiber-zero}]
  (i) The coordinates of $\bar y$ are distinct, so the roots of $p_{\bar y}$ are simple, and the implicit
  function theorem applies to $(z,\zeta)\mapsto p_{\bar y}(z)+\zeta$ at each $(\bar y_c,0)$. Only the constant term
  changes, so the roots still sum to $1$; for small $\eta$ they stay above $\tau$ because $\bar y_c>\tau$.
  (ii) The constant term $(-1)^Ke_K$ increases by $(-1)^K\eta u$; use \eqref{eq:d:pk}.
  (iii) $\dot\theta_c=-(-1)^K\eta/p_u'(\theta_c)$, and $|p_u'(\theta_c)|$ is bounded below; then
  $|\dot P_k|\le k\|\dot\theta\|_1$.
  (iv) $\pbar_{\theta(u)}(t)=\pbar_{\bar y}(t)+\eta u$, so, by \Cref{lem:intrep},
  $[\Ent]^1_0=\int_0^\infty\log(1+\eta/\pbar_{\bar y}(t))\,\mathrm dt$. Since $\pbar_{\bar y}(t)\ge\prod_c\bar y_c\ge\eta$ and
  $\log(1+x)\ge x-x^2/2\ge x/2$ for $x\in[0,1]$, $[\Ent]^1_0\ge\frac\eta2\Lam(\bar y)$. Finally,
  $\pbar_{\bar y}(t)\ge\max\{\prod_c\bar y_c,t^K\}$ gives $\Lam(\bar y)<\infty$.
\end{proof}

\end{document}